\documentclass[11pt]{article}
\usepackage[utf8]{inputenc}
\usepackage{amsmath}
\usepackage{amsthm}
\usepackage{amssymb}
\usepackage{mathtools}
\usepackage{fullpage}

\usepackage{listings}
\usepackage{graphicx}
\usepackage{multirow}
\usepackage{amssymb}
\usepackage{verbatim}
\usepackage{enumitem}
\usepackage{color}
\usepackage{comment}

\newtheorem{theorem}{Theorem}[section]
\newtheorem{definition}{Definition}
\newtheorem{lemma}[theorem]{Lemma}

\newtheorem{proposition}[theorem]{Proposition}
\newtheorem{claim}[theorem]{Claim}

\def\nat{{\mathbb N}}

\def\current{{\mathtt{current}}}

\begin{document}
\normalsize

\title{Qualitative Multi-Objective Reachability for 
Ordered Branching MDPs}
\author{Kousha Etessami\thanks{\tt kousha@inf.ed.ac.uk}
\\U. of Edinburgh
\and 
Emanuel Martinov\thanks{\tt eo.martinov@gmail.com}\\U. of Edinburgh 
}

\date{}
\maketitle

\begin{abstract}
We study qualitative multi-objective reachability problems 
for Ordered Branching Markov Decision Processes (OBMDPs), 
or equivalently context-free MDPs, 
building on prior results for single-target reachability 
on Branching Markov Decision Processes (BMDPs). 

We provide two separate algorithms for ``almost-sure'' and 
``limit-sure'' multi-target reachability for OBMDPs. 
Specifically, given an OBMDP, $\mathcal{A}$, 
given a starting non-terminal, and given a set of 
\textit{target} non-terminals $K$ of size $k = |K|$, 
our first algorithm decides whether 
the supremum probability, of generating a tree 
that contains every target non-terminal in set $K$, is $1$. 
Our second algorithm decides whether there is a strategy 
for the player to almost-surely (with probability $1$) 
generate a tree that contains every target non-terminal 
in set $K$.

The two separate algorithms are needed: we show 
that indeed, in this context, ``almost-sure'' $\not=$ 
``limit-sure'' for multi-target reachability, meaning 
that there are OBMDPs for which the player may not 
have any strategy to achieve probability exactly $1$ of 
reaching all targets in set $K$ in the same generated tree, 
but may have a sequence of strategies that achieve 
probability arbitrarily close to $1$. 
Both algorithms run in time $2^{O(k)} \cdot |\mathcal{A}|^{O(1)}$, 
where $|\mathcal{A}|$ is the total bit encoding length 
of the given OBMDP, $\mathcal{A}$. 
Hence they run in polynomial time when $k$ is fixed, and are 
fixed-parameter tractable with respect to $k$. 
Moreover, we show that even the qualitative almost-sure 
(and limit-sure)
multi-target reachability decision problem is in general NP-hard, 
when the size $k$ of the set $K$ of target non-terminals 
is not fixed.

\end{abstract}

\section{Introduction}

Ordered Branching Markov Decision Processes (OBMDPs) can
be viewed as controlled/probabilistic context-free
grammars, but without any terminal symbols, and
where moreover the non-terminals are partitioned
into two sets: controlled non-terminals and 
probabilistic non-terminals.
Each non-terminal, $N$, has an associated set of 
grammar rules of the form $N \rightarrow \gamma$,
where $\gamma$ is a (possibly empty) 
sequence of non-terminals.
Each probabilistic non-terminal is equipped
with a given probability distribution on its associated grammar rules.
For each controlled non-terminal,
$M$, there is an associated non-empty set of available actions,
$A_M$, which is in one-to-one correspondence with
the grammar rules of $M$. So, for each action, $a \in A_M$, there is an
 associated grammar rule $M \stackrel{a}{\rightarrow} \gamma$.
Given an OBMDP, given a ``start'' non-terminal,
and given a ``strategy'' for the controller, these together determine
a probabilistic process that generates a (possibly infinite)
random ordered tree.
The tree is formed via the usual parse tree
expansion of grammar rules, proceeding generation by generation,
in a top-down manner. Starting with a root node 
labeled by the ``start'' non-terminal, the ordered tree is generated
based on the controller's (possibly randomized) 
choice of action at each node of the tree that is labeled by a
controlled non-terminal, 
and based on the probabilistic choice of a grammar rule
at nodes that are labeled by a probabilistic non-terminal.

We assume that a general {\em strategy} for 
the controller can operate
as follows: at each node $v$ of the ordered tree, 
labeled by a controlled non-terminal,
the controller (player) can choose its action (or its probability distribution on actions) at $v$ based on the entire
``ancestor history'' of $v$, meaning based on the entire sequence of
labeled nodes and actions leading from the root node to $v$, 
{\em as well as} based on the ordered position of each
of its ancestors (including $v$ itself) among its siblings in the tree.

Ordered Branching Processes (OBPs) are OBMDPs without
any controlled non-terminals.  Both OBPs and OBMDPs are 
very similar to classic multi-type branching processes (BPs), 
and to Branching MDP (BMDPs), respectively.
The only difference is that for OB(MD)Ps the generated tree is 
{\em ordered}. 
In particular,
the rules for an OBMDP have an ordered {\em sequence} of non-terminals on their right hand side,
whereas there is no
such ordering in BPs or BMDPs: each rule for a given type associates
an unordered multi-set of ``offsprings'' of various types to that given type.
Branching processes and stochastic context-free grammars
have well-known applications in many fields, including
in natural language processing,
biology/bioinformatics
(e.g., \cite{KA02}, population genetics \cite{HJV05},
RNA modeling \cite{DEKM98},
and cancer tumor growth modelling \cite{Bozic13,RBCN13}), and physics
(e.g., nuclear chain reactions).
Generalizing these models to MDPs is natural, and
can allow us to study, and to optimize algorithmically, settings
where such random processes can partially be controlled.

The single-target reachability objective for OBMDPs amounts to
optimizing (maximizing or minimizing) the probability
that, starting at a given start (root) non-terminal,
the generated tree contains some given target non-terminal.
This objective has already been thoroughly studied for
BMDPs, as well as for (concurrent) stochastic game generalizations of BMDPs
(\cite{ESY-icalp15-IC,EMSY-icalp19-bcsg}).
Moreover, it turns out that there is really no difference at all
between BMDPs and OBMDPs when it comes to the single-target
reachability objective: all the algorithmic results from
\cite{ESY-icalp15-IC,EMSY-icalp19-bcsg}
carry over, mutatis
mutantis, for OBMDPs, and for their stochastic game generalizations.

A natural generalization of single-target reachability is 
multi-objective reachability, 
where the goal is to optimize each of the respective probabilities 
that the generated tree contains each of several 
different target non-terminals. 
Of course, there may be a trade-offs between 
these different objectives.

Our main concern in this paper is  \textit{qualitative}
multi-objective reachability problems,
where the aim is to determine whether there is a
strategy that guarantees that each of the given set of target non-terminals
is almost-surely (respectively, limit-surely) contained in the generated tree, i.e., with probability 1 (respectively, with probability arbitrarily close to $1$).
In fact, we show that the
\textit{almost-sure} and \textit{limit-sure} problems do
not coincide. That is, there are OBMDPs for which there is no single
strategy that achieves probability exactly $1$ for
reaching all targets, but where nevertheless,
for every $\epsilon > 0$, there is a strategy that
guarantees a probability $\ge 1 - \epsilon$, 
of reaching all targets.

By contrast, for both BMDPs and OBMDPs,
for single-target reachability, the {\em qualitative} 
almost-sure and limit-sure questions do coincide: 
there is a strategy that guarantees reaching 
the target non-terminal with probability $1$ if and only if 
there is a sequence of strategies that guarantee 
reaching the target with probabilities arbitrarily close to $1$ 
(\cite{ESY-icalp15-IC}).\footnote{The notion of general 
``strategy'' employed for BMDPs 
in \cite{ESY-icalp15-IC} is somewhat different than what 
we define in this paper for OBMDPs: 
it allows the controller to not only base its choice at a tree node on
 the ancestor chain of that node, but on the entire tree up to that ``generation''. This is needed for BMDPs because
there is no ordering available on ``siblings'' in the tree generated by a BMDP.    However, a careful look shows that 
the results of \cite{ESY-icalp15-IC} imply that, for OBMDPs,
for single-target reachability, almost-sure and limit-sure reachability also coincide under the notion of ``strategy''
we have defined in this paper, where choices are based only on the  ``ancestor history'' (with ordering information) of each node in the
{\em ordered} tree.  In particular the key ``queen/workers" strategy employed for almost-sure (=limit-sure) reachability in
\cite{ESY-icalp15-IC} can be mimicked using the ordering with respect to siblings that is available in ancestor histories of OBMDPs.  A natural question is what happens for multi-objective qualitative reachability in OBMDPs, if we allow the more general definition of strategy, which can depend at each node on the entire tree up to the ``generation" of that node (even on nodes that are not among its ancestors).  We leave this question open in this paper, but we conjecture that under that 
richer notion of strategy   ``almost-sure" = ``limit-sure" for multi-target reachability for (O)BMDPs, and that essentially the same algorithm that we provide for limit-sure multi-target reachability for OBMDPs under the weaker notion of strategy used in this paper works also to decide both limit-sure and almost-sure multi-target reachability under that richer notion of strategy for (O)BMDPs.}

We give two separate algorithms for  almost-sure and limit-sure
multi-objective reachability. For the
\textit{almost-sure} problem, we are 
given an OBMDP, a start non-terminal, 
and a set of target non-terminals, 
and we must decide whether there 
exists a strategy using which the process 
generates, with probability 1, a tree 
that contains all the given target non-terminals. 
If the answer is ``yes'', the 
algorithm can also construct a (randomized) witness strategy 
that achieves this.\footnote{This strategy is, 
however, necessarily not ``static'',
meaning it must actually use 
the ancestor history: the action distribution cannot be 
defined solely based on which 
non-terminal is being expanded.}
The algorithm for the \textit{limit-sure} problem 
decides whether the supremum probability of generating 
a tree that contains all given target non-terminals is $1$. 
If the answer is ``yes'', the algorithm can also construct,
given any $\epsilon > 0$, a randomized non-static
strategy that guarantees probability $\ge 1 - \epsilon$.
The limit-sure algorithm is only slightly more involved.

Both algorithms run in time $2^{O(k)} \cdot |\mathcal{A}|^{O(1)}$,
where $|\mathcal{A}|$ is the total bit encoding length
of the given OBMDP, $\mathcal{A}$, and $k = |K|$ is the size of 
the given set $K$ of target non-terminals. 
Hence they run in polynomial time when $k$ is fixed, and are
fixed-parameter tractable with respect to $k$.
Moreover, we show that the qualitative almost-sure (and limit-sure)
multi-target reachability decision problem is in general NP-hard,
when $k$ is not fixed.

Going beyond the goal of assuring probability $1$ of 
reaching each of a set of target non-terminals, 
we also consider more general qualitative multi-objective 
reachability/non-reachability problems, 
where we are given a set of target non-terminals, $K$, 
and where for each non-terminal $M \in K$, 
we are also given a $0$/$1$ probability $b_M \in \{0,1\}$, 
and an inequality $\Delta_M \in \{ = , < , > \}$, 
and where we wish to decide whether the controller 
has a single strategy using which, for all $M \in K$ 
the probability that the generated tree contains 
the non-terminal $M$ is $\Delta_M b_M$. 
We show that in some special cases these problems are decidable 
(efficiently). However, we leave open the decidability of 
the most general case of \textit{arbitrary} boolean combinations 
of such qualitative reachability and non-reachability queries 
over different target non-terminals. Furthermore, we leave open all
(both decision and approximation) {\em quantitative} 
multi-objective reachability questions, 
including when the goal is to approximate 
the tradeoff {\em pareto curve} of optimal probabilities 
for different reachability objectives.   These are intriguing questions for future research.

\medskip
\noindent \textbf{Related work.}
As already mentioned, the single-target reachability problem
for OBMDPs (and its stochastic game generalization)
is equivalent to the same problem for BMDPs,
and was studied in detail in \cite{ESY-icalp15-IC,EMSY-icalp19-bcsg},
even in the quantitative sense.
The same holds for another
fundamental objective, namely
\textit{termination/extinction},
i.e., where the objective is to
optimize the probability that the generated tree is finite.
The extinction objective for BMDPs, and the closely
related model of 1-exit recursive MDPs, was thoroughly studied 
in \cite{rmdp,rcsg2008,esy-icalp12}, including both qualitative
and quantitative algorithmic questions.   In particular, it was shown
in \cite{rmdp} that qualitative decision problems for termination of (O)BMDPs and 1-exit RMDPs
can be decided in polynomial time.   By directly using this result and building on it, it was shown in \cite{BBFK08}
that ``almost-sure" single-target reachability in 1-exit RMDPs, or equivalently in context-free MDPs with
{\em leftmost derivation}, can be decided in polynomial time.  However,   
context-free MDPs with leftmost derivation are very different than (O)BMDPs, which allow {\em simultaneous} derivation
of the tree from all unexpanded non-terminals in each generation (not just the leftmost one).  
Indeed, unlike  single-target reachability for OBMDPs (equivalently, context-free MDPs with simultaneous derivation), even for single-target reachability
for 1-exit RMDPs (equivalently, context-free MDPs with leftmost derivation),  ``almost-sure" $\neq$ ``limit-sure" and 
the decidability of ``limit-sure" reachability 
of a given target non-terminal remains an open question (despite the fact that there is a polynomial time algorithm for almost-sure reachability).

Algorithms for checking other properties of BPs
and BMDPs have also been investigated before, some of which
generalize termination and reachability.
In particular, model checking of BPs with properties given
by a deterministic parity tree automaton
was studied in \cite{CDK12}, and in  \cite{MichMio15}
for properties represented
by a subclass of alternating parity tree
automata. More recently, \cite{PrzSkrz16} investigated the
determinacy and the complexity of decision problems for 
ordered branching simple (turn-based)
stochastic games with respect to properties
defined by finite tree automata defining regular
languages on infinite trees. They showed that
(unlike the case with reachability) already for some basic regular
properties these games are not even determined, meaning they do
not have a value. Moreover, they show that for what
amounts to OBMDPs with a regular tree objective it is undecidable
to compare the optimal probability to a threshold value.
Their results do not have implications for (neither quantitative nor qualitative) multi-objective reachability.

Multi-objective reachability and model checking 
(with respect to omega-regular properties) has been studied 
for finite-state MDPs in \cite{multi-obj-MDPs2008}, 
both with respect to qualitative and quantitative problems. 
In particular, it was shown in \cite{multi-obj-MDPs2008} 
that for multi-objective reachability in finite-state MDPs, 
memoryless (but randomized) strategies are sufficient, 
that both qualitative and quantitative multi-objective 
reachability queries can be decided in P-time, and 
the \textit{Pareto curve} for them can be approximated 
within a desired error $\epsilon > 0$ in P-time in 
the size of the MDP and $1 / \epsilon$.

\medskip
\noindent {\bf Organization of the paper.}
Section 2 provides background and basic definitions. 
Section 3 gives an algorithm 
for determining the non-terminals starting from which there is 
a strategy that ensures that with a positive probability all 
target non-terminals in the given target set are in the 
generated tree. Sections 4 and 5 provide, respectively, 
our algorithms for the limit-sure and almost-sure 
multi-target reachability problems. 
Section 6 considers other special cases of 
qualitative multi-objective reachability/non-reachability.

\section{Background}

This section introduces background and definitions for 
Ordered Branching Markov Decision Processes (OBMDPs), and 
for the analysis of multi-objective reachability. 
First, we define OBMDPs 
in a general way that combines both control and 
probabilistic rules at each non-terminal, and that allows 
rules to have an arbitrarily long string of non-terminals 
on their right-hand side (RHS). Then we show that any 
OBMDP can be converted efficiently to 
an ``equivalent''\footnote{Equivalent w.r.t. all 
(multi-objective) reachability objectives we consider.} 
one in ``normal'' form.

\begin{definition}
    An \textbf{Ordered Branching Markov Decision Process (OBMDP)},
$\mathcal{A}$, is a 1-player controlled stochastic process,
represented by a tuple $\mathcal{A} = (V, \Sigma, \Gamma, R)$,
where $V = \{T_1, \ldots, T_n\}$ is a finite set of non-terminals, 
and $\Sigma$ is a finite non-empty action alphabet. 
For each $i \in [n]$, $\Gamma^i \subseteq \Sigma$ 
is a finite non-empty set
of actions for non-terminal $T_i \in V$,
and for	each $a \in \Gamma^i$,
$R(T_i, a)$ is a finite set of probabilistic rules
associated with the pair $(T_i, a)$. Each rule $r \in R(T_i, a)$
is a triple,  denoted by $T_i \xrightarrow{p_r} s_r$, where
$s_r \in V^*$ is a (possibly empty) ordered sequence (string)
of non-terminals and $p_r \in (0,1] \cap \mathbb{Q}$
is the positive probability of the rule $r$
(which we assume to be a rational number for
computational purposes). We assume that for each non-terminal
$T_i \in V$ and each $a \in \Gamma^i$, the rule probabilities in $R(T_i, a)$
sum to 1, i.e., $\sum_{r \in R(T_i, a)} p_r = 1$.
    \label{def:OBMDP}
\end{definition}

We denote by $|\mathcal{A}|$ 
the total bit encoding length of the OBMDP. 
If $|\Gamma^i| = 1$ for all non-terminals $T_i \in V$, 
then the model is called 
an \textbf{Ordered Branching Process (OBP)}.

In order to simplify the structure of the OBMDP model 
and to facilitate the proofs throughout the paper, 
we observe a simplified ``equivalent'' normal form 
for OBMDPs (Proposition \ref{prop:SNF-form} later on shows 
that OBMDPs can always be translated efficiently 
into this normal form). 
We extend the notation for rules in the model to adopt 
actions and not only probabilities, i.e., we will be 
using $T_i \xrightarrow{a} T_j$, where $a \in \Gamma^i$, 
to denote a rule where a non-terminal $T_i$ generates 
as a child (under player's choice of action $a \in \Gamma^i$) 
a copy of non-terminal $T_j$ (with probability $1$).

\begin{definition}
    An OBMDP is in \textbf{simple normal form (SNF)} if each 
non-terminal $T_i$ is in one of three possible forms:
\begin{itemize}
    \item \textsc{L-Form}: $T_i$ is a ``linear'' or 
``probabilistic'' non-terminal 
(i.e., the player has no choice of actions), and 
the associated rules for $T_i$ are given by: 
$T_i \xrightarrow{p_{i, 0}} \varnothing, 
T_i \xrightarrow{p_{i, 1}} T_1, \ldots, 
T_i \xrightarrow{p_{i, n}} T_n$, where for all 
$0 \le j \le n$, $p_{i, j} \ge 0$ denotes the probability of 
each rule, and $\sum_{j = 0}^{n} p_{i, j} = 1$.

	\item \textsc{Q-Form}: $T_i$ is a ``quadratic'' 
(or ``branching'') non-terminal, 
with a single associated rule (and no associated 
actions), of the form $T_i \xrightarrow{1} T_j \; T_r$.
	
	\item \textsc{M-Form}: $T_i$ is a ``controlled'' 
non-terminal, with a non-empty set of associated actions 
$\Gamma^i = \{a_1, \ldots, a_{m_i}\} \subseteq \Sigma$, 
and the associated rules have the form 
$T_i \xrightarrow{a_1} T_{j_1}, \ldots, 
T_i \xrightarrow{a_{m_i}} T_{j_{m_i}}$.\footnote{We assume, 
without loss of generality, that for 
$0 \le t < t' \le m_i$, $T_{j_t} \not= T_{j_{t'}}$.}
\end{itemize}
	\label{def:SNF-form}
\end{definition}

A \textit{derivation} for an OBMDP, starting at some start 
non-terminal $T_{start} \in V$, 
is  a (possibly infinite) labeled ordered tree, 
$X = (B,s)$, defined as follows. 
The set of nodes $B \subseteq \{l,r,u\}^*$ of the tree, $X$, 
is a \textit{prefix-closed} subset of $\{l,r,u\}^*$.\footnote{Here 
`l', `r', and `u', stand for `left', `right', and `unique' child, 
respectively.} 
So each node in $B$ is a string over $\{l,r,u\}$, and if 
$w = w' a \in B$, where $a \in \{l,r,u\}$, then $w' \in B$. 
As usual, when $w \in B$ and $w' = wa \in B$, for some 
$a \in \{l,r,u\}$, we call $w$ the \textit{parent} of $w'$, 
and we call $w'$ a \textit{child} of $w$ in the tree. 
A \textit{leaf} of $B$ is a node $w \in B$ that has no children 
in $B$. Let $\mathcal{L}_B \subseteq B$ denote the set of 
all leaves in $B$. 
The \textit{root} node is the empty string $\varepsilon$ 
(note that $B$ is prefix-closed,  so $\varepsilon \in B$). 
The function $s:B \rightarrow V \cup \{ \varnothing \}$ 
assigns either a non-terminal or the empty symbol as a label 
to each node of the tree, and must satisfy the following 
conditions: Firstly, $s(\varepsilon) = T_{start}$, in other 
words the root must be labeled by the start non-terminal; 
Inductively, if for any \textit{non-leaf} node 
$w \in B \setminus \mathcal{L}_B$ we have $s(w) = T_i$, 
for some $T_i \in V$, then:
\begin{itemize}
    \setlength{\topsep}{0em}
    \item if $T_i$ is a \textsf{Q}-form (branching) non-terminal, 
whose associated unique rule is $T_i \xrightarrow{1} T_j \; T_{j'}$, 
then $w$ must have exactly two children in $B$, 
namely $wl \in B$ and $wr \in B$, 
and moreover we must have $s(wl) = T_j$ and $s(wr) = T_{j'}$.

    \item if $T_i$ is a \textsf{L}-form (linear/probabilistic) 
non-terminal, then $w$ must have exactly one child in $B$, 
namely $wu$, and it must be the case that either $s(wu) = T_j$, 
where there exists some rule $T_i \xrightarrow{p_{i, j}} T_j$ 
with a positive probability $p_{i, j} > 0$, 
or else $s(wu) = \varnothing$, where there exists a rule 
$T_i \xrightarrow{p_{i, 0}} \varnothing$, 
with an empty right-hand side, and a positive probability 
$p_{i, 0} > 0$.

    \item if $T_i$ is a \textsf{M}-form (controlled) non-terminal, 
then $w$ must have exactly one child in $B$, namely $wu$, and it must 
be the case that $s(wu) = T_{j_t}$, where there exists some rule 
$T_i \xrightarrow{a_{t}} T_{j_t}$, associated with some action 
$a_t \in \Gamma^i$, having non-terminal $T_i$ as its left-hand side.
\end{itemize}

A derivation $X = (B,s)$ is \textit{finite} if the set $B$ 
is finite. A derivation $X' = (B',s')$ is called 
a \textit{subderivation} of a derivation $X = (B,s)$, 
if $B' \subseteq B$ and $s' = s|_{B'}$ 
(i.e., $s'$ is the function $s$, restricted 
to the domain $B'$). We use $X' \preceq X$ to denote the fact 
that $X'$ is a subderivation of $X$.

A \textit{complete} derivation, or a \textit{play}, $X = (B,s)$, 
is by definition a derivation in which for all leaves 
$w \in \mathcal{L}_B$, $s(w) = \varnothing$. 
For a play $X = (B,s)$, and a node $w \in B$, 
we define the \textit{subplay of $X$ rooted at $w$}, 
to be the play $X^w=(B^w,s^w)$, 
where $B^w = \{w' \in \{l,r,u\}^* \mid  w w' \in B \}$ and 
$s^w : B^w \rightarrow V \cup \{ \varnothing \}$ 
is given by, $s^w(w') := s(ww')$ for all $w' \in B^w$.\footnote{To 
avoid confusion, note that subderivation and subplay have 
very different meanings. Saying derivation $X$ is 
a ``subderivation'' of $X'$,  means that in a sense 
$X$ is a ``prefix'' of $X'$, as an ordered tree. 
Saying play $X$ is a subplay of play $X'$, means $X$ is 
a ``suffix'' of $X'$, more specifically $X$ is a subtree 
rooted at a specific node of $X'$.} 
Consider any derivation $X = (B,s)$, 
and any node $w = w_1 \ldots w_m \in B$, where $w_t \in \{l,r,u\}$ 
for all $t \in [m]$. We define the \textit{ancestor history} 
of $w$ to be a sequence $h_w \in V (\{l,r,u\} \times V)^*$, given by 
$h_w :=  s(\varepsilon) (w_1, s(w_1)) (w_2, s(w_1 w_2)) 
(w_3, s(w_1 w_2 w_3)) \ldots (w_m, s(w_1 w_2 \ldots w_m))$. 
In other words, the ancestor history $h_w$ of node $w$ specifies 
the sequence of moves that determine each ancestor of $w$ 
(starting at $\varepsilon$ and including $w$ itself), 
and also specifies the sequence of non-terminals that 
label each ancestor of $w$. 

For an OBMDP, $\mathcal{A}$, a sequence 
$h \in V(\{l,r,u\} \times V)^*$ 
is called a \textit{valid} ancestor history 
if there is some derivation $X = (B',s')$ of $\mathcal{A}$, 
and node $w \in B'$ such that $h = h_w$. 
We define the \textit{current non-terminal} of such 
a valid ancestor history $h$ to be $s'(w)$. 
In other words, it is the non-terminal that labels the last node 
of the ancestor history $h$. 
Let $\current(h)$ denote the current non-terminal of $h$. 
Let $H_{\mathcal{A}} \subseteq V (\{l,r,u\} \times V)^*$ 
denote the set of all valid ancestor histories of $\mathcal{A}$. 
A valid ancestor history $h \in H_{\mathcal{A}}$ is said to 
\textit{belong to the controller}, 
if $\current(h)$ is a \textsf{M}-form (controlled) non-terminal. 
Let $H^C_{\mathcal{A}}$ denote the set of all 
valid ancestor histories of the OBMDP, $\mathcal{A}$, 
that belong to the controller.

For an OBMDP, $\mathcal{A}$, 
a \textit{strategy} for the controller is a function, 
$\sigma : H^C_{\mathcal{A}} \rightarrow \Delta(\Sigma)$ 
from the set of valid ancestor histories belonging to the controller, 
to probability distributions on actions, such that moreover 
for any $h \in H^{C}_{\mathcal{A}}$, 
if $\current(h) = T_i$, then $\sigma(h) \in \Delta(\Gamma^i)$. 
(In other words, the probability distribution 
must have support only on the actions 
available at the current non-terminal.) 
Note that the strategy can choose different distributions 
on actions at different occurrences of the same non-terminal 
in the derivation tree, even when these occurrences 
happen to be ``siblings'' in the tree.

Let $\Psi$ be the set of all strategies. 
We say $\sigma \in \Psi$ is \textit{deterministic} if 
for all $h \in H^{C}_{\mathcal{A}}$, 
$\sigma(h)$ puts probability $1$ on a single action. 
We say $\sigma \in \Psi$ is \textit{static} if for each 
\textsf{M}-form (controlled) non-terminal $T_i$, 
there is some distribution $\delta_i \in \Delta(\Gamma^i)$, 
such that for any $h \in H^C_{\mathcal{A}}$ with $\current(h) = T_i$, 
$\sigma(h) = \delta_i$. 
In other words, a static strategy $\sigma$ plays, for each 
\textsf{M}-form non-terminal $T_i$, exactly 
the same distribution on actions at every occurrence of $T_i$, 
regardless of the ancestor history.

For an OBMDP, $\mathcal{A}$, fixing a start non-terminal $T_{i}$, 
and fixing a strategy $\sigma$ for the controller, 
determines a stochastic process that generates a random play, 
as follows. The process generates a sequence of finite derivations, 
$X_0$, $X_1$, $X_2$, $X_3$, $\ldots$, one for each ``generation'', 
such that for all $t \in \nat$,  $X_t \preceq X_{t+1}$. 
$X_0 = (B_0,s_0)$ is the initial derivation, at generation $0$, 
and consists of a single (root) node $B_0 = \{ \varepsilon \}$, 
labeled by the start non-terminal, 
$s_0(\varepsilon) = T_i$.\footnote{We can assume, without 
loss of generality, that the initial derivation consists 
of a single given root, because for any given collection 
$\mu \in V^*$ of multiple roots, we can always add 
an auxiliary non-terminal $T_f$ to the set $V$, 
where $\Gamma^f = \{a\}$ and the set $R(T_f, a)$ contains 
a single probabilistic rule, $T_f \xrightarrow{1} \mu$.} 
Inductively, for all $t \in \nat$ the derivation 
$X_{t+1} = (B_{t+1}, s_{t+1})$ is obtained from $X_t = (B_t, s_t)$ 
as follows. For each leaf $w \in \mathcal{L}_{B_t}$:
\begin{itemize}
    \setlength{\topsep}{0em}
    \item if $s_t(w) = T_i$ is a \textsf{Q}-form (branching) 
non-terminal, whose associated unique rule is 
$T_i \xrightarrow{1} T_j \; T_{j'}$, 
then $w$ must have exactly two children in $B_{t+1}$, namely 
$wl \in B_{t+1}$ and $wr \in B_{t+1}$, and moreover 
we must have $s_{t+1}(wl) = T_j$ and $s_{t+1}(wr) = T_{j'}$.

    \item if $s_t(w) = T_i$ is a \textsf{L}-form (probabilistic) 
non-terminal, then $w$ has exactly one child in $B_{t+1}$, namely $wu$, 
and for each rule $T_i \xrightarrow{p_{i, j}} T_j$ 
with  $p_{i, j} > 0$,  the probability that 
$s_{t+1}(wu) = T_j$ is $p_{i, j}$, and likewise 
when  $T_i \xrightarrow{p_{i, 0}} \varnothing$ 
is a rule with $p_{i, 0} > 0$, then 
$s_{t+1}(wu) = \varnothing$ with probability $p_{i, 0}$.

    \item if $s_t(w) = T_i$ is a \textsf{M}-form (controlled) 
non-terminal, then $w$ has exactly one child in $B_{t+1}$, 
namely $wu$, and for each action $a_z \in \Gamma^i$, 
with probability $\sigma(h_w)(a_z)$,  $s_{t+1}(wu) = T_{j_z}$, 
where $T_i \xrightarrow{a_z} T_{j_z}$ is the rule associated 
with $a_z$.
\end{itemize}

There are no other nodes in $B_{t+1}$. In particular, 
if $s_t(w) = \varnothing$, then in $B_{t+1}$ the node $w$ 
has no children. This defines a stochastic process, 
$X_0, X_1, X_2, \ldots$,  where $X_{t} \preceq X_{t+1}$, 
for all $t \in \nat$, and such that there is a unique play, 
$X = \lim_{t \rightarrow \infty} X_t$, such that 
$X_t \preceq X$ for all $t \in \nat$. In this sense, 
the random process defines a probability space of plays.

For our purposes, an \textit{objective} is specified 
by a property (i.e., a measurable set), $\mathcal{F}$, 
of plays, whose probability the player wishes to optimize 
(maximize or minimize). Different objectives can be 
considered for OBMDPs (and for their game extensions). 
In the \textit{termination} objective, 
the player aims to maximize/minimize the probability that 
the process terminates, i.e., that the play is a finite 
tree. This was studied in \cite{rmc} for 
purely stochastic OBPs, and in \cite{rmdp} 
(and \cite{rcsg2008}) for their MDP and (concurrent) 
stochastic game generalizations. 
Another objective is (single-target) \textit{reachability}, 
where the goal is to optimize (maximize or minimize) 
the probability of the play containing a given target 
non-terminal, starting at a given non-terminal. 
This objective was studied 
in \cite{ESY-icalp15-IC} (and \cite{EMSY-icalp19-bcsg}) 
for OBMDPs and their (concurrent) stochastic game generalizations.\footnote{The 
models analysed in \cite{ESY-icalp15-IC} and 
\cite{EMSY-icalp19-bcsg} are game generalizations 
of Branching Processes, but for the case of a single target 
computing reachability probabilities in Branching Processes 
is equivalent to computing reachability probabilities in 
Ordered Branching Processes (same holds for the MDP and 
game generalizations of these models).}

This paper considers the 
\textit{multi-objective reachability problem}, which is a 
natural extension of the previously studied 
(single-target) reachability problem. 
In the multi-objective setting we have multiple 
target non-terminals, and we want to optimize each of 
the respective probabilities of achieving multiple given 
objectives, each one being a boolean combination of 
reachability and non-reachability properties over 
different target non-terminals. 
Of course, there may be tradeoffs between 
optimizing the probabilities of achieving 
the different objectives.

To formalize things, we need some notation. Given 
a target non-terminal $T_q$, $q \in [n]$, let $Reach(T_q)$ 
denote the set of plays that contain some copy 
(some node) of non-terminal $T_q$. 
Respectively, let $Reach^{\complement}(T_q)$  denote the 
complement event, i.e., the set of plays that do 
\textit{not} contain a node labelled by non-terminal $T_q$. 
For any measurable set (i.e., property) of plays, 
$\mathcal{F}$, and for any strategy $\sigma$ for the player 
and a given start non-terminal $T_i$, 
we denote by $Pr_{T_i}^{\sigma}[\mathcal{F}]$ 
the probability that, starting at a non-terminal $T_i$ and 
under strategy $\sigma$, the generated play is in 
the set $\mathcal{F}$. 
Let $Pr_{T_i}^*[\mathcal{F}] := \sup_{\sigma \in \Psi} 
Pr_{T_i}^{\sigma}[\mathcal{F}]$.

The \textit{quantitative} multi-objective decision problem 
for OBMDPs is the following problem. We are given an OBMDP, 
a starting non-terminal $T_s \in V$, 
a collection of objectives (properties) 
$\mathcal{F}_1, \ldots, \mathcal{F}_k$ 
and corresponding probabilities $p_1, \ldots, p_k$. 
The problem asks to decide whether there exists a strategy 
$\sigma' \in \Psi$ such that $\bigwedge_{i \in [k]} 
Pr_{T_s}^{\sigma'}[\mathcal{F}_i] \triangle_i p_i$ holds, 
where $\triangle_i \in \{<, \le, =, \ge, >\}$. Observe that 
terms (i.e., probability queries 
$Pr_{T_s}^{\sigma'}[\mathcal{F}_i] \triangle_i p_i$, 
for any $i \in [k]$) with $\triangle_i := \; \le$ and 
$\triangle_i := \; \ge$ inequalities can be converted to asking 
whether either $Pr_{T_s}^{\sigma'}[\mathcal{F}_i] = p_i$, 
or $Pr_{T_s}^{\sigma'}[\mathcal{F}_i] < p_i$ 
(respectively, $Pr_{T_s}^{\sigma'}[\mathcal{F}_i] > p_i$). 
Moreover, we could in general allow for any boolean combination 
of terms (not just a conjunction). In any case, 
the whole query can be put into disjunctive normal form and 
the quantification over strategies can be 
pushed inside the disjunction. So any multi-objective question 
can eventually be transformed into a disjunction of finite 
number of (smaller) queries. (Note that, 
of course, this number can be exponential in the size of the 
original multi-objective question.) Hence, we can define a 
multi-objective decision problem only as a conjunction 
of equality and strict inequality queries.

One could also ask the \textit{limit} 
version of this question. For instance, 
whether for all $\epsilon > 0$, there exists 
a strategy $\sigma'_\epsilon$, such that  
$\bigwedge_{i \in [k]} Pr_{T_s}^{\sigma'_\epsilon} 
[\mathcal{F}_i] \ge p_i - \epsilon$. 
Moreover, we can also ask quantitative 
questions regarding computing (or approximating) 
the {\em Pareto curve} for the multiple objectives, but we will
not consider such questions in this paper.

The qualitative \textit{almost-sure} multi-objective decision 
problems for OBMDPs are the special case where $p_i = \{0, 1\}$ 
for each $i \in [k]$. In other words, these problems are 
phrased as asking whether, starting at a given non-terminal 
$T_s \in V$, there exists a strategy $\sigma \in \Psi$ 
such that $\bigwedge_{i \in [k]} Pr_{T_s}^{\sigma}[\mathcal{F}_i] 
\triangle_i \{0, 1\}$ (where as mentioned $\triangle_i \in 
\{<, =, >\}$). We can simplify the expression by transforming 
each clause of the form $Pr_{T_s}^{\sigma}[\mathcal{F}_i] > 0$ 
and $Pr_{T_s}^{\sigma}[\mathcal{F}_i] = 0$ into 
$Pr_{T_s}^{\sigma}[\mathcal{F}_i^{\complement}] < 1$ and 
$Pr_{T_s}^{\sigma}[\mathcal{F}_i^{\complement}] = 1$, 
respectively, where each $\mathcal{F}_i^{\complement}$ 
is the complement objective of $\mathcal{F}_i$.

Then, for a strategy $\sigma \in \Psi$ and a starting 
non-terminal $T_s \in V$, the expression can be rephrased as: 
$\bigwedge_{i \in [k_1]} Pr_{T_s}^{\sigma}[\mathcal{F}_i] < 1 
\wedge \bigwedge_{i \in [k_2]} 
Pr_{T_s}^{\sigma}[\mathcal{F}_i] = 1$, where 
$k_1 + k_2 = k$. And by Proposition \ref{prop:equiv}(1.) below, 
the qualitative (almost-sure) multi-objective decision problem 
reduces to asking whether there exists a strategy 
$\sigma' \in \Psi$ such that 
$\bigwedge_{i \in [k_1]} Pr_{T_s}^{\sigma'}[\mathcal{F}_i] < 1 
\wedge Pr_{T_s}^{\sigma'}[\bigcap_{i \in [k_2]} \mathcal{F}_i] = 1$.

The qualitative \textit{limit-sure} multi-objective decision 
problem for OBMDPs asks to decide whether, 
for every $\epsilon > 0$, there exists a strategy 
$\sigma'_{\epsilon} \in \Psi$ such that 
$\bigwedge_{i \in [k]} Pr_{T_s}^{\sigma'_{\epsilon}} 
[\mathcal{F}_i] \ge 1 - \epsilon$. Again by Proposition 
\ref{prop:equiv}(5.) below, it follows that the qualitative 
limit-sure multi-objective decision problem can be rephrased 
as asking whether, for all $\epsilon > 0$, there exists 
a strategy $\sigma'_{\epsilon} \in \Psi$ such that 
$Pr_{T_s}^{\sigma'_{\epsilon}}[\bigcap_{i \in [k]} \mathcal{F}_i] \ge 1 - \epsilon$.

The following proposition shows scenarios where the qualitative 
multi-objective problem for OBMDPs can be rephrased as a 
qualitative single-objective problem, but with multiple targets!

\begin{proposition}
    Given an OBMDP, with a starting non-terminal 
$T_s \in V$ and a collection 
$\mathcal{F}_1, \ldots, \mathcal{F}_{k}$ of $k$ objectives:
	\begin{enumerate}[label=(\arabic*.)]
        \item $\exists \sigma' \in \Psi:\; \bigwedge_{i \in [k]} Pr_{T_s}^{\sigma'}[\mathcal{F}_i] = 1$ \ \ {\rm if and only if} \ \  $\exists \sigma' \in \Psi:\; Pr_{T_s}^{\sigma'}[\bigcap_{i \in [k]} \mathcal{F}_i] = 1$.
		
		\item $\exists \sigma' \in \Psi:\; \bigvee_{i \in [k]} Pr_{T_s}^{\sigma'}[\mathcal{F}_i] < 1$ \ \ {\rm if and only if} \ \  $\exists \sigma' \in \Psi:\; Pr_{T_s}^{\sigma'}[\bigcap_{i \in [k]} \mathcal{F}_i] < 1$.
		
		\item $\exists \sigma' \in \Psi:\; \bigwedge_{i \in [k]} Pr_{T_s}^{\sigma'}[\mathcal{F}_i] = 0$ \ \ {\rm if and only if} \ \ $\exists \sigma' \in \Psi:\; Pr_{T_s}^{\sigma'}[\bigcup_{i \in [k]} \mathcal{F}_i] = 0$.
		
		\item $\exists \sigma' \in \Psi:\; \bigvee_{i \in [k]} Pr_{T_s}^{\sigma'}[\mathcal{F}_i] > 0$ \ \ {\rm if and only if} \ \ $\exists \sigma' \in \Psi:\; Pr_{T_s}^{\sigma'}[\bigcup_{i \in [k]} \mathcal{F}_i] > 0$.

Moreover, in each of the equivalence statements (1.) - (4.), 
a witness strategy $\sigma'$ for one of the sides 
is also a witness strategy for the other.
		
		\item Similar equivalence holds for the qualitative limit-sure multi-objective problem: $\forall \epsilon > 0, \exists \sigma'_{\epsilon} \in \Psi:\; \bigwedge_{i \in [k]} Pr_{T_s}^{\sigma'_{\epsilon}}[\mathcal{F}_i] \ge  1 - \epsilon$ \  \ {\rm if and only if} \ \  $\forall \epsilon > 0, \exists \sigma'_{\epsilon} \in \Psi:\; Pr_{T_s}^{\sigma'_{\epsilon}}[\bigcap_{i \in [k]} \mathcal{F}_i] \ge 1 - \epsilon$.

And from a witness strategy $\sigma'_{\epsilon}$ 
(for $\epsilon > 0$) for one of the two sides 
a witness strategy $\sigma''_{\epsilon'}$ 
(for potentially different $\epsilon' > 0$) 
can be obtained for the other.
	\end{enumerate}
	\label{prop:equiv}
\end{proposition}

\begin{proof}
$ $ \newline
(1.). For one direction of the statement, suppose there is 
a strategy $\sigma' \in \Psi$ for the player such that 
$Pr_{T_s}^{\sigma'}[\bigcap_{i \in [k]} \mathcal{F}_i] = 1$, 
i.e., almost-surely all objectives are satisfied in 
the same generated play. It follows that 
$Pr_{T_s}^{\sigma'}[\bigcup_{i \in [k]} \mathcal{F}_i^{\complement}] = 0$. 
Clearly, for each $i \in [k]$, 
$Pr_{T_s}^{\sigma'}[\mathcal{F}_i^{\complement}] = 0$ and hence, 
for each $i \in [k]:\; Pr_{T_s}^{\sigma'}[\mathcal{F}_i] = 1$.

Showing the other direction, suppose that there exists 
a strategy $\sigma' \in \Psi$ for the player such that 
$\bigwedge_{i \in [k]} Pr_{T_s}^{\sigma'}[\mathcal{F}_i] = 1$. 
Then, $\forall i \in [k]$, 
$Pr_{T_s}^{\sigma'}[\mathcal{F}_i^{\complement}] = 0$. 
By the union bound, $Pr_{T_s}^{\sigma'}[\bigcup_{i \in [k]} \mathcal{F}_i^{\complement}] = 0$ and, hence, $Pr_{T_s}^{\sigma'}[\bigcap_{i \in [k]} \mathcal{F}_i] = 1$.

\medskip \medskip

\noindent (2.). For one direction of the statement, suppose 
there is a strategy $\sigma' \in \Psi$ such that 
$Pr_{T_s}^{\sigma'}[\bigcap_{i \in [k]} \mathcal{F}_i] < 1$. 
Then $Pr_{T_s}^{\sigma'}[\bigcup_{i \in [k]} \mathcal{F}_i^{\complement}] > 0$. Clearly, $\exists i' \in [k]$ such that 
$Pr_{T_s}^{\sigma'}[\mathcal{F}_{i'}^{\complement}] > 0$ 
(otherwise, by the union bound the probability of 
the union of the events is $0$). 
Hence, $\bigvee_{i \in [k]} Pr_{T_s}^{\sigma'}[\mathcal{F}_i] < 1$.

As for the other direction, suppose there is a strategy 
$\sigma' \in \Psi$ and some $i' \in [k]$ such that 
$Pr_{T_s}^{\sigma'}[\mathcal{F}_{i'}] < 1$. 
Then $Pr_{T_s}^{\sigma'}[\bigcap_{i \in [k]} \mathcal{F}_i] \le Pr_{T_s}^{\sigma'}[\mathcal{F}_{i'}] < 1$.

\medskip \medskip
\noindent (3.) and (4.) follow directly from (1.) and (2.), 
respectively.

\medskip \medskip

\noindent (5.). For one direction of the statement, suppose 
that for every $\epsilon > 0$ there is a strategy 
$\sigma'_{\epsilon} \in \Psi$ such that 
$Pr_{T_s}^{\sigma'_{\epsilon}}[\bigcap_{i \in [k]} \mathcal{F}_i] \ge  1 - \epsilon$, i.e., 
limit-surely (with probability arbitrarily close to $1$) 
all objectives are satisfied in the same generated play. 
It follows that $Pr_{T_s}^{\sigma'_{\epsilon}}[\bigcup_{i \in [k]} \mathcal{F}_i^{\complement}] \le \epsilon$. 
Clearly, for each $i \in [k]$, 
$Pr_{T_s}^{\sigma'_{\epsilon}}[\mathcal{F}_i^{\complement}] \le \epsilon$, and hence, for each $i \in [k]:\; Pr_{T_s}^{\sigma'_{\epsilon}}[\mathcal{F}_i] \ge 1 - \epsilon$.

\medskip
Showing the other direction, suppose that for every 
$\epsilon > 0$ there exists a strategy 
$\sigma'_{\epsilon} \in \Psi$ such that 
$\bigwedge_{i \in [k]} Pr_{T_s}^{\sigma'_{\epsilon}} [\mathcal{F}_i] \ge 1 - \epsilon$. 
Then, for every $i \in [k]$, 
$Pr_{T_s}^{\sigma'_{\epsilon}}[\mathcal{F}_i^{\complement}] \le \epsilon$. By the union bound, 
$Pr_{T_s}^{\sigma'_{\epsilon}}[\bigcup_{i \in [k]} \mathcal{F}_i^{\complement}] \le k \epsilon$, and hence, 
$Pr_{T_s}^{\sigma'_{\epsilon}}[\bigcap_{i \in [k]} \mathcal{F}_i] \ge 1 - k \epsilon$. 
So for any $\epsilon > 0$, let $\epsilon' := \epsilon / k$ 
and $\sigma_{\epsilon} := \sigma'_{\epsilon'}$, 
where $\sigma'_{\epsilon'}$ satisfies 
$\bigwedge_{i \in [k]} Pr_{T_s}^{\sigma'_{\epsilon'}} [\mathcal{F}_i] \ge 1 - \epsilon' = 1 - \epsilon / k$. 
Then it follows that $Pr_{T_s}^{\sigma_{\epsilon}}[\bigcap_{i \in [k]} \mathcal{F}_i] \ge 1 - k \epsilon' = 1 - \epsilon$.
\end{proof}

In this paper, we address the qualitative 
(almost-sure and limit-sure) multi-objective 
reachability decision problems for OBMDPs. 
We are given a collection of \textit{generalized} 
reachability objectives $\mathcal{F}_1, \ldots, \mathcal{F}_k$,
where each such generalized reachability objective 
$\mathcal{F}_i, i \in [k]$ represents a set of plays 
described by a boolean combination in CNF form 
over the sets (of plays) $Reach(T_q), T_q \in V$ 
and under the operators \textit{union}, 
\textit{intersection} and \textit{complementation}. 
That is, each generalized reachability objective 
$\mathcal{F}_i, i \in [k]$ is of the form 
$\bigcap_{t \in [z_i]} (\bigcup_{t' \in [z_{i, t}]} 
\Phi(T_{q_{i, t, t'}}))$, where $\Phi \in \{Reach, Reach^{\complement}\}$, $T_{q_{i, t, t'}} \in V$ and 
the values $z_i, z_{i, t}$ are part of 
the objective $\mathcal{F}_i$.

We will show that, even in the case of having a single 
objective that asks to reach multiple target non-terminals 
from a given set in the same play, the almost-sure and 
limit-sure questions do not coincide and we give 
separate algorithms for detecting almost-sure and 
limit-sure multi-target reachability. 
(Recall from the related work section, that in the case of 
a single target the almost-sure and limit-sure questions are 
equivalent.)
The following example indeed illustrates that there are 
OBMDPs where, even though 
the supremum probability of reaching all target non-terminals 
from a given set in the same play is $1$, there may not exist 
a strategy for the player that actually achieves probability 
exactly $1$.

\medskip
\noindent \textbf{Example 1} 
Consider the following OBMDP with non-terminals 
$\{M, A, R_1, R_2\}$, where $R_1$ and $R_2$ are 
the target non-terminal. $M$ is the only ``controlled" 
non-terminal, and the rules are:
\begin{align*}
	& M \xrightarrow{a} MA && A \xrightarrow{1/2} R_1\\
	& M \xrightarrow{b} R_2 && A \xrightarrow{1/2} \varnothing
\end{align*}

The supremum probability, $Pr_M^*[Reach(R_1) \cap Reach(R_2)]$, 
starting at a non-terminal $M$, 
of reaching both targets is $1$. To see this, for 
any $\epsilon > 0$, let the strategy keep choosing 
deterministically action $a$ until 
$l := \lceil \log_2 (\frac{1}{\epsilon}) \rceil$ copies 
of non-terminal $A$ have been created, i.e., until 
the play reaches generation $l$. 
Then in the (unique) copy of non-terminal $M$ in 
generation $l$ the strategy switches 
deterministically to action $b$. The probability of 
reaching target $R_2$ is $1$. The probability of reaching 
target $R_1$ is $1 - 2^{-l} \ge 1 - \epsilon$. 
The player can delay arbitrarily long the moment 
when to switch from choosing action $a$ to choosing 
action $b$ for a non-terminal $M$.
Hence, $Pr_M^*[Reach(R_1) \cap Reach(R_2)] = 1$. 

However, $\nexists \sigma \in \Psi:\; 
Pr_M^{\sigma}[Reach(R_1) \cap Reach(R_2)] = 1$. 
To see this, note that if the strategy ever puts 
a positive probability on action $b$ in any ``round'', 
then with a positive probability target $R_1$ will not be reached 
in the play. So, to reach target $R_1$ with probability $1$, 
the strategy must deterministically choose action $a$ forever, 
from every occurrence of non-terminal $M$. But if it does this 
the probability of reaching target $R_2$ would be $0$.
\qed

\medskip
The following proposition is easy to prove 
(similar to analogous propositions in 
\cite{ESY-icalp15-IC,EMSY-icalp19-bcsg}) and shows 
that we can always efficiently convert an OBMDP into 
its SNF form (Definition \ref{def:SNF-form}).

\begin{proposition}
  Every OBMDP, $\mathcal{A}$, can be converted in P-time to an 
``equivalent'' OBMDP, $\mathcal{A'}$, in SNF form, such that 
$|\mathcal{A'}| \in O(|\mathcal{A}|)$. More precisely, 
the non-terminals $V = \{T_i \mid i \in [n]\}$ of $\mathcal{A}$ 
are a subset of the non-terminals of $\mathcal{A'}$, 
and any strategy $\sigma$ of $\mathcal{A}$ can be converted 
to a strategy $\sigma'$ of $\mathcal{A'}$ (and vice versa), 
such that starting at any non-terminal $T_s \in V$, and 
for any generalized reachability objective $\mathcal{F}$, 
using the strategies $\sigma$ and $\sigma'$ in $\mathcal{A}$ 
and $\mathcal{A'}$, respectively, the probability that the 
resulting play is in the set of plays, $\mathcal{F}$, 
is the same in both $\mathcal{A}$ and $\mathcal{A'}$.
	\label{prop:SNF-form}
\end{proposition}

\begin{proof}
    For a rule $T_i \xrightarrow{p_r} s_r,\; s_r \in V^*$ in 
$\mathcal{A}$ and a non-terminal $T_j$, let 
$m_{r, j} := |\{d \mid (s_r)_d = T_j,\; 1 \le d \le |s_r|\}|$ be 
the number of copies of $T_j$ in string $s_r$. 
We use the following procedure to convert, in P-time, 
any OBMDP, $\mathcal{A}$, into its SNF-form OBMDP, 
$\mathcal{A'}$.
\begin{enumerate}
	\item Initialize $\mathcal{A'}$ by adding all 
the non-terminals $T_i \in V$ from $\mathcal{A}$ and 
their corresponding action sets $\Gamma^i$.

	\item For each non-terminal $T_i$, such that 
$m_{r, i} > 1$ for some non-terminal $T_j$, 
action $a \in \Gamma^j$ and rule $r \in (T_j, a)$ 
from $\mathcal{A}$, create new non-terminals 
$T_{i_1}, \ldots, T_{i_z}$ in $\mathcal{A'}$ where 
$z = \lfloor \log_2 (\max_{r \in R} \{m_{r, i}\}) \rfloor$. 
Then add the rules $T_{i_1} \xrightarrow{1} T_i \; T_i,\; T_{i_2} \xrightarrow{1} T_{i_1} \; T_{i_1},\; \ldots,\; T_{i_z} \xrightarrow{1} T_{i_{z-1}} \; T_{i_{z-1}}$ to $\mathcal{A'}$. 
For every rule $r \in R$ in OBMDP, $\mathcal{A}$, where 
$m_{r, i} > 1$, if the binary representation of $m_{r, i}$ is 
$l_z \ldots l_2 l_1 l_0$, then we remove all copies of $T_i$ 
in string $s_r$ (i.e., the right-hand side of rule $r$) and add 
a copy of non-terminal $T_{i_t}$ to string $s_r$ if bit $l_t = 1$, 
for every $0 \le t \le [z]$. After this step, for every rule 
$r \in R$, the string $s_r$ consists of at most one copy of 
any non-terminal.
	
	\item For each non-terminal $T_i$, for each action 
$a_d \in \Gamma^i$, create a new non-terminal $T_d$ in 
$\mathcal{A'}$ and add the rule $T_i \xrightarrow{a_d} T_d$ 
to $\mathcal{A'}$.
	
	\item Next, for each such new non-terminal $T_d$ from 
point 3., for each rule $r$ from set $R(T_i, a_d)$ 
in $\mathcal{A}$: if $s_r = \varnothing$ (i.e., the set of 
offsprings under rule $r$ is empty), then add the rule 
$T_d \xrightarrow{p_r} \varnothing$ to $\mathcal{A'}$; 
if the set of offsprings consists of a single copy of some 
non-terminal $T_j$, then add the rule $T_d \xrightarrow{p_r} T_j$ 
to $\mathcal{A'}$; and if the set of offsprings is larger and 
$s_r$ does not have an associated non-terminal already, 
then create a new non-terminal $T_{d_r}$, associated with 
string $s_r$, in $\mathcal{A'}$ and add the rule 
$T_d \xrightarrow{p_r} T_{d_r}$ to $\mathcal{A'}$.
	
	\item Next, for each such new non-terminal $T_{d_r}$, 
associated with $s_r, r \in R(T_i, a_d)$, where 
$s_r$ has $m \ge 2$ non-terminals 
$T_{j_1}, \ldots, T_{j_m}$: if $m = 2$, add rule 
$T_{d_r} \xrightarrow{1} T_{j_1} \; T_{j_2}$ to $\mathcal{A'}$; 
and if $m > 2$, create $m - 2$ new non-terminals 
$T_{l_1}, \ldots, T_{l_{m-2}}$ in $\mathcal{A'}$ and 
add the rules $T_{d_r} \xrightarrow{1} T_{j_1} \; T_{l_1},\; 
T_{l_1} \xrightarrow{1} T_{j_2} \; T_{l_2},\; 
T_{l_2} \xrightarrow{1} T_{j_3} \; T_{l_3},\; \ldots,\; 
T_{l_{m-2}} \xrightarrow{1} T_{j_{m-1}} \; T_{j_m}$ 
to $\mathcal{A'}$.
\end{enumerate}
Now all non-terminals are of form \textsf{L}, \textsf{Q} or 
\textsf{M}.

The above procedure converts any OBMDP, $\mathcal{A}$, into 
one in SNF form by introducing $O(|\mathcal{A}|)$ new 
non-terminals and blowing up the size of $\mathcal{A}$ 
by a constant factor $O(1)$. 
Moreover, any strategy $\sigma$ of the original OBMDP, 
$\mathcal{A}$, can be converted to a strategy $\sigma'$ 
of the SNF-form OBMDP, $\mathcal{A'}$ (and vice versa), 
such that, under strategies $\sigma$ and $\sigma'$ 
in $\mathcal{A}$ and $\mathcal{A'}$, respectively, the 
probability that the resulting play is in the set of plays 
of a given generalized reachability objective $\mathcal{F}$ 
is the same in both $\mathcal{A}$ and $\mathcal{A'}$.
\end{proof}

From now on, throughout the rest of the paper 
\textit{we may assume, without loss of generality, 
that any OBMDP is in SNF form}. 
We shall hereafter use the notation $T_i \rightarrow T_j$ 
(respectively, $T_i \not\rightarrow T_j$), 
to denote that for non-terminal $T_i$ there {\em exists} 
(respectively, there does {\em not} exist) 
either an associated (controlled) rule $T_i \xrightarrow{a} T_j$, 
where $a \in \Gamma^i$, or an associated probabilistic rule 
$T_i \xrightarrow{p_{i,j}} T_j$ with a positive probability 
$p_{i, j} > 0$. Similarly, let $T_i \rightarrow \varnothing$ 
(respectively, $T_i  \not\rightarrow \varnothing$) denote 
that the rule $T_i \xrightarrow{p_{i, 0}} \varnothing$ has 
a positive probability $p_{i, 0} > 0$ (respectively, 
has probability $p_{i, 0} = 0$).

\begin{definition}
    The \textbf{dependency graph} of a SNF-form OBMDP, 
$\mathcal{A}$, is a directed graph that has a node $T_i$ 
for each non-terminal $T_i$, and contains an edge $(T_i, T_j)$ 
if and only if: either $T_i \rightarrow T_j$ or there is 
a rule $T_i \xrightarrow{1} T_j \; T_r$ or a rule 
$T_i \xrightarrow{1} T_r \; T_j$ in $\mathcal{A}$.
	\label{def:dependency_graph}
\end{definition}
Throughout this paper, for (SNF-form) OBMDP, $\mathcal{A}$, 
with non-terminals set $V$, we let $G = (U, E)$, 
with $U = V$, denote the dependency graph of $\mathcal{A}$ 
and let $G[C]$ denote the subgraph of $G$ induced by 
the subset $C \subseteq U$ of nodes (non-terminals).

Sometimes when the specific OBMDP, $\mathcal{A}$, is not clear 
from the context, we use $\mathcal{A}$ as 
superscript to specify the OBMDP in our notations. 
So, for instance, $\Psi^{\mathcal{A}}$ is the set of all 
strategies for $\mathcal{A}$; $G^{\mathcal{A}}$ is the 
dependency graph of $\mathcal{A}$; and 
$Pr_{T_i}^{\sigma, \mathcal{A}}[\mathcal{F}]$ is 
the probability of event $\mathcal{F}$, 
starting at a non-terminal $T_i$, under strategy $\sigma$, 
in $\mathcal{A}$.

We also extend the notation regarding probabilities 
of properties to ``start'' at a given ancestor history. 
That is, for an ancestor history $h$, we use 
$Pr_h^{\sigma, \mathcal{A}}[\mathcal{F}]$ to denote 
the conditional probability that, using 
$\sigma \in \Psi^{\mathcal{A}}$, conditioned on the event 
that there is a node in the play whose ancestor history is $h$, 
the \textit{subplay} rooted at $\current(h)$, is in the set 
$\mathcal{F}$. Whenever we use the notation 
$Pr_h^{\sigma, \mathcal{A}}[\mathcal{F}]$, the underlying 
conditional probability will be well defined. 
Again, the superscript $\mathcal{A}$ will be omitted when 
clear from context.

Note that one ancestor history $h$ can be a prefix 
of another ancestor history. 
We use the notation $h' := h(x,T_i)$, for some 
$x \in \{l,r,u\}$, to denote that $h$ is 
the immediately prior ancestor history to $h'$, 
which is obtained by concatenating the pair $(x,T_i)$ 
at the end of $h$.

\begin{definition}
    For a directed graph $G = (U, E)$, and a partition of its 
vertices $U = (U_1, U_P)$, an \textbf{end-component} is 
a set of vertices $C \subseteq U$ such that $G[C]$: 
(1) is strongly connected; (2) for all $u \in U_P \cap C$ 
and all $(u,u') \in E$, $u' \in C$; 
(3) and if $C = \{u\}$ (i.e., $|C| = 1$), then $(u, u) \in E$. 
A \textbf{maximal end-component (MEC)} is an end-component 
not contained in any larger end-component. 
A \textbf{MEC-decomposition} is a partition of the graph 
into MECs and nodes that do not belong to any MEC.
    \label{def:MEC}
\end{definition}

MECs are disjoint and the unique MEC-decomposition of such 
a directed graph $G$ (with partitioned nodes) can be computed 
in P-time (\cite{CY98}).\footnote{In \cite{CY98}, maximal 
end-components are referred to as \textit{closed components}.} 
More recent work provides more efficient algorithms for 
MEC-decomposition (see \cite{Chat-Hen14}). 
We will also be using the notion of a strongly 
connected component (SCC), which can be defined as a MEC 
where condition (2) from Definition \ref{def:MEC} above is not 
required. It is also well-known that an SCC-decomposition of 
a directed graph can be done in linear time.

For our setting here, given a SNF-form OBMDP with 
its dependency graph $G = (U, E), U = V$, the partition of $U$ 
that we will use is the following: 
$U_P := \{T_i \in U \mid T_i$ is of \textsf{L}-form$\}$ and 
$U_1 := \{T_i \in U \mid T_i$ is of \textsf{M}-form or 
\textsf{Q}-form$\}$.

Before we continue with the algorithms, let us observe 
that the qualitative multi-target reachability problems are 
in general NP-hard (coNP-hard), if the size of the set $K$ of 
target non-terminals is not bounded by a fixed constant.

\begin{proposition}
$ $
	\begin{enumerate}[label=(\arabic*.)]
		\item The following two problems are 
both NP-hard: given an OBMDP, a set $K \subseteq [n]$ of 
target non-terminals and a starting non-terminal $T_i \in V$, 
decide whether: (i) 
$\exists \sigma \in \Psi: Pr_{T_i}^\sigma[\bigcap_{q \in K} 
Reach(T_q)] = 1$, and (ii) 
$Pr_{T_i}^*[\bigcap_{q \in K} Reach(T_q)] = 1$.

		\item The following problem is coNP-hard: 
given an OBP (i.e., an OBMDP with no controlled non-terminals, 
and hence with only one trivial strategy $\sigma$), 
a set $K \subseteq [n]$ of 
target non-terminals and a starting non-terminal $T_i \in V$, 
decide whether $Pr_{T_i}^\sigma[\bigcap_{q \in K} Reach(T_q)] = 0$.
	\end{enumerate}
	\label{prop:NP-hardness}
\end{proposition}

\begin{proof}
    For (1.) we reduce from 3-SAT, and for (2.) from 
the complement problem (i.e., deciding unsatisfiability 
of a 3-CNF formula). The reductions are nearly identical, 
so we describe them both together. 
Consider a 3-CNF formula over variables $\{x_1,\ldots,x_n\}$:
\begin{align*}
  \bigwedge_{q \in [m]} (l_{q,1} \vee l_{q,2} \vee l_{q,3})
\end{align*}
where every $l_{q,j}$ is either $x_r$ or $\neg x_r$ 
for some $r \in [n]$. We construct an OBMDP as follows: 
to each clause $q \in [m]$ we associate 
a target non-terminal $R_q$ with a single associated rule 
$R_q \xrightarrow{1} \varnothing$; 
for each variable $x_r, r \in [n]$, we associate two purely 
probabilistic non-terminals $T_{r_a}, T_{r_b}$, and

\begin{itemize}
    \item for (1.), a controlled non-terminal $C_r$ with rules 
$C_r \xrightarrow{a} T_{r_a}$ and $C_r \xrightarrow{b} T_{r_b}$, 
or

    \item for (2.), a probabilistic non-terminal $C_r$ with 
rules $C_r \xrightarrow{1/2} T_{r_a}$ and 
$C_r \xrightarrow{1/2} T_{r_b}$.
\end{itemize}
For each non-terminal $T_{r_a}$, $r \in [n]$ we would 
in principle like to create a single rule, 
with probability $1$, whose RHS consists of the following 
non-terminals (in any order): 
$\{R_q \mid \exists j \in \{1,2,3\} \; s.t. \; l_{q,j} = x_r\}$, 
as well as the non-terminal $C_{r+1}$ if $r < n$; 
likewise, for each non-terminal $T_{r_b}$, $r \in [n]$, 
we would like to create a single rule, with probability $1$, 
whose RHS consists of 
$\{R_q \mid \exists j \in \{1,2,3\} \; s.t. \; l_{q,j} = 
\neg x_r\}$, as well as $C_{r+1}$ if $r < n$.

However, due to the simple normal form we have adopted 
in our definition of OBMDPs, such rules need to be 
``expanded'' (as shown in Proposition \ref{prop:SNF-form}) 
into a sequence of rules whose RHS has length $\leq 2$, 
using auxiliary non-terminals. 
So, for example, instead of a single rule of the form 
$T_{1_b} \stackrel{1}{\rightarrow} R_2 R_3 R_4  C_2$, 
we will have the following rules (using auxiliary 
non-terminals $T_{1_b}^j$): 
$T_{1_b} \xrightarrow{1} T_{1_b}^1 \; C_2$, 
$T_{1_b}^1 \xrightarrow{1} T_{1_b}^2 \; R_2$, 
and $T_{1_b}^2 \xrightarrow{1} R_3 \; R_4$. 
See Figure \ref{fig:NP-hardness-construction-example} 
for an example.

\begin{figure}[!ht]
    \begin{align*}
    & C_1 \xrightarrow{a} T_{1_a} && T_{1_a}
\xrightarrow{1} R_1 \; C_2 && C_2 \xrightarrow{a} T_{2_a} &&
T_{2_a} \xrightarrow{1} T_{2_a}^1 \; C_3 && C_3 \xrightarrow{a} T_{3_a} && T_{3_a} \xrightarrow{1} R_1 \; R_3 \\
    & C_1 \xrightarrow{b} T_{1_b} && T_{1_b} \xrightarrow{1}
T_{1_b}^1 \; C_2 && C_2 \xrightarrow{b} T_{2_b} && T_{2_a}^1
\xrightarrow{1} R_2 \; R_3 && C_3 \xrightarrow{b} T_{3_b} &&
T_{3_b} \xrightarrow{1} R_2 \; R_4 \\
    & && T_{1_b}^1 \xrightarrow{1} T_{1_b}^2 \; R_2 && && T_{2_b}
\xrightarrow{1} T_{2_b}^1 \; C_3 \\
    & && T_{1_b}^2 \xrightarrow{1} R_3 \; R_4 && && T_{2_b}^1
\xrightarrow{1} R_1 \; R_4
    \end{align*}
    \caption{Reduction example: an OBMDP obtained from 
the 3-SAT formula $(x_1 \vee \neg x_2 \vee x_3) \wedge 
(\neg x_1 \vee x_2 \vee \neg x_3) \wedge 
(\neg x_1 \vee x_2 \vee x_3) \wedge 
(\neg x_1 \vee \neg x_2 \vee \neg x_3)$. 
This construction is for problem (1.); 
the construction for problem (2.) is very similar, with 
the controlled non-terminals $C_r, r \in [n]$ changed 
to purely probabilistic non-terminals instead 
(with $1/2$ probability on each of their two rules).}
    \label{fig:NP-hardness-construction-example}
\end{figure}

This reduction closely resembles a well-known 
reduction (\cite[Theorem 3.5]{siscla85}) 
for NP-hardness of model checking  eventuality formulas 
in linear temporal logic. 
The immediate children of the branching non-terminals 
$T_{r_a}$ and $T_{r_b}$ keep track of which clauses 
are satisfied under each of the two truth assignments 
to the variable $x_r$ (`true' corresponds to $T_{r_a}$, 
and `false' corresponds to $T_{r_b}$). 
In fact, for the OBMDP obtained for problem (1.), 
there is a one-to-one correspondence between 
truth assignments to all variables of the formula and 
deterministic static strategies.

It follows that, for the OBMDP in statement (1.), 
if there exists a satisfying truth assignment 
for the formula, then starting at non-terminal $C_1$, 
there exists a (deterministic and static) strategy 
$\sigma'$ for the player such that 
$Pr_{C_1}^{\sigma'}[\bigcap_{q \in [m]} Reach(R_q)] = 1$.

Otherwise, if the formula is unsatisfiable, 
then we claim that $\forall \sigma \in \Psi:\; 
Pr_{C_1}^{\sigma}[\bigcap_{q \in [m]} Reach(R_q)] = 0$. 
(And hence, that $Pr_{C_1}^{*}[\bigcap_{q \in [m]} Reach(R_q)] 
= 0 < 1$.) To see this, note that an arbitrary 
(possibly randomized, and not necessarily static) strategy 
in the constructed OBMDP corresponds to a 
(possibly correlated) probability distribution on 
assignments of truth values to the variables in 
the corresponding formula. (The distribution may be 
correlated, because the strategy may be non-static, 
but this doesn't matter.) So if the formula is unsatisfiable, 
then under any strategy for the player (i.e., 
any probability distribution on assignments of truth values), 
there is probability $0$ that the generated play tree 
contains all target non-terminals (respectively, 
that the random truth assignment satisfies all clauses 
in the formula).

For the OBP obtained for problem (2.), it follows from 
the same arguments that the formula is unsatisfiable 
if and only if 
$Pr_{C_1}^{\sigma}[\bigcap_{q \in [m]} Reach(R_q)] = 0$ 
(where $\sigma$ is just the trivial strategy, since 
there are no controlled non-terminals in 
the OBP obtained for (2.)).
\end{proof}

\section{Algorithm for deciding 
$\max_{\sigma} Pr_{T_i}^{\sigma}[\bigcap_{q \in K} Reach(T_q)] \stackrel{?}{=} 0$}

In this section we present an algorithm that, given an OBMDP 
and a set $K \subseteq [n]$ of $k = |K|$ target non-terminals, 
computes, for every subset of target non-terminals 
$K' \subseteq K$, the set $Z_{K'} \subseteq V$ of non-terminals 
such that, starting at a non-terminal $T_i \in Z_{K'}$, 
using any strategy $\sigma$, the probability 
that the generated play contains a copy of every non-terminal 
in $K'$ is $0$. In other words, the algorithm 
(Figure \ref{fig:QualMoR-NonReach}) computes, 
$\forall K' \subseteq K$, the set $Z_{K'} := \{T_i \in V \mid 
\forall \sigma \in \Psi: \; Pr_{T_i}^{\sigma}[\bigcap_{q \in K'} 
Reach(T_q)] = 0\}$. The algorithm uses as a preprocessing step 
an algorithm from \cite[Proposition 4.1]{ESY-icalp15-IC}. 
Namely, let us denote by $W_q$ the set $\{T_q\} \cup 
\{T_i \in V \mid \exists \sigma \in \Psi:\; 
Pr_{T_i}^{\sigma}[Reach(T_q)] > 0\}$. We can compute, 
for each $q \in K$, the set $W_q$ in P-time using the algorithm 
from \cite[Proposition 4.1]{ESY-icalp15-IC}, 
together with a single deterministic static witness strategy 
for every non-terminal in $W_q$. 
Let $K'_{-i}$ denote the set $K' - \{i\}$.

\begin{figure}[ht!]
	\begin{enumerate}[label=\Roman*.]
		\item Initialize $\bar{Z}_{\{q\}} := W_q$, 
		for each $q \in K$. Let $\bar{Z}_{\emptyset} := V$.
		
		\item For $l = 2 \ldots k$: \\
				\hspace*{0.1in} For every subset of 
				target non-terminals $K' \subseteq K$ of 
				size $|K'| = l$:
		\begin{enumerate}[label=\arabic*.]
			\item Initialize $\bar{Z}_{K'} := \big\{T_i \in V \mid$  one of the following holds:
			\begin{itemize}
				\item[-] $T_i$ is of \textsf{L}-form where $i \in K'$ and $T_i \rightarrow T_j,\; T_j \in \bar{Z}_{K'_{-i}}$.
				
				\item[-] $T_i$ is of \textsf{M}-form where $i \in K'$ and $\exists a' \in \Gamma^i:\; T_i \xrightarrow{a'} T_j,\; T_j \in \bar{Z}_{K'_{-i}}$.
				
				\item[-] $T_i$ is of \textsf{Q}-form ($T_i \xrightarrow{1} T_j \; T_r$) where $i \in K'$ and $\exists K_L \subseteq K'_{-i}:\; T_j \in \bar{Z}_{K_L} \wedge T_r \in \bar{Z}_{K'_{-i} - K_L}$.
				
				\item[-] $T_i$ is of \textsf{Q}-form ($T_i \xrightarrow{1} T_j \; T_r$) and $\exists K_L \subset K' \; (K_L \not= \emptyset): T_j \in \bar{Z}_{K_L} \wedge T_r \in \bar{Z}_{K' - K_L}. \}$
			\end{itemize}
			
			\item Repeat until no change has occurred 
			to $\bar{Z}_{K'}$:
			\begin{enumerate}[label=(\alph*)]
				\item add $T_i \not\in \bar{Z}_{K'}$ to $\bar{Z}_{K'}$, if of \textsf{L}-form and $T_i \rightarrow T_j,\; T_j \in \bar{Z}_{K'}$.
				
				\item add $T_i \not\in \bar{Z}_{K'}$ to $\bar{Z}_{K'}$, if of \textsf{M}-form and $\exists a' \in \Gamma^i:\; T_i \xrightarrow{a'} T_j,\; T_j \in \bar{Z}_{K'}$.
				
				\item add $T_i \not\in \bar{Z}_{K'}$ to $\bar{Z}_{K'}$, if of \textsf{Q}-form ($T_i \xrightarrow{1} T_j \; T_r$) and $T_j \in \bar{Z}_{K'} \vee T_r \in \bar{Z}_{K'}$.
			\end{enumerate}
			
			\item $Z_{K'} := V - \bar{Z}_{K'}$.
		\end{enumerate}
	\end{enumerate}
	\caption{Algorithm for computing set 
	$\{T_i \in V \mid \forall \sigma \in \Psi:\; 
	Pr_{T_i}^{\sigma}[\bigcap_{q \in K'} Reach(T_q)] = 0\}$ 
	for every subset of target non-terminals $K' \subseteq K$ 
	in a given OBMDP.}
	\label{fig:QualMoR-NonReach}
\end{figure}

\begin{proposition}
    The algorithm in Figure \ref{fig:QualMoR-NonReach} computes, 
given an OBMDP, $\mathcal{A}$, and a set $K \subseteq [n]$ 
of $k = |K|$ target non-terminals, for every subset of 
target non-terminals $K' \subseteq K$, the set 
$Z_{K'} := \{T_i \in V \mid \forall \sigma \in \Psi:\; 
Pr_{T_i}^{\sigma}[\bigcap_{q \in K'} Reach(T_q)] = 0\}$. 
The algorithm runs in time $4^k \cdot |\mathcal{A}|^{O(1)}$.
The algorithm can also be augmented to compute 
a deterministic (non-static) strategy $\sigma'_{K'}$ 
and a rational value $b_{K'} > 0$, such that 
for all $T_i \not\in Z_{K'}$, 
$Pr_{T_i}^{\sigma'_{K'}}[\bigcap_{q \in K'} Reach(T_q)] \ge b_{K'} > 0$.
	\label{prop:QualMoR-NonReach}
\end{proposition}

\begin{proof}
    The running time of the algorithm follows from the facts that 
step II. executes for $2^k$ iterations and 
inside each iteration, step II.1. requires time at most 
$2^k \cdot |\mathcal{A}|^{O(1)}$ and the loop at 
step II.2. executes in time at most $|\mathcal{A}|^{O(1)}$.

We need to prove that for every $K' \subseteq K: T_i \in Z_{K'}$ 
if and only if $\forall \sigma \in \Psi:\; 
Pr_{T_i}^{\sigma}[\bigcap_{q \in K'} Reach(T_q)] = 0 
\Leftrightarrow Pr_{T_i}^{\sigma}[\bigcup_{q \in K'} 
Reach^{\complement}(T_q)] = 1$ (or equivalently, 
that $T_i \in \bar{Z}_{K'}$ 
if and only if $\exists \sigma'_{K'} \in \Psi: Pr_{T_i}^{\sigma'_{K'}}[\bigcap_{q \in K'} Reach(T_q)] > 0$). 
We in fact show that there is a value $b_{K'} > 0$ and a 
strategy $\sigma'_{K'} \in \Psi$ such that 
$T_i \in \bar{Z}_{K'}$ if and only if 
$Pr_{T_i}^{\sigma'_{K'}}[\bigcap_{q \in K'} Reach(T_q)] \ge b_{K'}$. 
We analyse this by a double induction with the top-layer 
induction based on the size of set $K'$, or in other words 
the time of constructing set $\bar{Z}_{K'}$. 
Clearly for the base case (step I.) of a single 
target non-terminal $T_q, q \in K$, by the P-time algorithm 
from \cite[Proposition 4.1]{ESY-icalp15-IC}, 
there is a (deterministic static) strategy $\sigma'_{\{q\}}$ 
for the player and a value $b_{\{q\}} > 0$ where 
$T_i \in \bar{Z}_{\{q\}}$ if and only if 
$Pr_{T_i}^{\sigma'_{\{q\}}}[Reach^{\complement}(T_q)] 
\le 1 - b_{\{q\}} < 1 \Leftrightarrow 
Pr_{T_i}^{\sigma'_{\{q\}}}[Reach(T_q)] \ge b_{\{q\}} > 0$. 
Now, constructing set $\bar{Z}_{K'}$ for a subset 
$K' \subseteq K$ of target non-terminals of size $l$, 
assume that for each $K'' \subset K'$ of size $\le l - 1$, 
there is a strategy $\sigma'_{K''}$ for the player 
and a value $b_{K''} > 0$ such that for all 
$T_j \in \bar{Z}_{K''},\; Pr_{T_j}^{\sigma'_{K''}}[\bigcup_{q \in K''} Reach^{\complement}(T_q)] \le 1 - b_{K''} < 1 \Leftrightarrow Pr_{T_j}^{\sigma'_{K''}}[\bigcap_{q \in K''} Reach(T_q)] \ge b_{K''} > 0$. And for all $T_j \in Z_{K''}$, 
it holds that $\forall \sigma \in \Psi: Pr_{T_j}^{\sigma}[\bigcap_{q \in K''} Reach(T_q)] = 0$.
	
First, let us prove the direction where if $T_i \in \bar{Z}_{K'}$, 
then $\exists \sigma'_{K'} \in \Psi:\;$ $Pr_{T_i}^{\sigma'_{K'}} [\bigcup_{q \in K'} Reach^{\complement}(T_q)] \le 1 - b_{K'} < 1 \Leftrightarrow Pr_{T_i}^{\sigma'_{K'}}[\bigcap_{q \in K'} Reach(T_q)] \ge b_{K'} > 0$, 
for some value $b_{K'} > 0$. We use a second (nested) induction, 
based on the iteration in which non-terminal $T_i$ 
was added to set $\bar{Z}_{K'}$. Consider the base case 
where $T_i$ is a non-terminal added to set $\bar{Z}_{K'}$ 
at the initialization step II.1.

\begin{enumerate}[label=(\roman*)]
	\item Suppose $T_i$ is of \textsf{L}-form where $i \in K'$ 
(i.e., $T_i$ is a target non-terminal in set $K'$) and 
$T_i \rightarrow T_j,\; T_j \in \bar{Z}_{K'_{-i}}$, where 
$\exists \sigma'_{K'_{-i}} \in \Psi: Pr_{T_j}^{\sigma'_{K'_{-i}}}[\bigcap_{q \in K'_{-i}} Reach(T_q)] \ge b_{K'_{-i}}$, for some value $b_{K'_{-i}} > 0$. 
Due to the fact that the play up to (and including) a copy 
of non-terminal $T_i, i \in K'$ has already reached the 
target $T_i$ and using strategy $\sigma'_{K'_{-i}}$ from the 
next generation as if the play starts in it, it follows 
that there exists a strategy $\sigma'_{K'}$ such that, 
for an ancestor history $h := T_i(u, T_j)$:
\begin{align*}
    & Pr_{T_i}^{\sigma'_{K'}} \Big[ \bigcap_{q \in K'} Reach(T_q) \Big]
= Pr_{T_i}^{\sigma'_{K'}} \Big[ \bigcap_{q \in K'_{-i}} Reach(T_q) 
\; \Big\vert \; Reach(T_i) \Big] \cdot Pr_{T_i}^{\sigma'_{K'}} 
\Big[ Reach(T_i) \Big] \\
    & = Pr_{T_i}^{\sigma'_{K'}} \Big[ \bigcap_{q \in K'_{-i}} 
Reach(T_q) \Big] \ge p_{ij} \cdot Pr_h^{\sigma'_{K'}} 
\Big[ \bigcap_{q \in K'_{-i}} Reach(T_q) \Big] \\
    & = p_{ij} \cdot Pr_{T_j}^{\sigma'_{K'_{-i}}} \Big[ \bigcap_{q \in K'_{-i}} Reach(T_q) \Big] \ge p_{ij} \cdot b_{K'_{-i}} > 0
\end{align*}
where $p_{ij} > 0$ is the probability of the rule 
$T_i \xrightarrow{p_{ij}} T_j$. 
So let $b_{K'}^i := p_{ij} \cdot b_{K'_{-i}}$.

	\item Suppose $T_i$ is of \textsf{M}-form where $i \in K'$ 
and $\exists a' \in \Gamma^i:\; T_i \xrightarrow{a'} T_j,\; T_j \in \bar{Z}_{K'_{-i}}$. Again let $h := T_i(u, T_j)$. 
By combining the witness strategy $\sigma'_{K'_{-i}}$ 
from the induction assumption for a starting non-terminal $T_j$ 
with the initial local choice of 
choosing deterministically action $a'$ starting at a 
non-terminal $T_i$, we obtain a combined strategy $\sigma'_{K'}$, 
such that starting at a (target) non-terminal $T_i$, 
we satisfy $Pr_{T_i}^{\sigma'_{K'}}[\bigcap_{q \in K'} Reach(T_q)] = Pr_h^{\sigma'_{K'}}[\bigcap_{q \in K'_{-i}} Reach(T_q)] = Pr_{T_j}^{\sigma'_{K'_{-i}}}[\bigcap_{q \in K'_{-i}} Reach(T_q)] \ge b_{K'_{-i}} > 0$. So let $b_{K'}^i := b_{K'_{-i}}$.

	\item Suppose $T_i$ is of \textsf{Q}-form 
(i.e., $T_i \xrightarrow{1} T_j \; T_r$) and there exists 
a proper split of the target non-terminals from $K'$, 
implied by $K_L \subset K'$ (where $K_L \not= \emptyset$) and 
$K' - K_L$, such that 
$T_j \in \bar{Z}_{K_L} \wedge T_r \in \bar{Z}_{K' - K_L}$. 
So, by the inductive assumption, 
$\exists \sigma'_{K_L} \in \Psi: Pr_{T_j}^{\sigma'_{K_L}}[\bigcap_{q \in K_L} Reach(T_q)] \ge b_{K_L} > 0$ 
and $\exists \sigma'_{K' - K_L} \in \Psi: Pr_{T_r}^{\sigma'_{K' - K_L}}[\bigcap_{q \in K' - K_L} Reach(T_q)] \ge b_{K' - K_L} > 0$, 
for some values $b_{K_L}, b_{K' - K_L} > 0$. 
Let $h_l := T_i(l, T_j)$ and $h_r := T_i(r, T_r)$. 
Hence, by combining the two strategies $\sigma'_{K_L}$ and 
$\sigma'_{K' - K_L}$ to be used from the next generation 
from the left and right child, respectively, as if the play 
starts in them, it follows that 
$\exists \sigma'_{K'} \in \Psi: 
Pr_{T_i}^{\sigma'_{K'}}[\bigcap_{q \in K'} Reach(T_q)] \ge 
Pr_{h_l}^{\sigma'_{K'}}[\bigcap_{q \in K_L} Reach(T_q)] \cdot 
Pr_{h_r}^{\sigma'_{K'}}[\bigcap_{q \in K' - K_L} Reach(T_q)] = 
Pr_{T_j}^{\sigma'_{K_L}}[\bigcap_{q \in K_L} Reach(T_q)] 
\cdot Pr_{T_r}^{\sigma'_{K' - K_L}}[\bigcap_{q \in K' - K_L} 
Reach(T_q)] \ge b_{K_L} \cdot b_{K' - K_L} > 0$, and so let 
$b_{K'}^i := b_{K_L} \cdot b_{K' - K_L}$.
	
	\item Suppose $T_i$ is of \textsf{Q}-form 
(i.e., $T_i \xrightarrow{1} T_j \; T_r$) where $i \in K'$ and 
there exists a split of the target non-terminals from set $K'_{-i}$, 
implied by $K_L \subseteq K'_{-i}$ and $K'_{-i} - K_L$, such that 
$T_j \in \bar{Z}_{K_L} \wedge T_r \in \bar{Z}_{K'_{-i} - K_L}$. 
Combining in the same way as in (iii) above 
the two witness strategies from 
the induction assumption for non-terminals $T_j$ and $T_r$, and 
the fact that the play starts in the target non-terminal $T_i, i \in K'$, 
it follows that $\exists \sigma'_{K'} \in \Psi: 
Pr_{T_i}^{\sigma'_{K'}}[\bigcap_{q \in K'} Reach(T_q)] = 
Pr_{T_i}^{\sigma'_{K'}}[\bigcap_{q \in K'_{-i}} Reach(T_q)] \ge 
Pr_{T_j}^{\sigma'_{K_L}}[\bigcap_{q \in K_L} Reach(T_q)] \cdot 
Pr_{T_r}^{\sigma'_{K'_{-i} - K_L}}[\bigcap_{q \in K'_{-i} - K_L} 
Reach(T_q)] \ge b_{K_L} \cdot b_{K'_{-i} - K_L} > 0$, and 
so let $b_{K'}^i := b_{K_L} \cdot b_{K'_{-i} - K_L}$.
\end{enumerate}

Now consider the inductive step of the nested induction, 
i.e., non-terminals $T_i$ added to set $\bar{Z}_{K'}$ at step II.2. 
If $T_i$ is of \textsf{L}-form, 
then for a non-terminal $T_i$ there is a positive probability 
of generating a child of a non-terminal $T_j \in \bar{Z}_{K'}$, 
for which we already know that $\exists \sigma'_{K'} \in \Psi:\; 
Pr_{T_j}^{\sigma'_{K'}}[\bigcap_{q \in K'} Reach(T_q)] 
\ge b_{K'}^j > 0$. Let $h := T_i(u, T_j)$.
Using the strategy $\sigma'_{K'}$ in the next generation as if the play 
starts in it, we get an augmented strategy $\sigma'_{K'}$, such 
that $Pr_{T_i}^{\sigma'_{K'}} [\bigcap_{q \in K'} Reach(T_q)] 
\ge p_{ij} \cdot Pr_h^{\sigma'_{K'}} [\bigcap_{q \in K'} Reach(T_q)] 
= p_{ij} \cdot Pr_{T_j}^{\sigma'_{K'}} [\bigcap_{q \in K'} 
Reach(T_q)] \ge p_{ij} \cdot b_{K'}^j > 0$, 
where $p_{ij} > 0$ is the probability 
of the rule $T_i \xrightarrow{p_{ij}} T_j$. 
Let $b_{K'}^i := p_{ij} \cdot b_{K'}^j$.

If type $T_i$ is of \textsf{M}-form, then 
$\exists a' \in \Gamma^i:\; T_i \xrightarrow{a'} T_j,\; 
T_j \in \bar{Z}_{K'}$, where $\exists \sigma'_{K'} \in \Psi$ 
such that $Pr_{T_j}^{\sigma'_{K'}}[\bigcap_{q \in K'} Reach(T_q)] \ge b_{K'}^j > 0$. Again let $h := T_i(u, T_j)$.
Hence, by combining the witness strategy $\sigma'_{K'}$ 
for a starting non-terminal $T_j$ (from the nested 
induction assumption) with the initial local choice of 
choosing deterministically action $a'$ starting at a 
non-terminal $T_i$, we obtain an augmented strategy 
$\sigma'_{K'}$ for a starting non-terminal $T_i$, such that 
$Pr_{T_i}^{\sigma'_{K'}}[\bigcap_{q \in K'} Reach(T_q)] = 
Pr_h^{\sigma'_{K'}}[\bigcap_{q \in K'} Reach(T_q)] = 
Pr_{T_j}^{\sigma'_{K'}}[\bigcap_{q \in K'} Reach(T_q)] \ge 
b_{K'}^i > 0$, where let $b_{K'}^i := b_{K'}^j$.

If type $T_i$ is of \textsf{Q}-form 
(i.e., $T_i \xrightarrow{1} T_j \; T_r$), then 
$T_j \in \bar{Z}_{K'} \vee T_r \in \bar{Z}_{K'}$, and so 
$\exists \sigma'_{K'} \in \Psi:\; Pr_{T_y}^{\sigma'_{K'}} 
[\bigcup_{q \in K'} Reach^{\complement}(T_q)] \le 1 - b_{K'}^y < 1$, 
where $y \in \{j, r\}$. Let $h_y := T_i(x, T_y)$ and 
$h_{\bar{y}} := T_i(\bar{x}, T_{\bar{y}})$, 
where $\bar{y} \in \{j, r\} - \{y\}$, 
$x \in \{l, r\}$ and $\bar{x} \in \{l, r\} - \{x\}$. 
By augmenting this $\sigma'_{K'}$ to be used from 
the next generation from the child of non-terminal $T_y$ 
as if the play starts in it and using an arbitrary strategy 
from the child of non-terminal $T_{\bar{y}}$, it holds that 
$Pr_{T_i}^{\sigma'_{K'}} [\bigcup_{q \in K'} Reach^{\complement}(T_q)] 
\le Pr_{h_y}^{\sigma'_{K'}}[\bigcup_{q \in K'} Reach^{\complement}(T_q)] 
\cdot Pr_{h_{\bar{y}}}^{\sigma'_{K'}}[\bigcup_{q \in K'} 
Reach^{\complement}(T_q)] \le Pr_{T_y}^{\sigma'_{K'}} 
[\bigcup_{q \in K'} Reach^{\complement}(T_q)] \le 1 - b_{K'}^i < 1$, 
where let $b_{K'}^i := b_{K'}^y$.

Finally, let $b_{K'} := \min_{T_i \in \bar{Z}_{K'}} \{b_{K'}^i\}$. 

Clearly, the constructed non-static strategy $\sigma'_{K'}$ 
can be described in time $4^k \cdot |\mathcal{A}|^{O(1)}$.

\medskip \medskip
Secondly, let us show the opposite direction, i.e., 
where if non-terminal $T_i \in Z_{K'}$, then 
$\forall \sigma \in \Psi:\; 
Pr_{T_i}^{\sigma}[\bigcap_{q \in K'} Reach(T_q)] = 0$. 
For all non-terminals 
$T_i \in Z_{K'}$, for a copy of non-terminal $T_i$ in the play, 
it holds that: if $T_i$ is of \textsf{L}-form, only a child 
of a non-terminal in $Z_{K'}$ can be generated; if $T_i$ 
is of \textsf{M}-form, regardless of player's choice on 
actions $\Gamma^i$, similarly only a child of a non-terminal 
in $Z_{K'}$ is generated as an offspring; 
if $T_i$ is of \textsf{Q}-form, both children have non-terminals 
belonging to $Z_{K'}$. This is due to non-terminals not being 
added to set $\bar{Z}_{K'}$ at step II.2.

Fix an arbitrary strategy $\sigma$ for the player. 
Then starting at a non-terminal $T_i \in Z_{K'}$ 
and under $\sigma$, the generated tree can contain only 
copies of non-terminals in set $Z_{K'}$, i.e., 
the play stays confined to non-terminals from set $Z_{K'}$ 
(note that the play may terminate). What is more, there is 
\textit{no} \textsf{Q}-form non-terminal $T_i$ in $Z_{K'}$ 
(whether $T_i$ is a target from $K'$ or not) such that 
non-terminal $T_i$ splits the job, of reaching 
the target non-terminals from set $K'$, amongst its two children. 
In other words, for each \textsf{Q}-form non-terminal 
$T_i \in Z_{K'}$ (i.e., $T_i \xrightarrow{1} T_j \; T_r$), 
$\forall K_L \subset K'$ (where $K_L \not= \emptyset$): 
$T_j \in Z_{K_L} \vee T_r \in Z_{K' - K_L}$; and 
if $T_i$ happens to be 
a target non-terminal itself from set $K'$ (i.e., $i \in K'$), 
then $\forall K_L \subseteq K'_{-i}: T_j \in Z_{K_L} \vee T_r \in Z_{K'_{-i} - K_L}$ (this is due to non-terminal $T_i$ not added 
to set $\bar{Z}_{K'}$ at step II.1.). 
So the only possibility, under $\sigma$ and starting at 
some non-terminal $T_i \in Z_{K'}$, to generate 
with a positive probability a tree (play) that contains copies 
of all targets from set $K'$, is (1) if all target non-terminals 
from $K'$ were never added to set $\bar{Z}_{K'}$ and, thus, 
belong to set $Z_{K'}$, and (2) if it is, in fact, some path $w$ 
(starting at the root) in the generated tree that contains 
copies of all the target non-terminals from set $K'$. 
Consider such a path $w$ and the very first copy $o$ of 
any of the target non-terminals $T_q, q \in K'$ along path $w$. 
Let $o$ be of a \textsf{L}-form target non-terminal $T_v$, 
let $o'$ be the successor child of $o$ along the path $w$ 
(say of some non-terminal $T_j$), and let $h$ be 
an ancestor history that follows along path $w$ up until 
(and including) $o'$ and ends in $o'$ 
(i.e., $\current(h) = T_j$). 
Then it follows that 
$Pr_h^{\sigma}[\bigcap_{q \in K'_{-v}} Reach(T_q)] > 0$. 
But it is easy to see that from $\sigma$ one can easily 
construct a strategy $\sigma'_{K'_{-v}}$ such that 
$Pr_{T_j}^{\sigma'_{K'_{-v}}}[\bigcap_{q \in K'_{-v}} Reach(T_q)] > 0$, i.e., $T_j \in \bar{Z}_{K'_{-v}}$.
But this contradicts the fact that the \textsf{L}-form 
non-terminal $T_v$ hasn't been added to 
set $\bar{Z}_{K'}$ at step II.1. 
Similarly follows the argument for if 
$T_v$ is of \textsf{M}-form or \textsf{Q}-form. 

So for all non-terminals $T_i \in Z_{K'}$, regardless of 
strategy $\sigma$ for the player, there is a zero probability 
of generating a tree that contains all target non-terminals 
from set $K'$ 
(i.e., $\forall \sigma \in \Psi:\; Pr_{T_i}^{\sigma} 
[\bigcap_{q \in K'} Reach(T_q)] = 0$). That concludes the proof.
\end{proof}

\section{Algorithm for deciding 
$Pr_{T_i}^*[\bigcap_{q \in K} 
Reach(T_q)] \stackrel{?}{=} 1$}

In this section we present an algorithm for deciding, 
given an OBMDP, $\mathcal{A}$, a set $K \subseteq [n]$ of 
$k = |K|$ target non-terminals and a starting non-terminal $T_i$, 
whether $Pr_{T_i}^*[\bigcap_{q \in K} Reach(T_q)] := 
\sup_{\sigma \in \Psi} Pr_{T_i}^{\sigma}[\bigcap_{q \in K} 
Reach(T_q)] = 1$, i.e., the optimal probability of 
generating a play (tree) that contains all target non-terminals 
from set $K$ is $= 1$. Recall, from Example 1, 
that there need not be 
a strategy for the player that achieves probability 
exactly $1$, which is the question in the next section 
(almost-sure multi-target reachability). 
However, there may nevertheless 
be a sequence of strategies that achieve probabilities 
arbitrarily close to 1 (limit-sure multi-target reachability), 
and the question of the existence of such a sequence is 
what we address in this section. 
In other words, we are asking whether there exists a sequence 
of strategies $\langle \sigma^*_{\epsilon_j} \mid 
j \in \mathbb{N} \rangle$ such that $\forall j \in \mathbb{N}$, 
$\epsilon_j > \epsilon_{j+1} > 0$ 
(i.e., $\lim_{j \rightarrow \infty} \epsilon_j = 0$) and 
$Pr_{T_i}^{\sigma^*_{\epsilon_j}}[\bigcap_{q \in K} 
Reach(T_q)] \ge 1 - \epsilon_j$. 
The algorithm runs in time $4^k \cdot |\mathcal{A}|^{O(1)}$, and 
hence is fixed-parameter tractable with respect to $k$.

First, as a preprocessing step, for each subset of 
target non-terminals $K' \subseteq K$, we compute the set 
$Z_{K'} := \{T_i \in V \mid \forall \sigma \in \Psi:\; 
Pr_{T_i}^{\sigma}[\bigcap_{q \in K'} Reach(T_q)] = 0\}$, 
using the algorithm from Proposition 
\ref{prop:QualMoR-NonReach}. Let also denote by $AS_q$, 
for every $q \in K$, the set of non-terminals $T_j$ 
(including the target non-terminal $T_q$ itself) for which 
$Pr_{T_j}^*[Reach(T_q)] = 1$. These sets can be computed 
in P-time by applying the algorithm from 
\cite[Theorem 9.3]{ESY-icalp15-IC} to each target non-terminal 
$T_q, \; q \in K$. Recall that it was shown in 
\cite{ESY-icalp15-IC} that for OBMDPs with a single target 
the almost-sure and limit-sure reachability problems coincide. 
So in fact, for every $q \in K$, there exists a strategy $\tau_q$ 
such that for every $T_j \in AS_q:\; 
Pr_{T_j}^{\tau_q}[Reach(T_q)] = 1$.

After this preprocessing step, we apply the algorithm in Figure 
\ref{fig:QualMoR-LS-Reach} to identify the non-terminals $T_i$ 
for which $Pr_{T_i}^*[\bigcap_{q \in K} Reach(T_q)] = 1$. 
Again let $K'_{-i}$ denote the set $K' - \{i\}$.

\begin{figure}[hp!]
	\small
	\begin{enumerate}[label=\Roman*.]
		\item Let $F_{\{q\}} := AS_q$, for each $q \in K$. 
		$F_{\emptyset} := V$.
		
		\item For $l = 2 \ldots k$: \\
			\hspace*{0.1in} For every subset of target non-terminals 
			$K' \subseteq K$ of size $|K'| = l$:
		\begin{enumerate}[label=\arabic*.]
			\item $D_{K'} := \{T_i \in V - Z_{K'} \mid $ one of the 
			following holds:
			\begin{itemize}
				\item[-] $T_i$ is of \textsf{L}-form where 
				$i \in K'$, $T_i \not\rightarrow \varnothing$ and 
				$\forall T_j \in V$: if $T_i \rightarrow T_j$, 
				then $T_j \in F_{K'_{-i}}$. 
				
				\item[-] $T_i$ is of \textsf{M}-form where 
				$i \in K'$ and $\exists a^* \in \Gamma^i: 
				T_i \xrightarrow{a^*} T_j,\; T_j \in F_{K'_{-i}}$.
				
				\item[-] $T_i$ is of \textsf{Q}-form 
				($T_i \xrightarrow{1} T_j \; T_r$) where $i \in K'$ 
				and $\exists K_L \subseteq K'_{-i}: 
				T_j \in F_{K_L} \wedge T_r \in F_{K'_{-i} - K_L}$.
				
				\item[-] $T_i$ is of \textsf{Q}-form 
				($T_i \xrightarrow{1} T_j \; T_r$) where 
				$\exists K_L \subset K'\; (K_L \not= \emptyset): 
				T_j \in F_{K_L} \wedge T_r \in F_{K' - K_L}.\}$
			\end{itemize}
			
			\item Repeat until no change has occurred to $D_{K'}$:
			\begin{enumerate}[label=(\alph*)]
				\item add $T_i \not\in D_{K'}$ to $D_{K'}$, if of 
				\textsf{L}-form, $T_i \not\rightarrow \varnothing$ 
				and $\forall T_j \in V$: if $T_i \rightarrow T_j$, 
				then $T_j \in D_{K'}$.
				
				\item add $T_i \not\in D_{K'}$ to $D_{K'}$, if of 
				\textsf{M}-form and $\exists a^* \in \Gamma^i: 
				T_i \xrightarrow{a^*} T_j,\; T_j \in D_{K'}$.
				
				\item add $T_i \not\in D_{K'}$ to $D_{K'}$, if of 
				\textsf{Q}-form ($T_i \xrightarrow{1} T_j \; T_r$) 
				and $T_j \in D_{K'} \vee T_r \in D_{K'}$.
			\end{enumerate}
			
			\item Let $X := V - (D_{K'} \cup Z_{K'})$.
			
			\item Initialize $S_{K'} := \{T_i \in X \mid $ 
			either $i \in K'$, or $T_i$ is of \textsf{L}-form and 
			$T_i \rightarrow \varnothing \vee T_i \rightarrow T_j,\; 
			T_j \in Z_{K'} \} \cup \bigcup_{\emptyset \subset K'' \subset K'} (X \cap S_{K''})$.
			
			\item Repeat until no change has occurred to $S_{K'}$:
			\begin{enumerate}[label=(\alph*)]
        		\item add $T_i \in X - S_{K'}$ to $S_{K'}$, if of 
        		\textsf{L}-form and $T_i \rightarrow T_j,\; 
        		T_j \in S_{K'} \cup Z_{K'}$.
        			
            	\item add $T_i \in X - S_{K'}$ to $S_{K'}$, if of 
            	\textsf{M}-form and $\forall a \in \Gamma^i:\; 
            	T_i \xrightarrow{a} T_j,\; T_j \in S_{K'} \cup Z_{K'}$.
            		
            	\item add $T_i \in X - S_{K'}$ to $S_{K'}$, if of 
            	\textsf{Q}-form ($T_i \xrightarrow{1} T_j \; T_r$) 
            	and $T_j \in S_{K'} \cup Z_{K'}\; \wedge \; 
            	T_r \in S_{K'} \cup Z_{K'}$.
			\end{enumerate}
			
			\item $\mathcal{C} \leftarrow$ MEC decomposition of $G[X - S_{K'}]$.
			
			\item For every $q \in K'$, let 
			$H_q := \{T_i \in X - S_{K'} \mid T_i$ is of 
			\textsf{Q}-form ($T_i \xrightarrow{1} T_j \; T_r$) and 
			$( (T_j \in X - S_{K'} \wedge 
			T_r \in \bar{Z}_{\{q\}} ) \vee 
			(T_j \in \bar{Z}_{\{q\}} \wedge T_r \in X - S_{K'}) ) \}$.
			
			\item Let $F_{K'} := \bigcup \{C \in \mathcal{C} \mid 
			P_C = K' \vee (P_C \not= \emptyset \wedge P_C \not= K' 
			\wedge \exists T_i \in C, \exists a \in \Gamma^i:\; 
			T_i \xrightarrow{a} T_j,\; T_j \in F_{K' - P_C})\}$, 
			where $P_C = \{q \in K' \mid C \cap H_q \not= \emptyset\}$.
			
			\item Repeat until no change has occurred to $F_{K'}$:
			\begin{enumerate}[label=(\alph*)]
 	       		\item add $T_i \in X - (S_{K'} \cup F_{K'})$ to 
 	       		$F_{K'}$, if of \textsf{L}-form and 
 	       		$T_i \rightarrow T_j,\; T_j \in F_{K'} \cup D_{K'}$.
 	       		
    		   	\item add $T_i \in X - (S_{K'} \cup F_{K'})$ to 
    		   	$F_{K'}$, if of \textsf{M}-form and 
    		   	$\exists a^* \in \Gamma^i: 
    		   	T_i \xrightarrow{a^*} T_j,\; T_j \in F_{K'}$.
    		   	
    		   	\item add $T_i \in X - (S_{K'} \cup F_{K'})$ to 
    		   	$F_{K'}$, if of \textsf{Q}-form 
    		   	($T_i \xrightarrow{1} T_j \; T_r$) and 
    		   	$T_j \in F_{K'} \vee T_r \in F_{K'}$.
			\end{enumerate}
			
			\item If $X \not= S_{K'} \cup F_{K'}$, let $S_{K'} := X - F_{K'}$ and go to step 5.
			
			\item Else, i.e., if $X = S_{K'} \cup F_{K'}$, let $F_{K'} := F_{K'} \cup D_{K'}$.
		\end{enumerate}
		
		\item \textbf{Output} $F_K$.
	\end{enumerate}
	\caption{Algorithm for limit-sure 
	multi-target reachability.  The output is the set  
	$F_K = \{T_i \in V \mid 
	Pr_{T_i}^*[\bigcap_{q \in K} Reach(T_q)] = 1\}$.}
	\label{fig:QualMoR-LS-Reach}
\end{figure}

\begin{theorem}
	The algorithm in Figure \ref{fig:QualMoR-LS-Reach} computes, 
given an OBMDP, $\mathcal{A}$, 
and a set $K \subseteq [n]$ of $k = |K|$ target non-terminals, 
for each subset $K' \subseteq K$, 
the set of non-terminals $F_{K'} := \{T_i \in V \mid 
Pr_{T_i}^*[\bigcap_{q \in K'} Reach(T_q)] = 1\}$. 
The algorithm runs in time $4^k \cdot |\mathcal{A}|^{O(1)}$. 
Moreover, for each $K' \subseteq K$, given $\epsilon > 0$, 
the algorithm can also be augmented to compute 
a randomized non-static strategy 
$\sigma^{\epsilon}_{K'}$ such that 
$Pr_{T_i}^{\sigma^{\epsilon}_{K'}}[\bigcap_{q \in K'} 
Reach(T_q)] \ge 1 - \epsilon$ for all non-terminals $T_i \in F_{K'}$.
	\label{theorem:QualMoR-LS-Reach}
\end{theorem}

\begin{proof}
    We will refer to the loop executing steps II.5. 
through II.10. for a specific subset $K' \subseteq K$ as 
the ``inner" loop and the iteration through all subsets of $K$ 
as the ``outer" loop. Clearly the inner loop terminates, 
due to step II.10. always adding at least one non-terminal 
to set $S_{K'}$ and step II.11. eventually executing. 
The running time of the algorithm follows from the facts that 
the outer loop executes for $2^k$ iterations and 
inside each iteration of the outer loop, steps II.1. and II.4. 
require time at most $2^k \cdot |\mathcal{A}|^{O(1)}$ 
and the inner loop executes for at most $|V|$ iterations, 
where during each inner loop iteration the nested loops execute 
in time at most $|\mathcal{A}|^{O(1)}$.

For the proof of correctness, we show that for every subset 
of target non-terminals $K' \subseteq K$, $F_{K'}$ 
(from the decomposition $V = F_{K'} \cup S_{K'} \cup Z_{K'}$) is the 
set of non-terminals $T_i$ for which the following property holds:
\begin{center}
	$(A)_{K'}^i$: $\sup_{\sigma \in \Psi} 
	Pr_{T_i}^{\sigma}[\bigcap_{q \in K'} Reach(T_q)] = 
	Pr_{T_i}^*[\bigcap_{q \in K'} Reach(T_q)] = 1$, 
	i.e., $\forall \epsilon > 0,\; \exists \sigma_{K'}^{\epsilon}$ 
	such that $Pr_{T_i}^{\sigma_{K'}^{\epsilon}} 
	[\bigcap_{q \in K'} Reach(T_q)] \ge 1 - \epsilon$.
\end{center}
Otherwise, if $T_i \in S_{K'}$, then the following property holds:
\begin{center}
	$(B)_{K'}^i$: $\sup_{\sigma \in \Psi} 
	Pr_{T_i}^{\sigma}[\bigcap_{q \in K'} Reach(T_q)] < 1$, 
	i.e., there exists a value $g > 0$ such that 
	$\forall \sigma \in \Psi:\; 
	Pr_{T_i}^{\sigma}[\bigcap_{q \in K'} Reach(T_q)] 
	\le 1 - g$.
\end{center}
Clearly, for non-terminals $T_i \in Z_{K'}$, 
property $(B)_{K'}^i$ holds, since 
$\sup_{\sigma \in \Psi} Pr_{T_i}^{\sigma}[\bigcap_{q \in K'} 
Reach(T_q)] = 0 < 1$ (by Proposition \ref{prop:QualMoR-NonReach}). 
Finally, the answer for the full set of targets is $F := F_K$.

We base this proof 
on an induction on the size of subset $K'$, 
i.e., on the time of computing sets $S_{K'}$ and $F_{K'}$ for 
$K' \subseteq K$. 
And in the process, for each subset $K' \subseteq K$ of 
target non-terminals, we show how to construct a randomized 
\textit{non-static} strategy $\sigma_{K'}^{\epsilon}$ 
(for any given $\epsilon > 0$) that ensures 
$Pr_{T_i}^{\sigma_{K'}^{\epsilon}}[\bigcap_{q \in K'} 
Reach(T_q)] \ge 1 - \epsilon$ for each non-terminal 
$T_i \in F_{K'}$.

Clearly for any subset of target non-terminals, 
$K' := \{q\} \subseteq K$, of size $l = 1$, each non-terminal 
$T_i \in F_{\{q\}}$ (respectively, $T_i \in V - F_{\{q\}}$) 
satisfies property $(A)_{\{q\}}^i$ (respectively, $(B)_{\{q\}}^i$), 
due to step I. and the definition of the $AS_q, q \in K$ sets. 
Furthermore, for each such subset $\{q\} \subseteq K$, 
there is in fact a strategy $\sigma_{\{q\}}$ such that 
$\forall T_i \in F_{\{q\}}:\; 
Pr_{T_i}^{\sigma_{\{q\}}}[Reach(T_q)] = 1$. 
Moreover, by \cite[Theorem 9.4]{ESY-icalp15-IC}, this strategy $\sigma_{\{q\}}$ is non-static and deterministic. 
Analysing subset $K'$ of target non-terminals of size $l$ 
as part of step II., assume that, for every $K'' \subset K'$ 
of size $\le l - 1$, sets $S_{K''}$ and $F_{K''}$ 
have already been computed, and for each non-terminal $T_j$ 
belonging to set $F_{K''}$ (respectively, set $S_{K''}$) 
property $(A)_{K''}^j$ (respectively, $(B)_{K''}^j$) holds. 
That is, by induction assumption, for each $K'' \subset K'$, 
for every $\epsilon > 0$ there is a randomized non-static 
strategy $\sigma_{K''}^{\epsilon}$ such that for any 
$T_j \in F_{K''}$: $Pr_{T_j}^{\sigma_{K''}}[\bigcap_{q \in K''} 
Reach(T_q)] \ge 1 - \epsilon$, and also for any 
$T_j \in S_{K''}$: $\sup_{\sigma \in \Psi} 
Pr_{T_j}^{\sigma}[\bigcap_{q \in K''} Reach(T_q)] < 1$. 
We now need to show that at end of the inner loop analysis 
of subset $K'$, property $(A)_{K'}^i$ 
(respectively, $(B)_{K'}^i$) holds for every non-terminal 
$T_i \in F_{K'}$ (respectively, $T_i \in S_{K'}$).

First we show that property $(A)_{K'}^i$ holds for each 
non-terminal $T_i$ belonging to set $D_{K'}$ ($\subseteq F_{K'}$), 
pre-computed prior to the execution of the inner loop for $K'$.

\begin{lemma}
	Every non-terminal $T_i \in D_{K'}$ satisfies property $(A)_{K'}^i$.
	\label{lemma:D_K'-analysis-LS}
\end{lemma}

\begin{proof}
    The lemma is proved via a nested induction based on the time 
when a non-terminal is added to set $D_{K'}$. 
Consider the base case where $T_i \in D_{K'}$ is a non-terminal, 
added at the initialization step II.1.
	\begin{enumerate}[label=(\roman*)]
        \item Suppose $T_i$ is of \textsf{L}-form where 
$i \in K'$ and for all associated rules a child is generated 
that is of a non-terminal $T_j \in F_{K'_{-i}}$, where 
property $(A)_{K'_{-i}}^j$ holds. Then, for every 
$\epsilon > 0$, using the witness strategy 
$\sigma_{K'_{-i}}^{\epsilon}$ from the induction assumption 
for all such non-terminals $T_j$ in the next generation 
as if the play starts in it, we obtain a strategy 
$\sigma_{K'}^{\epsilon}$ for a starting (target) non-terminal $T_i$ 
such that:
        \begin{align*}
& Pr_{T_i}^{\sigma_{K'}^{\epsilon}} 
    \Big[ \bigcap_{q \in K'} Reach(T_q) \Big] = 
    Pr_{T_i}^{\sigma_{K'}^{\epsilon}} 
    \Big[ \bigcap_{q \in K'_{-i}} Reach(T_q) \; \Big\vert \; 
    Reach(T_i) \Big] \cdot Pr_{T_i}^{\sigma_{K'}^{\epsilon}} 
    \Big[  Reach(T_i) \Big]  \\
= & Pr_{T_i}^{\sigma_{K'}^{\epsilon}} \Big[ 
    \bigcap_{q \in K'_{-i}} Reach(T_q) \Big] = 
    \sum_{j} p_{ij} \cdot 
    Pr_{T_i(u, T_j)}^{\sigma_{K'}^{\epsilon}} 
    \Big[ \bigcap_{q \in K'_{-i}} Reach(T_q) \Big] \\
= & \sum_{j} p_{ij} \cdot Pr_{T_j}^{\sigma_{K'_{-i}}^{\epsilon}} 
    \Big[ \bigcap_{q \in K'_{-i}} Reach(T_q) \Big] \ge  
    \sum_{j} p_{ij} \cdot (1 - \epsilon) = 1 - \epsilon
        \end{align*}
where $p_{ij}$ is the probability of rule $T_i \xrightarrow{p_{ij}} T_j$.

		\item Suppose $T_i$ is of \textsf{M}-form where 
$i \in K'$ and $\exists a^* \in \Gamma^i$ such that 
$T_i \xrightarrow{a^*} T_j,\; T_j \in F_{K'_{-i}}$, 
where property $(A)_{K'_{-i}}^j$ holds. 
Let $h := T_i(u, T_j)$. By combining every 
witness strategy $\sigma_{K'_{-i}}^{\epsilon},\; \epsilon > 0$ 
from property $(A)_{K'_{-i}}^j$ from the induction assumption 
for non-terminal $T_j$, as if the play starts in it, 
with the initial local choice 
of choosing action $a^*$ deterministically starting 
at a non-terminal $T_i$, we obtain for every $\epsilon > 0$ 
a combined strategy $\sigma_{K'}^{\epsilon}$ such that 
starting at a (target) non-terminal $T_i$, it follows that
$Pr_{T_i}^{\sigma_{K'}^{\epsilon}}[\bigcap_{q \in K'} 
Reach(T_q)] = Pr_h^{\sigma_{K'}^{\epsilon}}[\bigcap_{q \in K'_{-i}} 
Reach(T_q)] = Pr_{T_j}^{\sigma_{K'_{-i}}^{\epsilon}} 
[\bigcap_{q \in K'_{-i}} Reach(T_q)] \ge 1 - \epsilon$.

		\item Suppose $T_i$ is of \textsf{Q}-form 
($T_i \xrightarrow{1} T_j \; T_r$) where $i \in K'$ and 
there exists a split of the rest of the target non-terminals, 
implied by $K_L \subseteq K'_{-i}$ and $K'_{-i} - K_L$, 
such that $T_j \in F_{K_L} \wedge T_r \in F_{K'_{-i} - K_L}$. 
Let $h_l := T_i(l, T_j)$ and $h_r := T_i(r, T_r)$. 
For every $\epsilon > 0$, if we let 
$\epsilon' := 1 - \sqrt{1 - \epsilon}$, then by combining the 
two witness strategies $\sigma_{K_L}^{\epsilon'}$ and 
$\sigma_{K'_{-i} - K_L}^{\epsilon'}$ from the induction assumption 
for non-terminals $T_j$ and $T_r$, respectively, 
to be used in the next generation as if the play starts in it, 
we obtain a strategy $\sigma_{K'}^{\epsilon}$ for a 
starting (target) non-terminal $T_i$ such that 
$Pr_{T_i}^{\sigma_{K'}^{\epsilon}}[\bigcap_{q \in K'} 
Reach(T_q)] = Pr_{T_i}^{\sigma_{K'}^{\epsilon}} 
[\bigcap_{q \in K'_{-i}} Reach(T_q)] \ge 
Pr_{h_l}^{\sigma_{K'}^{\epsilon}}[\bigcap_{q \in K_L} Reach(T_q)] 
\cdot Pr_{h_r}^{\sigma_{K'}^{\epsilon}} 
[\bigcap_{q \in K'_{-i} - K_L} Reach(T_q)] = 
Pr_{T_j}^{\sigma_{K_L}^{\epsilon'}} 
[\bigcap_{q \in K_L} Reach(T_q)] \cdot 
Pr_{T_r}^{\sigma_{K'_{-i} - K_L}^{\epsilon'}}[\bigcap_{q \in K'_{-i} - K_L} Reach(T_q)] \ge (1 - \epsilon')^2 = 
(\sqrt{1 - \epsilon})^2  = 1 - \epsilon$.

		\item Suppose $T_i$ is of \textsf{Q}-form 
($T_i \xrightarrow{1} T_j \; T_r$) where there exists 
a proper split of the target non-terminals from set $K'$, 
implied by $K_L \subset K'$ (where $K_L \not= \emptyset$) and 
$K' - K_L$, such that $T_j \in F_{K_L} \wedge T_r \in F_{K' - K_L}$. 
Similarly, for each $\epsilon > 0$, let 
$\epsilon' := 1 - \sqrt{1 - \epsilon}$ and combine 
the two witness strategies $\sigma_{K_L}^{\epsilon'}$ and 
$\sigma_{K' - K_L}^{\epsilon'}$ from the induction assumption 
for non-terminals $T_j$ and $T_r$ in the same way as in (iii). 
It follows that property $(A)_{K'}^i$ is satisfied.
	\end{enumerate}

Now consider non-terminals $T_i$ added to set $D_{K'}$ at step II.2. 
If $T_i$ is of \textsf{L}-form, then all associated rules generate 
a child of non-terminal $T_j$ already in $D_{K'}$, 
where $(A)_{K'}^j$ holds by the (nested) induction. 
So using for every $\epsilon > 0$ the strategy 
$\sigma_{K'}^{\epsilon}$ from the nested induction assumption 
for each such non-terminal $T_j$ and applying the same argument 
as in (i), then property $(A)_{K'}^i$ is also satisfied.

If $T_i$ is of \textsf{M}-form, then 
$\exists a^* \in \Gamma^i:\; T_i \xrightarrow{a^*} T_j,\; 
T_j \in D_{K'}$. Again let $h := T_i(u, T_j)$. 
By combining, for every $\epsilon > 0$, the witness strategy 
$\sigma_{K'}^{\epsilon}$ for non-terminal $T_j$ 
(from the nested induction assumption), as if the play 
starts in it, with the initial local 
choice of choosing action $a^*$ deterministically 
starting at a non-terminal $T_i$, we obtain an augmented 
strategy $\sigma_{K'}^{\epsilon}$ for a starting non-terminal 
$T_i$ such that $Pr_{T_i}^{\sigma_{K'}^{\epsilon}} 
[\bigcap_{q \in K'} Reach(T_q)] = Pr_h^{\sigma_{K'}^{\epsilon}} 
[\bigcap_{q \in K'} Reach(T_q)] = 
Pr_{T_j}^{\sigma_{K'}^{\epsilon}} 
[\bigcap_{q \in K'} Reach(T_q)] \ge 1 - \epsilon$.

If $T_i$ is of \textsf{Q}-form ($T_i \xrightarrow{1} T_j \; T_r$), 
then $T_j \in D_{K'} \vee T_r \in D_{K'}$, i.e., 
for every $\epsilon > 0$, $\exists \sigma_{K'}^{\epsilon} \in \Psi$ 
such that $Pr_{T_y}^{\sigma_{K'}^{\epsilon}}[\bigcap_{q \in K'} 
Reach(T_q)] \ge 1 - \epsilon \Leftrightarrow 
Pr_{T_y}^{\sigma_{K'}^{\epsilon}}[\bigcup_{q \in K'} 
Reach^{\complement}(T_q)] \le \epsilon$, where $y \in \{j, r\}$. 
Let $h_y := T_i(x, T_y)$ and 
$h_{\bar{y}} := T_i(\bar{x}, T_{\bar{y}})$, where 
$\bar{y} \in \{j, r\} - \{y\}$, $x \in \{l, r\}$ and 
$\bar{x} \in \{l, r\} - \{x\}$. By augmenting strategy 
$\sigma_{K'}^{\epsilon}$ to be used from the next generation 
from the child of non-terminal $T_y$ as if the play starts in it 
and using an arbitrary strategy from the child of non-terminal 
$T_{\bar{y}}$, it follows that 
$Pr_{T_i}^{\sigma_{K'}^{\epsilon}}[\bigcup_{q \in K'} 
Reach^{\complement}(T_q)] \le Pr_{h_y}^{\sigma_{K'}^{\epsilon}} 
[\bigcup_{q \in K'} Reach^{\complement}(T_q)] \cdot 
Pr_{h_{\bar{y}}}^{\sigma_{K'}^{\epsilon}}[\bigcup_{q \in K'} 
Reach^{\complement}(T_q)] \le 
Pr_{T_y}^{\sigma_{K'}^{\epsilon}} [\bigcup_{q \in K'} 
Reach^{\complement}(T_q)] \le \epsilon \Leftrightarrow 
Pr_{T_i}^{\sigma_{K'}^{\epsilon}}[\bigcap_{q \in K'} 
Reach(T_q)] \ge 1 - \epsilon$, i.e., property $(A)_{K'}^i$ holds.
\end{proof}

\medskip
Next, we show that if $T_i \in S_{K'}$, then 
property $(B)_{K'}^i$ is satisfied.

\begin{lemma}
	Every non-terminal $T_i \in S_{K'}$ satisfies 
property $(B)_{K'}^i$.
	\label{lemma:property_B-LS}
\end{lemma}

\begin{proof}
    Again this is proved via a nested induction based on the time 
a non-terminal is added to set $S_{K'}$. Assuming that all 
non-terminals $T_j$, added already to set $S_{K'}$ 
in previous iterations and steps of the inner loop, 
satisfy $(B)_{K'}^j$, then we need to show that for 
a new addition $T_i$ to set $S_{K'}$, 
property $(B)_{K'}^i$ also holds.

Consider the non-terminals $T_i$ added to set $S_{K'}$ 
at the initialization step II.4. 

If $T_i$ is of \textsf{L}-form where 
$T_i \rightarrow \varnothing \vee T_i \rightarrow T_j,\; 
T_j \in Z_{K'}$, then with a constant positive probability 
non-terminal $T_i$ immediately either does not generate 
any offspring at all or generates a child of non-terminal 
$T_j \in Z_{K'}$, for which we already know that $(B)_{K'}^j$ holds. 
It is clear that property $(B)_{K'}^i$ is also satisfied.

If, for some subset $K'' \subset K'$, $T_i \in S_{K''}$, i.e., 
property $(B)_{K''}^i$ holds, then there is a value $g > 0$ 
such that $\forall \sigma \in \Psi:\; 
Pr_{T_i}^{\sigma}[\bigcap_{q \in K'} Reach(T_q)] \le 
Pr_{T_i}^{\sigma}[\bigcap_{q \in K''} Reach(T_q)] \le 1 - g$ 
and so property $(B)_{K'}^i$ is also satisfied. Note that if, 
for some subset $K'' \subset K'$, $T_i \in Z_{K''}$, then 
similarly $T_i \in Z_{K'}$ and so already $T_i \not\in X$.

If $T_i$ is a target non-terminal in set $K'$ (i.e., $i \in K'$), 
then since it has not been added to set $D_{K'}$ in 
step II.1: (1) if of \textsf{L}-form, 
it generates with a constant positive probability a child of 
non-terminal $T_j \in S_{K'_{-i}} \cup Z_{K'_{-i}}$, 
where $(B)_{K'_{-i}}^j$ holds; (2) if of \textsf{M}-form, 
irrespective of the strategy it generates a child 
of non-terminal $T_j \in S_{K'_{-i}} \cup Z_{K'_{-i}}$, where 
again $(B)_{K'_{-i}}^j$ holds; (3) and if of \textsf{Q}-form, 
it generates two children of non-terminals $T_j, T_r$, 
for which no matter how we split the rest of 
the target non-terminals from set $K'_{-i}$ 
(into subsets $K_L \subseteq K'_{-i}$ and $K'_{-i} - K_L$), 
either $(B)_{K_L}^j$ holds or $(B)_{K'_{-i} - K_L}^r$ holds. 
In other words, for a target non-terminal $T_i$ in the initial 
set $S_{K'}$ there is no sequence of strategies to ensure 
that the rest of the target non-terminals are reached 
with probability arbitrarily close to 1 
(the reasoning behind this last statement is the same as the 
arguments in (i) - (iii) below, since for a starting (target) 
non-terminal $T_i$: $\forall \sigma \in \Psi:\; 
Pr_{T_i}^{\sigma}[\bigcap_{q \in K'} Reach(T_q)] = 
Pr_{T_i}^{\sigma}[\bigcap_{q \in K'_{-i}} Reach(T_q)]$).

Observe that by the end of step II.4. all target non-terminals 
$T_q, q \in K'$ belong either to set $D_{K'}$ or set $S_{K'}$. 
Now consider a non-terminal $T_i$ added to set $S_{K'}$ 
in step II.5. during some iteration of the inner loop.
	\begin{enumerate}[label=(\roman*)]
		\item Suppose $T_i$ is of \textsf{L}-form. 
Then $T_i \rightarrow T_j,\; T_j \in S_{K'} \cup Z_{K'}$, 
where $(B)_{K'}^j$ holds. So irrespective of the strategy 
there is a constant positive probability to generate 
a child of the above non-terminal $T_j$ such that 
$Pr_{T_j}^*[\bigcap_{q \in K'} Reach(T_q)] < 1$, 
or in other words, $\exists g > 0$ such that 
$\forall \sigma \in \Psi:\; Pr_{T_j}^{\sigma} 
[\bigcap_{q \in K'} Reach(T_q)] \le 1 - g$. 
Let $h := T_i(u, T_j)$. 
But, there is a value $g > 0$ such that 
$\forall \sigma \in \Psi:\; 
Pr_{T_i}^{\sigma}[\bigcup_{q \in K'} 
Reach^{\complement}(T_q)] \ge p_{ij} \cdot 
Pr_h^{\sigma}[\bigcup_{q \in K'} Reach^{\complement}(T_q)] \ge 
p_{ij} \cdot g$ if and only if $\forall \sigma \in \Psi:\; 
Pr_{T_j}^{\sigma}[\bigcup_{q \in K'} 
Reach^{\complement}(T_q)] \ge g$, where $p_{ij} > 0$ 
is the probability of the rule $T_i \xrightarrow{p_{ij}} T_j$. 
And since the latter part of the statement holds, then the 
former, showing property $(B)_{K'}^i$, also holds.
			
		\item Suppose $T_i$ is of \textsf{M}-form. 
Then $\forall a \in \Gamma^i:\; T_i \xrightarrow{a} T_j,\; 
T_j \in S_{K'} \cup Z_{K'}$. So irrelevant of strategy $\sigma$ 
for the player, starting in a non-terminal $T_i$ 
the next generation surely consists of some non-terminal 
$T_j$ with the property $\sup_{\sigma \in \Psi} 
Pr_{T_j}^{\sigma}[\bigcap_{q \in K'} Reach(T_q)] < 1$, 
i.e., $\forall \sigma \in \Psi:\; 
Pr_{T_j}^{\sigma}[\bigcap_{q \in K'} Reach(T_q)] \le 1 - g$, 
for some value $g > 0$. 
Clearly, for some value $g > 0$, $\forall \sigma \in \Psi:\; 
Pr_{T_i}^{\sigma}[\bigcap_{q \in K'} Reach(T_q)] \le 
\max_{\{T_j \in S_{K'} \cup Z_{K'}\}} 
Pr_{T_i(u, T_j)}^{\sigma}[\bigcap_{q \in K'} Reach(T_q)] \le 1 - g$ 
(i.e., property $(B)_{K'}^i$) 
if and only if $\forall \sigma \in \Psi:\; 
\max_{\{T_j \in S_{K'} \cup Z_{K'}\}} 
Pr_{T_j}^{\sigma} [\bigcap_{q \in K'} Reach(T_q)] \le 1 - g$, 
where the latter is satisfied.

		\item Suppose $T_i$ is of \textsf{Q}-form 
(i.e., $T_i \xrightarrow{1} T_j \; T_r$), 
then $T_j, T_r \in S_{K'} \cup Z_{K'}$, i.e., 
both $(B)_{K'}^j$ and $(B)_{K'}^r$ are satisfied. We know that:
\begin{enumerate}
    \setlength{\topsep}{0em}
    \item[1)] Neither of the two children can 
single-handedly reach all target non-terminals from set $K'$ 
with probability arbitrarily close to $1$. 
That is, for some value $g > 0$, 
for every $\sigma \in \Psi$, $Pr_{T_j}^{\sigma} 
[\bigcap_{q \in K'} Reach(T_q)] \le 1 - g$ and 
$Pr_{T_r}^{\sigma}[\bigcap_{q \in K'} Reach(T_q)] \le 1 - g$.

    \item[2)] Moreover, since $T_i$ was not added to 
set $D_{K'}$ in step II.1., then $\forall K_L \subset K'$ 
(where $K_L \not= \emptyset$) either $(B)_{K_L}^j$ holds 
(i.e., $T_j \not\in F_{K_L}$) or $(B)_{K' - K_L}^r$ holds 
(i.e., $T_r \not\in F_{K' - K_L}$), i.e., there is 
some value $g > 0$ such that either 
$\forall \sigma \in \Psi:\; Pr_{T_j}^{\sigma} 
[\bigcap_{q \in K_L} Reach(T_q)] \le 1 - g$ or 
$\forall \sigma \in \Psi:\; Pr_{T_r}^{\sigma} 
[\bigcap_{q \in K' - K_L} Reach(T_q)] \le 1 - g$.

(Statements 1) and 2) hold for the same value $g > 0$, since 
there are only finitely many subsets of $K'$, so we can take 
$g$ to be the minimum of all such values from all 
the properties $(B)_{K''}^{j/r}$ ($K'' \subseteq K'$).)
\end{enumerate}

Let $h_l := T_i(l, T_j)$ and $h_r := T_i(r, T_r)$. 
Notice that for any $\sigma \in \Psi$ and for any 
$q' \in K'$, $Pr_{T_i}^{\sigma}[\bigcup_{q \in K'} 
Reach^{\complement}(T_q)] \ge Pr_{T_i}^{\sigma} 
[Reach^{\complement}(T_{q'})] = 
Pr_{h_l}^{\sigma}[Reach^{\complement}(T_{q'})] \cdot 
Pr_{h_r}^{\sigma}[Reach^{\complement}(T_{q'})]$.

We claim that $\exists g_i > 0$ such that 
$\forall \sigma \in \Psi:\; \bigvee_{q \in K'} 
Pr_{T_j}^{\sigma}[Reach^{\complement}(T_q)] \cdot 
Pr_{T_r}^{\sigma}[Reach^{\complement}(T_q)] \ge g_i$. 
But for any $q \in K'$ and for any $\sigma \in \Psi$ 
one can obviously construct $\sigma' \in \Psi$ such that 
$Pr_{T_j}^{\sigma}[Reach^{\complement}(T_q)] = 
Pr_{h_l}^{\sigma'}[Reach^{\complement}(T_q)]$ and 
similarly for non-terminal $T_r$. 
Therefore, it follows from the claim that 
$\forall \sigma \in \Psi:\; \bigvee_{q \in K'} 
Pr_{h_l}^{\sigma}[Reach^{\complement}(T_q)] \cdot 
Pr_{h_r}^{\sigma}[Reach^{\complement}(T_q)] \ge g_i$ 
and, therefore, it follows that $\forall \sigma \in \Psi:\; 
Pr_{T_i}^{\sigma}[\bigcup_{q \in K'} 
Reach^{\complement}(T_q)] \ge g_i \Leftrightarrow 
Pr_{T_i}^{\sigma}[\bigcap_{q \in K'} Reach(T_q)] 
\le 1 - g_i$.

Suppose the opposite, i.e., assume $(\mathcal{P})$ that 
$\forall g' > 0,\; \exists \sigma_{g'} \in \Psi:\; 
\bigwedge_{q \in K'} Pr_{T_j}^{\sigma_{g'}} 
[Reach^{\complement}(T_q)] \cdot Pr_{T_r}^{\sigma_{g'}} 
[Reach^{\complement}(T_q)] < g'$. 
Now for any $q \in K'$, by statement 2) above, we know 
that $T_j \not\in F_{\{q\}} \vee T_r \not\in F_{K'_{-q}}$ and 
$T_j \not\in F_{K'_{-q}} \vee T_r \not\in F_{\{q\}}$. 
First, suppose that in fact for some $q' \in K'$ it 
is the case that $T_j \not\in F_{\{q'\}} \wedge 
T_r \not\in F_{\{q'\}}$ (i.e., $T_j \in S_{\{q'\}} \cup Z_{\{q'\}} 
\wedge T_r \in S_{\{q'\}} \cup Z_{\{q'\}}$). 
That is, for some value $g > 0$, $\forall \sigma \in \Psi:\; 
Pr_{T_j}^{\sigma}[Reach^{\complement}(T_{q'})] \ge g$ 
and $Pr_{T_r}^{\sigma}[Reach^{\complement}(T_{q'})] \ge g$, 
where our claim follows directly by letting $g_i := g^2$ 
(hence, contradiction to $(\mathcal{P})$). 
Second, suppose that for some $q' \in K'$ it is the case that 
$T_j \not\in F_{K'_{-q'}} \wedge T_r \not\in F_{K'_{-q'}}$ 
(i.e., $T_j \in S_{K'_{-q'}} \cup Z_{K'_{-q'}} \wedge 
T_r \in S_{K'_{-q'}} \cup Z_{K'_{-q'}}$). 
But then $T_i$ would have been added to set $S_{K'_{-q'}}$ 
at step II.5.(c) when constructing the answer for subset of 
targets $K'_{-q'}$. However, we already know that 
$T_i \in \bigcap_{K'' \subset K'} F_{K''}$ (following from 
steps II.3. and II.4. that $T_i \not\in \bigcup_{K'' \subset K'} 
(S_{K''} \cup Z_{K''})$). 
Hence, again a contradiction.

Therefore, it follows that for every $q \in K'$, either 
$T_j \not\in F_{\{q\}} \wedge T_j \not\in F_{K'_{-q}}$ or 
$T_r \not\in F_{\{q\}} \wedge T_r \not\in F_{K'_{-q}}$. 
And in particular, the essential part is that 
$\forall q \in K'$, either $T_j \not\in F_{\{q\}}$ or 
$T_r \not\in F_{\{q\}}$.
That is, for every $q \in K'$, for some value $g > 0$ 
either $\forall \sigma \in \Psi:$ 
$Pr_{T_j}^{\sigma}[Reach^{\complement}(T_q)] \ge g$, 
or $\forall \sigma \in \Psi:$ 
$Pr_{T_r}^{\sigma}[Reach^{\complement}(T_q)] \ge g$. 
But then, combined with assumption $(\mathcal{P})$, 
it actually follows that there exists a subset $K'' \subseteq K'$ 
such that $\forall \epsilon > 0,\; 
\exists \sigma_{\epsilon} \in \Psi:\; \bigwedge_{q \in K''} 
Pr_{T_r}^{\sigma_{\epsilon}}[Reach^{\complement}(T_q)] 
\le \epsilon \; \wedge \; \bigwedge_{q \in K' - K''} 
Pr_{T_j}^{\sigma_{\epsilon}}[Reach^{\complement}(T_q)] 
\le \epsilon$. And by Proposition \ref{prop:equiv}(5.), 
it follows that $\forall \epsilon > 0,\; 
\exists \sigma'_{\epsilon} \in \Psi:\; 
Pr_{T_r}^{\sigma'_{\epsilon}}[\bigcap_{q \in K''} 
Reach(T_q)] \ge 1 - \epsilon \wedge 
Pr_{T_j}^{\sigma'_{\epsilon}}[\bigcap_{q \in K' - K''} 
Reach(T_q)] \ge 1 - \epsilon$, i.e., $T_j \in F_{K' - K''} \wedge 
T_r \in F_{K''}$, contradicting the known facts 1) and 2). 
Hence, assumption $(\mathcal{P})$ is wrong and our claim 
is satisfied.
\end{enumerate}

Now consider non-terminals $T_i$ added to set $S_{K'}$ in 
step II.10. at some iteration of the inner loop, i.e., 
$T_i \in Y_{K'} := X - (S_{K'} \cup F_{K'}) \subseteq \bar{Z}_{K'}$. 
Due to the fact that $T_i$ has not been added previously 
to sets $D_{K'}$, $S_{K'}$ or $F_{K'}$, 
then all of the following hold:

\begin{enumerate}[label=(\arabic*.)]
	\item $i \not\in K'$;
	
	\item if $T_i$ is of \textsf{L}-form, then a non-terminal 
$T_i$ generates with probability 1 a non-terminal which 
belongs to set $Y_{K'}$ (otherwise $T_i$ would have been added 
to sets $S_{K'}$ or $F_{K'}$ in step II.4., II.5. or step II.9., 
respectively);
	
	\item if $T_i$ is of \textsf{M}-form, then 
$\forall a \in \Gamma^i:\; T_i \xrightarrow{a} T_d,\; 
T_d \not\in F_{K'} \cup D_{K'}$ (otherwise $T_i$ 
would have been added to sets $F_{K'}$ or $D_{K'}$ in 
step II.2. or step II.9., respectively), and 
$\exists 'a \in \Gamma^i:\; T_i \xrightarrow{a'} T_j,\; 
T_j \not\in S_{K'} \cup Z_{K'}$, i.e., $T_j \in Y_{K'}$ 
(otherwise $T_i$ would have been added to set $S_{K'}$ 
in step II.5.); and
	
	\item if $T_i$ is of \textsf{Q}-form 
($T_i \xrightarrow{1} T_j \; T_r$), then w.l.o.g. 
$T_j \in Y_{K'}$ and $T_r \in Y_{K'} \cup S_{K'} \cup Z_{K'}$ 
(since $T_i$ has not been added to the other sets in 
steps II.2., II.5., or II.9.).
\end{enumerate}

Observe that any MEC in subgraph $G[X - S_{K'}]$, that 
contains a node from set $Y_{K'}$, is in fact entirely 
contained in subgraph $G[Y_{K'}]$, and also that 
there is at least one MEC in $G[Y_{K'}]$. 
This is due to statements (2.) - (4.) and the two key facts 
that all nodes in $G[Y_{K'}]$ have at least one outgoing edge 
and there is only a finite number of nodes.

However, consider any MEC, $C$, in $G[Y_{K'}]$ 
($Y_{K'} \subseteq X - S_{K'}$). As $C$ has not been added to 
set $F_{K'}$ at step II.8., then $P_C \not= K'$ 
(where $P_C = \{q \in K' \mid C \cap H_q \not= \emptyset\}$) and:

\begin{itemize}[topsep=0em]
	\setlength{\itemsep}{0em}
	\item either $P_C = \emptyset$,
	
	\item or $P_C \not= \emptyset$ and for every $T_u \in C$ 
of \textsf{M}-form it holds that $\forall b \in \Gamma^u:\; 
T_u \xrightarrow{b} T_v,\; T_v \not\in F_{K' - P_C}$.
\end{itemize}

First, let us focus on the second point. Note that for any 
non-terminal $T_j \in C$, clearly $T_j \in F_{P_C}$, 
and in fact, $\exists \sigma_{P_C} \in \Psi:\; 
Pr_{T_j}^{\sigma_{P_C}}[\bigcap_{q \in P_C} Reach(T_q)] = 1$. 
That is because, starting at a non-terminal $T_j \in C$, 
due to $C$ being a MEC in $G[Y_{K'}]$, such a strategy 
$\sigma_{P_C}$ can ensure that, for each $q \in P_C$, 
infinitely often a copy of a \textsf{Q}-form 
non-terminal in set $H_q \cap C$ is generated, 
which in turn spawns an independent copy of 
some non-terminal in set $\bar{Z}_{\{q\}}$ 
and thus infinitely often provides a positive probability 
bounded away from zero 
(by Proposition \ref{prop:QualMoR-NonReach}) 
to reach target non-terminal $T_q$.

\begin{enumerate}[topsep=0em]
	\item[(*)] We claim that for any \textsf{Q}-form non-terminal 
$T_i \in C$ (i.e., $T_i \xrightarrow{1} T_j \; T_r$ where 
w.l.o.g. $T_j \in C \subseteq Y_{K'}$), it is guaranteed 
that $T_r \not\in F_{K' - P_C}$. To see this, if it was 
the case that $T_r \in F_{K' - P_C}$, then, since 
$T_j \in F_{P_C}$, it would follow that $T_i$ would have been 
added to set $D_{K'}$ in step II.1., leading to a contradiction. 
	
	\item[(**)] What is more, due to the definition of set $P_C$, 
it follows that for any \textsf{Q}-form non-terminal 
$T_i \in C$ (i.e., $T_i \xrightarrow{1} T_j \; T_r$ where 
w.l.o.g. $T_j \in C$), $T_r \in \bigcap_{q' \in K' - P_C} Z_{\{q'\}}$, 
i.e., $\sup_{\sigma \in \Psi} Pr_{T_r}^{\sigma} 
[Reach(T_{q'})] = 0$, for each $q' \in K' - P_C$. 
Note also that $T_r \not\in C$, since $C \subseteq Y_{K'} 
\subseteq \bar{Z}_{K'} \subseteq \bar{Z}_{\{q\}},\; 
\forall q \in K'$ (so if $T_r \in C$, then $P_C = K'$ 
and $C$ would have been added to set $F_{K'}$ in step II.8.).

	\item[(***)] Furthermore, as stated in 
the second bullet point above, for every non-terminal 
$T_u \in C$ of \textsf{M}-form and $\forall b \in \Gamma^u:\; 
T_u \xrightarrow{b} T_v,\; T_v \not\in F_{K' - P_C}$.
\end{enumerate}

\medskip
And as we know, for every $T_v \in S_{K' - P_C} \cup 
Z_{K' - P_C}$, property $(B)_{K' - P_C}^v$ holds. 
In other words, there exists a value $g > 0$ such that 
regardless of strategy $\sigma$, for any 
$T_v \not\in F_{K' - P_C}$, $Pr_{T_v}^{\sigma} 
[\bigcap_{q \in K' - P_C} Reach(T_q)] \le 1 - g$.

Now let $\sigma$ be an arbitrary strategy fixed for the player. 
Denote by $w$ the path (in the play), where $w$ begins 
at a starting non-terminal $T_i \in C$ and evolves in 
the following way. If the current copy $o$ on the path $w$ 
is of a \textsf{L}-form or a \textsf{M}-form non-terminal 
$T_j \in C$, then $w$ follows along the unique successor 
of $o$ in the play. And if the current copy $o$ on path 
$w$ is of a \textsf{Q}-form non-terminal $T_j \in C$ 
($T_j \xrightarrow{1} T_{j'} \; T_r$ where w.l.o.g. 
$T_{j'} \in C$), then $w$ follows along the child of 
non-terminal $T_{j'}$. If the current copy $o$ on path $w$ 
is of a non-terminal not belonging in $C$, then the path $w$ 
terminates. Denote by $\square C$ the event that path $w$ 
is infinite, i.e., all non-terminals observed along path $w$ 
are in $C$ and path $w$ never leaves $C$ and never terminates. 
Then for any starting non-terminal $T_i \in C$:
\begin{align*}
	& Pr_{T_i}^{\sigma} \Big[ \bigcap_{q \in K'} 
Reach(T_q) \Big] = Pr_{T_i}^{\sigma} 
\Big[ \Big( \bigcap_{q \in P_C} Reach(T_q) \Big) \cap 
\Big( \bigcap_{q \in K' - P_C} Reach(T_q) \Big) \Big] \\
	& \le Pr_{T_i}^{\sigma} \Big[ \bigcap_{q \in K' - P_C} 
Reach(T_q) \Big] = Pr_{T_i}^{\sigma} 
\Big[ \Big( \bigcap_{q \in K' - P_C} Reach(T_q) \Big) \cap 
\square C \Big] + \\
	& Pr_{T_i}^{\sigma} \Big[ \Big( \bigcap_{q \in K' - P_C} 
Reach(T_q) \Big) \cap \neg \square C \Big] = 
Pr_{T_i}^{\sigma} \Big[ \Big( \bigcap_{q \in K' - P_C} 
Reach(T_q) \Big) \cap \neg \square C \Big] \\
    & \le \max_{T_v \not\in F_{K' - P_C}} \; \sup_{\tau \in \Psi} 
Pr_{T_v}^{\tau} \Big[ \bigcap_{q \in K' - P_C} 
Reach(T_q) \Big] \le 1 - g
\end{align*}
The event of reaching all target non-terminals 
from set $K' - P_C$ can be split into the event of 
reaching all targets non-terminals from set $K' - P_C$ and 
path $w$ being infinite union with the event of 
reaching all targets non-terminals from set $K' - P_C$ and 
path $w$ being finite. 
Moreover, $Pr_{T_i}^{\sigma}[(\bigcap_{q \in K' - P_C} 
Reach(T_q)) \cap \square C] = 0$, due to statements (1.), 
(*) and (**). The second to last inequality follows: because of 
statement (**); and also due to statements (*), (**) and (***), 
once event $\neg \square C$ occurs and path $w$ leaves 
MEC, $C$, it terminates immediately in some non-terminal 
$T_v \not\in C$ which also satisfies that 
$T_v \not\in F_{K' - P_C}$. 
And the last inequality follows from property $(B)_{K' - P_C}^v$ 
for any such non-terminal $T_v \not\in F_{K' - P_C}$.

And since $\sigma$ was an arbitrary strategy for the player, 
then it follows that for any such MEC, $C$, in $G[Y_{K'}]$ 
(where  $P_C \not= \emptyset$) and for any non-terminal $T_i \in C$: 
$Pr_{T_i}^*[\bigcap_{q \in K'} Reach(T_q)] < 1$, i.e., 
property $(B)_{K'}^i$ holds.

Analysing MECs, $C$, where $P_C = \emptyset$, the argument is 
similar. Property (**) holds by definition of set $P_C$. 
And by property (3.), for every \textsf{M}-form 
non-terminal $T_u \in C$ and for every $b \in \Gamma^u: 
T_u \xrightarrow{b} T_{u'},\; 
T_{u'} \in (Y_{K'} \cup S_{K'} \cup Z_{K'})$. 
Then because of properties (1.), (3.) and (**), it follows 
that for any $T_i \in C$, $\forall \sigma \in \Psi:\; 
Pr_{T_i}^{\sigma}[\bigcap_{q \in K'} Reach(T_q)] \le 
\max_{T_{u'} \in (Y_{K'} \cup S_{K'} \cup Z_{K'})} 
Pr_{T_{u'}}^{\sigma}[\bigcap_{q \in K'} Reach(T_q)]$.

For non-terminals $T_{u'}$ 
in sets $S_{K'}$ and $Z_{K'}$, we already know by induction 
that property $(B)_{K'}^{u'}$ is satisfied. 
Moreover, from standard algorithms for MEC-decomposition, 
one can see that there is an ordering of the MECs 
in $G[Y_{K'}]$ where the bottom level (level $0$) 
consists of MECs, $C''$, that have no out-going edges 
from the MEC at all (these are analogous to 
bottom strongly connected components in an SCC-decomposition) 
and for which $P_{C''} \not= K'$, and for further ``levels'' of 
MECs in the ordering the following is true: MECs or 
nodes that do not belong to any MEC, at level $t \ge 1$, 
have directed paths out of them leading to MECs (or nodes 
not in any MEC) at levels $< t$. 
If we rank the MECs and the independent nodes (not belonging 
to any MEC) in $G[Y_{K'}]$, using this ordering, 
and use an inductive argument, it can be shown that, in 
the case when the above mentioned non-terminal $T_{u'}$ 
belongs to $Y_{K'}$ and MEC, $C$, has rank $t \ge 1$ 
in the ordering, then $T_{u'}$ belongs to a lower rank $< t$, 
and thus by the inductive argument, has been shown to 
have property $(B)_{K'}^{u'}$.

Therefore, for any non-terminal $T_i$ in any MEC, $C$, 
in $G[Y_{K'}]$, $(B)_{K'}^i$ holds. 
And also by the inductive argument above for the ordering 
of nodes in $G[Y_{K'}]$, same holds for any 
non-terminal $T_i \in Y_{K'}$ not belonging to a MEC.
\end{proof}

\medskip
Now we show that for non-terminals $T_i \in F_{K'}$, 
when the inner loop for subset $K' \subseteq K$ terminates, 
the property $(A)_{K'}^i$ is satisfied. That is:
\begin{align*}
    \forall \epsilon > 0,\; \exists \sigma_{K'}^{\epsilon} \in \Psi:\; 
Pr_{T_i}^{\sigma_{K'}^{\epsilon}} \Big[ \bigcap_{q \in K'} 
Reach(T_q) \Big] \ge 1 - \epsilon
\end{align*}
We will also show how to construct such a strategy
$\sigma_{K'}^{\epsilon}$, for a given $\epsilon > 0$.
Since we have already proved it for non-terminals in set $D_{K'}$, 
in the following Lemma we refer to the part of set $F_{K'}$ 
not containing set $D_{K'}$, i.e., to set $F_{K'} = X - S_{K'}$.

\begin{lemma}
	Every non-terminal $T_i \in F_{K'}$ satisfies 
the property $(A)_{K'}^i$.
	\label{lemma:property_A-LS}
\end{lemma}

\begin{proof}
    Denote by $F_{K'}^0$ the initialized set of non-terminals 
from step II.8. Let us first observe the properties 
for non-terminals $T_i \in F_{K'} = X - S_{K'}$. 
None of them is a target non-terminal from set $K'$, 
i.e., $i \not\in K'$. If $T_i$ is of \textsf{L}-form, then:
	\begin{enumerate}
		\setlength{\itemsep}{0em}
		\item[(L.0)] if $T_i$ belongs to a MEC, 
$C \subseteq F_{K'}^0$, then a non-terminal $T_i$ 
generates with probability 1 as offspring some non-terminal 
either in set $C$ or in set $D_{K'}$ 
(since \textsf{L}-form non-terminals in $X - S_{K'}$ do not 
have associated probabilistic rules to non-terminals 
in $S_{K'} \cup Z_{K'}$).
		
		\item[(L)] otherwise, a non-terminal $T_i$ 
generates with probability 1 as offspring some non-terminal 
either in set $F_{K'}$ or in set $D_{K'}$.
	\end{enumerate}
If $T_i$ is of \textsf{M}-form, then $\forall a \in \Gamma^i:\; 
T_i \xrightarrow{a} T_d,\; T_d \not\in D_{K'}$ and:
	\begin{enumerate}
		\setlength{\itemsep}{0em}
		\item[(M.0)] if $T_i$ belongs to a MEC, 
$C \subseteq F_{K'}^0$, then $\exists a^* \in \Gamma^i:\; 
T_i \xrightarrow{a^*} T_j,\; T_j \in C$.
		
		\item[(M)] otherwise, $\exists a^* \in \Gamma^i:\; 
T_i \xrightarrow{a^*} T_j,\; T_j \in F_{K'}$.
	\end{enumerate}
If $T_i$ is of \textsf{Q}-form 
(i.e., $T_i \xrightarrow{1} T_j \; T_r$), then 
$T_j, T_r \not\in D_{K'}$ and:
	\begin{enumerate}[topsep=0em]
		\setlength{\itemsep}{0.5em}
		\item[] if $T_i$ belongs to a MEC, 
		        $C \subseteq F_{K'}^0$, then:
			\begin{enumerate}[topsep=0.1em]
				\setlength{\itemsep}{0.3em}
				\item[(Q.0)] either, w.l.o.g., $T_j \in C$ and there exists 
				some $q \in K'$ such that $T_r \in \bar{Z}_{\{q\}}$,
				
				\item[(Q.1)] or, w.l.o.g., $T_j\in C$ and there is no 
				$q \in K'$ such that $T_r \in \bar{Z}_{\{q\}}$.
			\end{enumerate}
				
		\item[(Q)] otherwise, i.e., if $T_i \not\in F_{K'}^0$, 
		then w.l.o.g., $T_j \in F_{K'}$.
	\end{enumerate}

\medskip
\begin{enumerate}
	\item[($\mathfrak{P}$)] Let us recall that 
for every $q \in K'$, there is a deterministic static strategy 
$\sigma'_{\{q\}}$ for the player and a value $b_{\{q\}} > 0$ 
such that, for each non-terminal $T_r \in \bar{Z}_{\{q\}}$, 
$Pr_{T_r}^{\sigma'_{\{q\}}}[Reach(T_q)] \ge b_{\{q\}}$. 
Let $b := \min_{q \in K'} \{b_{\{q\}}\} > 0$.
\end{enumerate}

\medskip
Given $\epsilon > 0$, let $\epsilon' := (1 - \sqrt{1 - \epsilon})/k$ 
(where $k = |K|$) and let us prove the Lemma and construct 
the randomized non-static strategy $\sigma_{K'}^{\epsilon}$ 
inductively. 

Consider the non-terminals added to set $F_{K'}$ at 
the initialization step II.8. during the last iteration 
of the inner loop. And, in particular, consider every MEC, $C$, 
added at step II.8. There is one of two reasons for why $C$ 
was added to set $F_{K'}^0$.

For the first reason, suppose that $1 \le |P_C| < l = |K'|$ and 
that there is a non-terminal $T_u \in C$ of 
\textsf{M}-form where $\exists b \in \Gamma^u:\; 
T_u \xrightarrow{b} T_{u'},\; T_{u'} \in F_{K' - P_C}$.

Consider any finite ancestor history $h$ of height $t$ 
(meaning the length of the sequence of ancestors that 
the history represents is $t$) such that $h$ starts at 
a non-terminal $T_v \in C$ and all non-terminals in $h$ 
belong to the MEC, $C$. 
Let $o$ denote the non-terminal copy at the end of 
the ancestor history $h$.

If $o$ is a copy of the non-terminal $T_u \in C$ (from above), let strategy 
$\sigma_{K'}^{\epsilon}$ choose uniformly at random 
among actions from statement (M.0) if it is not the case that, 
for each $q \in P_C$, at least $d := \lceil \log_{(1 - \frac{b}{k})} 
\epsilon' \rceil$ copies of the \textsf{Q}-form non-terminals 
$T_j \in C \cap H_q$ have been encountered along 
the ancestor history $h$. Otherwise, $\sigma_{K'}^{\epsilon}$ 
chooses deterministically action $b$, and therefore 
generates immediately a child $o''$ of non-terminal $T_{u'}$ 
(from above). In the entire subtree, rooted at 
$o''$, strategy $\tau$ is employed as if the play 
starts in $o''$, where 
$Pr_{T_{u'}}^{\tau}[\bigcap_{q' \in K' - P_C} 
Reach(T_{q'})] \ge \sqrt{1 - \epsilon}$ 
(exist by the induction assumption due to 
$T_{u'} \in F_{K' - P_C}$).

If $o$ is of another \textsf{M}-form non-terminal $T_i \in C$, 
let $\sigma_{K'}^{\epsilon}$ choose uniformly at random 
among actions from statement (M.0) and so in the next generation 
the single generated successor $o'$ is of a non-terminal $T_j \in C$, 
where we proceed to use strategy $\sigma_{K'}^{\epsilon}$ (that 
is being described). If $o$ is of a non-terminal 
$T_i \in C$ of \textsf{L}-form, from statement (L.0) we know 
that in the next generation the single generated successor 
$o'$ is of some non-terminal $T_j \in C \cup D_{K'}$. 
If $T_j \in D_{K'}$, then we use at $o'$ and its subtree of 
descendants the randomized non-static strategy 
from property $(A)_{K'}^j$, that guarantees probability 
$\ge 1 - \epsilon$ of reaching all targets in set $K'$, 
as if the play starts in $o'$. 
If $T_j \in C$, then we proceed by using the same strategy 
$\sigma_{K'}^{\epsilon}$ (that is currently being described) 
at $o'$. If $o$ is of a non-terminal $T_i \in C$ of 
\textsf{Q}-form ($T_i \xrightarrow{1} T_j \; T_r$), 
there are two cases for the two successor children $o'$ 
(of non-terminal $T_j$) and $o''$ (of non-terminal $T_r$):
\begin{itemize}
	\item \textbf{either} property (Q.0) is satisfied, 
i.e., $T_j \in C$ and $T_r \in \bar{Z}_{\{q\}}$, 
for some $q \in K'$. Then, in the next generation, 
we continue using the same strategy $\sigma_{K'}^{\epsilon}$ 
(that is currently being described) 
at $o'$ and for the entire subtree of play, 
rooted at $o''$, strategy $\sigma_{K'}^{\epsilon}$ 
chooses uniformly at random a target non-terminal 
$T_q, q \in K'$, such that $T_r \in \bar{Z}_{\{q\}}$, 
and employs the strategy $\sigma'_{\{q\}}$ from 
statement ($\mathfrak{P}$) as if the play starts at $o''$. 
Note that $Pr_{h(r, T_r)}^{\sigma_{K'}^{\epsilon}}[Reach(T_q)] \ge 
\frac{b}{|P_C|} \ge \frac{b}{k} > 0$, where 
$h(r,T_r)$ refers to the ancestor history for 
the right child $o''$ and where $|P_C| < l = |K'| \le k = |K|$.
	
	\item \textbf{or} property (Q.1) is satisfied. 
Then, in the next generation, we continue using strategy 
$\sigma_{K'}^{\epsilon}$ for $o'$, whereas for $o''$ 
the strategy is irrelevant and an arbitrary one is chosen 
for $o''$ and thereafter in $o''$'s tree of descendants.
\end{itemize}

That concludes the description of the randomized non-static 
strategy $\sigma_{K'}^{\epsilon}$ for non-terminals in MEC, $C$. 
Now we need to show that, indeed, that for any $T_i \in C:\; 
Pr_{T_i}^{\sigma_{K'}^{\epsilon}}[\bigcap_{q \in K'} 
Reach(T_q)] \ge 1 - \epsilon$.

Denote by $w$ the path (in the play) that begins 
at a starting non-terminal $T_i \in C$ and is defined as follows. 
If the current copy $o$ on the path $w$ is of a \textsf{L}-form or 
a \textsf{M}-form non-terminal $T_j \in C$, then $w$ follows 
along the unique successor of $o$ in the play. And if 
the current copy $o$ on path $w$ is of a \textsf{Q}-form 
non-terminal $T_j \in C$ ($T_j \xrightarrow{1} T_{j'} \; T_r$ 
where w.l.o.g. $T_{j'} \in C$), then $w$ follows along 
the child of non-terminal $T_{j'}$. If the current copy $o$ 
on path $w$ is: either of a non-terminal not belonging in $C$; or 
of the non-terminal $T_{u'} \in F_{K' - P_C}$ (from above) and, 
for each $q \in P_C$, at least $d$ copies of 
the \textsf{Q}-form non-terminals in set $C \cap H_q$ 
have already been encountered along $w$ - 
then the path $w$ terminates. 
Denote by $\square C$ the event that path $w$ (as defined) 
is infinite, i.e., path $w$ never terminates, and by 
$\neg \square_{D} C$ (respectively, $\neg \square_{u'} C$) 
the event that path $w$ is finite and terminates 
(according to the above definition of when it can terminate) 
in a copy of a non-terminal in set $D_{K'}$ (respectively, 
in a copy of non-terminal $T_{u'} \in F_{K' - P_C}$). 
Observe that under strategy $\sigma_{K'}^{\epsilon}$ 
for any starting non-terminal $T_i \in C$, 
$P_{T_i}^{\sigma_{K'}^{\epsilon}}[\square C] = 0$, and 
let $p := P_{T_i}^{\sigma_{K'}^{\epsilon}}[\neg \square_D C]$ 
(note that 
$P_{T_i}^{\sigma_{K'}^{\epsilon}}[\neg \square_{u'} C] = 1 - p$).

Now under strategy $\sigma_{K'}^{\epsilon}$ and starting at 
any non-terminal $T_i \in C$, with probability $1$: 
\begin{itemize}
    \item[(i)] either path $w$ 
terminates in a copy $o$ of a non-terminal in set $D_{K'}$, 
for which we already know that there is a strategy 
to reach all target non-terminals from set $K'$ with 
probability $\ge 1 - \epsilon$ (and according to 
$\sigma_{K'}^{\epsilon}$ such a strategy is employed 
at $o$ and its subtree of descendants). Hence, 
in the event of $\neg \square_D C$, 
with probability $\ge 1 - \epsilon$ all target non-terminals 
from set $K'$ are contained in the generated play, i.e., 
$Pr_{T_i}^{\sigma_{K'}^{\epsilon}}[\bigcap_{q \in K'} 
Reach(T_q) \mid \neg \square_D C] \ge 1 - \epsilon$.

    \item[(ii)] or, path $w$ terminates in a copy of a 
non-terminal $T_{u'} \in F_{K' - P_C}$. 
Then, for each $q \in P_C$, with probability $1$ 
(due to $C$ being a MEC and due to the 
description of strategy $\sigma_{K'}^{\epsilon}$) 
at least $d = \lceil \log_{(1 - \frac{b}{k})} \epsilon' \rceil$ 
copies $o$ of the \textsf{Q}-form non-terminals 
$T_j \in C \cap H_q$ were generated along the path $w$. 
And each such copy $o$ generates two children, $o'$ of some 
non-terminal $T_{j'} \in C$ (the successor on path $w$) and 
$o''$ of some non-terminal 
$T_r \in \bar{Z}_{\{q\}}$, where $o''$ has independently 
a positive probability bounded away from zero 
(in fact, $\ge \frac{b}{k}$ due to the uniformly at 
random choice over strategies from statement ($\mathfrak{P}$), 
where, by Proposition \ref{prop:QualMoR-NonReach}, 
the value $b > 0$ does not depend on the history or 
the time when $o''$ is generated) to reach the respective 
target non-terminal $T_q$ in a finite number of generations.
\end{itemize}

So suppose event $\neg \square_{u'} C$ occurs and 
let, for each $q \in P_C$, $Pr_{T_i}^{\sigma_{K'}^{\epsilon}} 
[\Diamond_{\le m} T_q \mid \neg \square_{u'} C]$ denote 
the conditional probability, starting 
at a non-terminal $T_i \in C$ and under the described 
strategy $\sigma_{K'}^{\epsilon}$, to reach target $T_q$ 
with at most $m$ generated copies of the \textsf{Q}-form 
non-terminals in 
set $C \cap H_q$ along the path $w$ in the play, conditioned 
on event $\neg \square_{u'} C$ occurring. 
Note that $\forall q \in P_C:\; 
Pr_{T_i}^{\sigma_{K'}^{\epsilon}}[\Diamond_{\le 1} T_q \mid 
\neg \square_{u'} C] \ge \frac{b}{|P_C|} \ge \frac{b}{k}$. 
That is, because with 
probability $1$ under strategy $\sigma_{K'}^{\epsilon}$, 
starting at a non-terminal $T_i \in C$, a copy $o$ of a 
\textsf{Q}-form non-terminal in set $C \cap H_q$ is generated 
along path $w$ and then there is a probability $\ge \frac{b}{k}$ 
to reach target $T_q$ from the right child of $o$. It follows 
that for any $T_i \in C$ and any $q \in P_C$:
\begin{align*}
    Pr_{T_i}^{\sigma_{K'}^{\epsilon}} 
[\neg \Diamond_{\le d} T_q \mid \neg \square_{u'} C] \le 
\Big( 1 - \frac{b}{k} \Big)^d 
\Leftrightarrow Pr_{T_i}^{\sigma_{K'}^{\epsilon}} 
[\Diamond_{\le d} T_q \mid \neg \square_{u'} C] \ge 1 - \Big (1 - \frac{b}{k} \Big)^d 
\end{align*}

Since $d \ge \log_{(1 - \frac{b}{k})} \epsilon'$, then 
$Pr_{T_i}^{\sigma_{K'}^{\epsilon}}[\Diamond_{\le d} T_q \mid 
\neg \square_{u'} C] \ge 1 - \epsilon'$. 
Then for any $T_i \in C$ and any $q \in P_C$:
\begin{align*}
	& Pr_{T_i}^{\sigma_{K'}^{\epsilon}}[Reach(T_q) \mid 
\neg \square_{u'} C] \ge 
Pr_{T_i}^{\sigma_{K'}^{\epsilon}}[\Diamond_{\le d} T_q \mid 
\neg \square_{u'} C] \ge 1 - \epsilon' \Leftrightarrow \\
    & Pr_{T_i}^{\sigma_{K'}^{\epsilon}}[Reach^{\complement}(T_q) \mid 
\neg \square_{u'} C] \le \epsilon'
\end{align*}

So, by the union bound:
\begin{align}
    & Pr_{T_i}^{\sigma_{K'}^{\epsilon}} 
\Big[ \bigcup_{q \in P_C} Reach^{\complement}(T_q) \Bigm\vert 
\neg \square_{u'} C \Big] \le |P_C| \cdot \epsilon' 
\le k \cdot \epsilon' = 1 - \sqrt{1 - \epsilon} 
\nonumber \\
    & \Leftrightarrow Pr_{T_i}^{\sigma_{K'}^{\epsilon}} 
\Big[ \bigcap_{q \in P_C} Reach(T_q) \Bigm\vert 
\neg \square_{u'} C \Big] \ge \sqrt{1 - \epsilon}
    \label{eq:P_C}
\end{align}

And in some finite number of generations, in a copy of 
the non-terminal $T_u$ along path $w$ action $b \in \Gamma^u$ 
is chosen deterministically, where 
$T_u \xrightarrow{b} T_{u'},\; T_{u'} \in F_{K' - P_C}$. 
There exists $\sigma_{K' - P_C}^{1 - \sqrt{1 - \epsilon}} \in \Psi$ 
such that 
$Pr_{T_{u'}}^{\sigma_{K' - P_C}^{1 - \sqrt{1 - \epsilon}}} 
[\bigcap_{q' \in K' - P_C} Reach(T_{q'})] \ge \sqrt{1 - \epsilon}$. 
Then for any starting non-terminal $T_i \in C$:
\begin{align}
	& Pr_{T_i}^{\sigma_{K'}^{\epsilon}} 
\Big[ \bigcap_{q' \in K' - P_C} Reach(T_{q'}) \Bigm\vert 
\neg \square_{u'} C \Big] = 
Pr_{T_{u'}}^{\sigma_{K' - P_C}^{1 - \sqrt{1 - \epsilon}}} 
\Big[ \bigcap_{q' \in K' - P_C} Reach(T_{q'}) \Big] \ge 
\sqrt{1 - \epsilon}
    \label{eq:K'-P_C}
\end{align}
The equality follows from the fact that there is zero 
probability to reach targets from set $K' - P_C$ before 
path $w$ terminates and also 
from the fact that strategy 
$\sigma_{K'}^{\epsilon}$ utilizes strategy 
$\sigma_{K'- P_C}^{1 - \sqrt{1 - \epsilon}}$ from the 
occurrence of $T_{u'}$ (when event $\neg \square_{u'} C$ 
happens) as if the play starts in it.

Using (\ref{eq:P_C}) and (\ref{eq:K'-P_C}), it follows that 
for any starting non-terminal $T_i \in C$:
\begin{align*}
	& Pr_{T_i}^{\sigma_{K'}^{\epsilon}} 
\Big[ \bigcap_{q \in K'} Reach(T_q) \Bigm\vert 
\neg \square_{u'} C \Big] \\
    & = Pr_{T_i}^{\sigma_{K'}^{\epsilon}} 
\Big[ \bigcap_{q \in P_C} Reach(T_q) \Bigm\vert 
\neg \square_{u'} C \Big] \cdot 
Pr_{T_i}^{\sigma_{K'}} \Big[ \bigcap_{q' \in K' - P_C} 
Reach(T_{q'}) \Bigm\vert \neg \square_{u'} C \Big] \\
	& = Pr_{T_i}^{\sigma_{K'}^{\epsilon}} 
\Big[ \bigcap_{q \in P_C} Reach(T_q) \Bigm\vert 
\neg \square_{u'} C \Big] \cdot 
Pr_{T_{u'}}^{\sigma_{K'- P_C}^{1 - \sqrt{1 - \epsilon}}} 
\Big[ \bigcap_{q' \in K' - P_C} Reach(T_{q'}) \Big] \ge 
(\sqrt{1 - \epsilon})^2 = 1 - \epsilon
\end{align*}
And putting it all together, it follows that for any starting 
non-terminal $T_i \in C$:
\begin{align*}
    & Pr_{T_i}^{\sigma_{K'}^{\epsilon}} \Big[ \bigcap_{q \in K'} 
Reach(T_q) \Big] = Pr_{T_i}^{\sigma_{K'}^{\epsilon}} 
\Big[ \Big( \bigcap_{q \in K'} Reach(T_q) \Big) \cap 
\square C \Big] \\
    & + Pr_{T_i}^{\sigma_{K'}^{\epsilon}} 
\Big[ \Big( \bigcap_{q \in K'} Reach(T_q) \Big) \cap 
\neg \square_{D} C \Big] + Pr_{T_i}^{\sigma_{K'}^{\epsilon}} 
\Big[ \Big( \bigcap_{q \in K'} Reach(T_q) \Big) \cap 
\neg \square_{u'} C \Big] \\
    & = Pr_{T_i}^{\sigma_{K'}^{\epsilon}} 
\Big[ \bigcap_{q \in K'} Reach(T_q) \Bigm\vert 
\neg \square_{D} C \Big] \cdot Pr_{T_i}^{\sigma_{K'}^{\epsilon}} 
\Big[ \neg \square_{D} C \Big] + \\
    & Pr_{T_i}^{\sigma_{K'}^{\epsilon}} 
\Big[ \bigcap_{q \in K'} Reach(T_q) \Bigm\vert 
\neg \square_{u'} C \Big] \cdot Pr_{T_i}^{\sigma_{K'}^{\epsilon}} 
\Big[ \neg \square_{u'} C \Big] \ge (1 - \epsilon) \cdot p + 
(1 - \epsilon) \cdot (1 - p) = 1 - \epsilon
\end{align*}

Now the second reason, why a MEC, $C$, in $G[F_{K'}]$ 
was added to $F_{K'}^0$ at step II.8., is if $P_C = K'$. 
Consider any finite ancestor history $h$, that starts at 
a non-terminal $T_v \in C$ and that all non-terminals in $h$ 
belong to the MEC, $C$. 
Let $o$ denote the non-terminal copy at the end of 
the ancestor history $h$. 
If $o$ is of a \textsf{L}-form or \textsf{Q}-form non-terminal 
in $C$, let $\sigma_{K'}^{\epsilon}$ behave the same way 
as was described before. And if $o$ is of a 
\textsf{M}-form non-terminal $T_i \in C$, let 
$\sigma_{K'}^{\epsilon}$ choose uniformly at random 
among actions from statement (M.0). So with probability $1$: 
either a copy of a \textsf{L}-form non-terminal in $C$ 
generates a child $o'$ of some non-terminal in set $D_{K'}$, 
where $\sigma_{K'}^{\epsilon}$ employs a strategy at $o'$ 
and its subtree of descendants such that all targets in set $K'$ 
are reached with probability $\ge 1 - \epsilon$ (exists by 
the induction assumption); or, for each $q \in P_C = K'$, 
infinitely often copies of 
the \textsf{Q}-form non-terminals $T_j \in C \cap H_q$ 
are observed. 
In the latter case, it follows that, for each $q \in P_C = K'$, 
infinitely many independent copies $o'$ of non-terminals 
$T_r \in \bar{Z}_{\{q\}}$ are generated, each of which 
has independently a positive probability bounded away from zero 
(again, $\ge \frac{b}{k}$ where, 
by Proposition \ref{prop:QualMoR-NonReach}, 
value $b > 0$ does not depend on the history or the time 
when entity $o'$ is generated) to reach 
the corresponding target non-terminal $T_q$ 
in a finite number of generations. 
Hence for any $T_v \in C$, it is satisfied that 
$Pr_{T_v}^{\sigma_{K'}^{\epsilon}}[\bigcap_{q \in K'} 
Reach(T_q)] \ge 1 - \epsilon$.

Therefore, for each type $T_i$ in some MEC, 
$C \subseteq F_{K'}^0$, property $(A)_{K'}^i$ is satisfied.

\medskip
Now consider the non-terminals $T_i$ added to set $F_{K'}$ 
in step II.9. during the last iteration of the inner loop.
\begin{enumerate}[label=(\roman*)]
	\item If $T_i$ is of \textsf{L}-form, then 
by statement (L) we know that with probability $1$ 
a copy $o$ of non-terminal $T_i$ 
in the next generation produces a single successor $o'$ 
of some non-terminal $T_j \in F_{K'} \cup D_{K'}$, 
where by induction $(A)_{K'}^j$ holds. 
So using, for any given $\epsilon > 0$, the strategy 
$\sigma_{K'}^{\epsilon}$ from the induction assumption 
for each such non-terminal $T_j$ in the next generation 
as if the play starts in it, 
then property $(A)_{K'}^i$ is also satisfied. 

    \item If $T_i$ is of \textsf{M}-form, then 
by statement (M), $\exists a^* \in \Gamma^i:\; 
T_i \xrightarrow{a^*} T_j,\; T_j \in F_{K'}$. 
Let $h := T_i(u, T_j)$. 
So, for every $\epsilon > 0$, combining the already 
described strategy $\sigma_{K'}^{\epsilon}$ 
for non-terminal $T_j$ (from the induction assumption), 
as if the play starts 
in it, with the initial local 
choice of choosing deterministically action $a^*$, 
starting at a non-terminal $T_i$, we obtain an 
augmented strategy $\sigma_{K'}^{\epsilon}$ for 
a starting non-terminal $T_i$ such that 
$Pr_{T_i}^{\sigma_{K'}^{\epsilon}}[\bigcap_{q \in K'} 
Reach(T_q)] =  Pr_h^{\sigma_{K'}^{\epsilon}} 
[\bigcap_{q \in K'} Reach(T_q)] = 
Pr_{T_j}^{\sigma_{K'}^{\epsilon}}[\bigcap_{q \in K'} 
Reach(T_q)] \ge 1 - \epsilon$, i.e., $(A)_{K'}^i$ holds.

	\item If $T_i$ is of \textsf{Q}-form 
(i.e., $T_i \xrightarrow{1} T_j \; T_r$), then, 
by statement (Q), w.l.o.g. $T_j \in F_{K'}$, 
where we already know that, for every $\epsilon > 0$,
there is a strategy $\sigma_{K'}^{\epsilon}$ such that 
$Pr_{T_j}^{\sigma_{K'}^{\epsilon}}[\bigcap_{q \in K'} 
Reach(T_q)] \ge 1 - \epsilon$. Let $h_l := T_i(l, T_j)$ 
and $h_r := T_i(r, T_r)$. 
Augmenting strategy $\sigma_{K'}^{\epsilon}$ to be used 
from the next generation from the child of non-terminal $T_j$ 
as if the play starts in it and using an arbitrary strategy 
from the child of non-terminal $T_r$, it follows that 
$Pr_{T_i}^{\sigma_{K'}^{\epsilon}}[\bigcup_{q \in K'} 
Reach^{\complement}(T_q)] \le Pr_{h_l}^{\sigma_{K'}^{\epsilon}} 
[\bigcup_{q \in K'} Reach^{\complement}(T_q)] \cdot 
Pr_{h_r}^{\sigma_{K'}^{\epsilon}}[\bigcup_{q \in K'} 
Reach^{\complement}(T_q)] \le 
Pr_{T_j}^{\sigma_{K'}^{\epsilon}}[\bigcup_{q \in K'} 
Reach^{\complement}(T_q)] \le \epsilon$, resulting 
in property $(A)_{K'}^i$ also being satisfied.
\end{enumerate}
\end{proof}

This completes the proof of 
Theorem \ref{theorem:QualMoR-LS-Reach} 
and the analysis of the limit-sure algorithm. 
The proof of Lemma \ref{lemma:property_A-LS} describes how to construct, 
for any subset $K' \subseteq K$ and any given $\epsilon > 0$, 
the witness strategy $\sigma_{K'}^{\epsilon}$ for the non-terminals 
in set $F_{K'}$. These non-static strategies $\sigma_{K'}^{\epsilon}$ 
are described as functions that map finite ancestor histories 
belonging to the controller to distributions over actions 
for the current non-terminal in the ancestor history, and 
can be described in such a form 
in time $(\log \frac{1}{\epsilon})^{O(1)} \cdot 4^{k}  \cdot |\mathcal{A}|^{O(1)}$.
\end{proof}

\section{Algorithm for deciding $\stackrel{?}{\exists} \sigma \in \Psi:\; Pr_{T_i}^{\sigma}[\bigcap_{q \in K} Reach(T_q)] = 1$}

In this section we present an algorithm for solving 
the qualitative almost-sure multi-target reachability problem 
for an OBMDP, $\mathcal{A}$, i.e., given a set $K \subseteq [n]$ 
of $k = |K|$ target non-terminals and a starting non-terminal 
$T_i$, deciding whether there is a strategy for the player 
under which the probability of generating a tree that 
contains all target non-terminals from set $K$ is $1$. 
The algorithm runs in time $4^k \cdot |\mathcal{A}|^{O(1)}$, and 
hence is fixed-parameter tractable with respect to $k$.

As in the previous section, first as a preprocessing step, 
for each subset of 
target non-terminals $K' \subseteq K$, we compute the set 
$Z_{K'} := \{T_i \in V \mid \forall \sigma \in \Psi:\; Pr_{T_i}^{\sigma}[\bigcap_{q \in K'} Reach(T_q)] = 0\}$, 
using the algorithm from Proposition 
\ref{prop:QualMoR-NonReach}. Let also denote by $AS_q$, 
for every $q \in K$, the set of non-terminals types $T_j$ 
(including the target non-terminal $T_q$ itself) for which 
there exists a strategy $\tau$ such that 
$Pr_{T_j}^{\tau}[Reach(T_q)] = 1$. 
These sets can be computed in P-time by applying the 
algorithm from \cite[Theorem 9.3]{ESY-icalp15-IC} to each 
target non-terminal $T_q,\; q \in K$. Recall that it was shown 
in \cite{ESY-icalp15-IC} that for OBMDPs with a single 
target the almost-sure and limit-sure reachability 
problems coincide.

After this preprocessing step, we apply the algorithm in 
Figure \ref{fig:QualMoR-AS-Reach} to identify 
the non-terminals $T_i$ for which there is a strategy 
$\sigma^*$ for the player such that $Pr_{T_i}^{\sigma^*}[\bigcap_{q \in K} Reach(T_q)] = 1$. 
Again $K'_{-i}$ denotes the set $K' - \{i\}$.

\begin{figure}[hp!]
	\small
	\begin{enumerate}[label=\Roman*.]
		\item Let $F_{\{q\}} := AS_q$, for each $q \in K$. $F_{\emptyset} := V$.
		
		\item For $l = 2 \ldots k$: \\
			\hspace*{0.1in} For every subset of target 
			non-terminals $K' \subseteq K$ of size $|K'| = l$: 
		\begin{enumerate}[label=\arabic*.]
			\item $D_{K'} := \{T_i \in V - Z_{K'} \mid $ one of the following holds:
			\begin{itemize}
				\item[-] $T_i$ is of \textsf{L}-form where 
				$i \in K'$, $T_i \not\rightarrow \varnothing$ 
				and $\forall T_j \in V$: if $T_i \rightarrow T_j$, then $T_j \in F_{K'_{-i}}$.
				
				\item[-] $T_i$ is of \textsf{M}-form where 
				$i \in K'$ and $\exists a^* \in \Gamma^i: T_i \xrightarrow{a^*} T_j,\; T_j \in F_{K'_{-i}}$.
				
				\item[-] $T_i$ is of \textsf{Q}-form 
				($T_i \xrightarrow{1} T_j \; T_r$) where 
				$i \in K'$ and $\exists K_L \subseteq K'_{-i}: T_j \in F_{K_L} \wedge T_r \in F_{K'_{-i} - K_L}$.
				
				\item[-] $T_i$ is of \textsf{Q}-form 
				($T_i \xrightarrow{1} T_j \; T_r$) where 
				$\exists K_L \subset K'\; (K_L \not= \emptyset): T_j \in F_{K_L} \wedge T_r \in F_{K' - K_L}.\}$
			\end{itemize}
			
			\item Repeat until no change has occurred to $D_{K'}$:
			\begin{enumerate}[label=(\alph*)]
				\item add $T_i \not\in D_{K'}$ to $D_{K'}$, if 
				of \textsf{L}-form, 
				$T_i \not\rightarrow \varnothing$ and 
				$\forall T_j \in V$: if $T_i \rightarrow T_j$, 
				then $T_j \in D_{K'}$.
				
				\item add $T_i \not\in D_{K'}$ to $D_{K'}$, if 
				of \textsf{M}-form and 
				$\exists a^* \in \Gamma^i: T_i \xrightarrow{a^*} T_j,\; T_j \in D_{K'}$.
				
				\item add $T_i \not\in D_{K'}$ to $D_{K'}$, if 
				of \textsf{Q}-form 
				($T_i \xrightarrow{1} T_j \; T_r$) and 
				$T_j \in D_{K'} \vee T_r \in D_{K'}$.
			\end{enumerate}
			
			\item Let $X := V - (D_{K'} \cup Z_{K'})$.
			
			\item Initialize $S_{K'} := \{T_i \in X \mid $ 
			either $i \in K'$, or $T_i$ is of \textsf{L}-form 
			and $T_i \rightarrow \varnothing \vee T_i \rightarrow T_j,\; T_j \in Z_{K'} \} \; \cup \; \bigcup_{\emptyset \subset K'' \subset K'} (X \cap S_{K''})$.
			
			\item Repeat until no change has occurred to $S_{K'}$:
			\begin{enumerate}[label=(\alph*)]
        		\item add $T_i \in X - S_{K'}$ to $S_{K'}$, if 
        		of \textsf{L}-form and $T_i \rightarrow T_j,\; T_j \in S_{K'} \cup Z_{K'}$.
        			
            	\item add $T_i \in X - S_{K'}$ to $S_{K'}$, if 
            	of \textsf{M}-form and 
            	$\forall a \in \Gamma^i:\; T_i \xrightarrow{a} T_j,\; T_j \in S_{K'} \cup Z_{K'}$.
            		
            	\item add $T_i \in X - S_{K'}$ to $S_{K'}$, if 
            	of \textsf{Q}-form 
            	($T_i \xrightarrow{1} T_j \; T_r$) and 
            	$T_j \in S_{K'} \cup Z_{K'}\; \wedge \; T_r \in S_{K'} \cup Z_{K'}$.
			\end{enumerate}
			
			\item $\mathcal{C} \leftarrow$ SCC decomposition of $G[X-S_{K'}]$.
			
			\item For every $q \in K'$, let 
			$H_q := \{T_i \in X - S_{K'} \mid T_i$ is of \textsf{Q}-form ($T_i \xrightarrow{1} T_j \; T_r$) 
			and $( (T_j \in X - S_{K'} \wedge T_r \in \bar{Z}_{\{q\}} ) \vee (T_j \in \bar{Z}_{\{q\}} \wedge T_r \in X - S_{K'}) ) \}$.
			
			\item Let $F_{K'} := \bigcup \; \{\cup_{q \in K'} (H_q \cap C) \mid C \in \mathcal{C}$ s.t. $\forall q' \in K': H_{q'} \cap C \not= \emptyset \}$.
			
			\item Repeat until no change has occurred to $F_{K'}$:
			\begin{enumerate}[label=(\alph*)]
 	       		\item add $T_i \in X - (S_{K'} \cup F_{K'})$ to 
 	       		$F_{K'}$, if of \textsf{L}-form and 
 	       		$T_i \rightarrow T_j,\; T_j \in F_{K'} \cup D_{K'}$.
 	       		
    		   	\item add $T_i \in X - (S_{K'} \cup F_{K'})$ to 
    		   	$F_{K'}$, if of \textsf{M}-form and 
    		   	$\exists a^* \in \Gamma^i: T_i \xrightarrow{a^*} T_j,\; T_j \in F_{K'}$.
    		   	
    		   	\item add $T_i \in X - (S_{K'} \cup F_{K'})$ to 
    		   	$F_{K'}$, if of \textsf{Q}-form 
    		   	($T_i \xrightarrow{1} T_j \; T_r$) and 
    		   	$T_j \in F_{K'} \vee T_r \in F_{K'}$.
			\end{enumerate}
			
			\item If $X \not= S_{K'} \cup F_{K'}$, let $S_{K'} := X - F_{K'}$ and go to step 5.
			
			\item Else, i.e., if $X = S_{K'} \cup F_{K'}$, let $F_{K'} := F_{K'} \cup D_{K'}$.
		\end{enumerate}
		
		\item \textbf{Output} $F_K$.
	\end{enumerate}
		\caption{Algorithm for almost-sure 
	multi-target reachability. The output is the set 
	$F_K = \{T_i \in V \mid \exists \sigma \in \Psi:\; 
	Pr_{T_i}^{\sigma}[\bigcap_{q \in K} Reach(T_q)] = 1\}$.}
	\label{fig:QualMoR-AS-Reach}
\end{figure}

\begin{theorem}
    The algorithm in Figure \ref{fig:QualMoR-AS-Reach} 
computes, given an OBMDP, $\mathcal{A}$, 
and a set $K \subseteq [n]$ of $k = |K|$ target non-terminals, 
for each subset $K' \subseteq K$, the set of non-terminals 
$F_{K'} := \{T_i \in V \mid \exists \sigma \in \Psi:\; Pr_{T_i}^{\sigma}[\bigcap_{q \in K'} Reach(T_q)] = 1\}$. 
The algorithm runs in time 
$4^k \cdot |\mathcal{A}|^{O(1)}$. Moreover, for each 
$K' \subseteq K$, the algorithm can also be augmented 
to compute a randomized non-static strategy 
$\sigma_{K'}^*$ such that 
$Pr_{T_i}^{\sigma_{K'}^*}[\bigcap_{q \in K'} Reach(T_q)] = 1$ 
for all non-terminals $T_i \in F_{K'}$.
	\label{theorem:QualMoR-AS-Reach}
\end{theorem}

\begin{proof}
    We will refer to the loop executing steps II.5. 
through II.10. for a specific subset $K' \subseteq K$ 
as the ``inner" loop and the iteration through all subsets 
of $K$ as the ``outer" loop. Clearly the inner loop 
terminates, due to step II.10. always adding at least one 
non-terminal to set $S_{K'}$ and step II.11. 
eventually executing. 
The running time of the algorithm follows from the facts 
that the outer loop executes for $2^k$ iterations and inside 
each iteration of the outer loop, 
steps II.1. and II.4. require time at most 
$2^k \cdot |\mathcal{A}|^{O(1)}$ and 
the inner loop executes for at most $|V|$ iterations, 
where during each inner loop iteration the nested loops 
execute in time at most $|\mathcal{A}|^{O(1)}$.

For the proof of correctness, we show that for every subset 
of target non-terminals $K' \subseteq K$, $F_{K'}$ 
(from the decomposition $V = F_{K'} \cup S_{K'} \cup Z_{K'}$) 
is the set of non-terminals $T_i$ for which 
the following property holds:
\begin{center}
	$(A)_{K'}^i$: $\exists \sigma_{K'} \in \Psi$ such that 
	$Pr_{T_i}^{\sigma_{K'}}[\bigcap_{q \in K'} Reach(T_q)] = 1$.
\end{center}
Otherwise, if $T_i \in S_{K'}$, then the following property holds: 
\begin{center}
	$(B)_{K'}^i$: $\forall \sigma \in \Psi:\; Pr_{T_i}^{\sigma}[\bigcap_{q \in K'} Reach(T_q)] < 1 \Leftrightarrow Pr_{T_i}^{\sigma}[\bigcup_{q \in K'} Reach^{\complement}(T_q)] > 0$, i.e., 
	the probability of generating a tree that contains 
	at least one copy for each of the $T_q, q \in K'$ 
	target non-terminals, is $< 1$.
\end{center}
Clearly, for non-terminals $T_i \in Z_{K'}$, property 
$(B)_{K'}^i$ holds because, 
by Proposition \ref{prop:QualMoR-NonReach}, 
$\forall \sigma \in \Psi:\; 
Pr_{T_i}^{\sigma}[\bigcap_{q \in K'} Reach(T_q)] = 0 < 1$. 
Finally, the answer for the full set of targets is $F := F_K$.
	
As in the proof from the previous section, 
we base this proof on an induction on the size of subset $K'$, 
i.e. on the time of computing sets $S_{K'}$ and $F_{K'}$ 
for $K' \subseteq K$. And in the process, 
for each subset $K' \subseteq K$ of target non-terminals, 
we construct a randomized \textit{non-static} strategy 
$\sigma_{K'}$ for the player that ensures 
$Pr_{T_i}^{\sigma_{K'}}[\bigcap_{q \in K'} Reach(T_q)] = 1$
for each non-terminal $T_i \in F_{K'}$. In the end, 
$\sigma := \sigma_{K}$ is the strategy that guarantees 
almost-sure reachability of all given targets in the same play.

To begin with, observe that clearly for any subset of 
target non-terminals, $K' := \{q\} \subseteq K$, of 
size $l = 1$, each non-terminal $T_i \in F_{\{q\}}$ 
(respectively, $T_i \in V - F_{\{q\}}$) satisfies property 
$(A)_{\{q\}}^i$ (respectively, $(B)_{\{q\}}^i$), due to step I. 
and the definition of the $AS_q, q \in K$ sets. 
Hence, for each such subset $\{q\} \subseteq K$, there is 
a strategy $\sigma_{\{q\}}$ such that 
$\forall T_i \in F_{\{q\}}: \; Pr_{T_i}^{\sigma_{\{q\}}}[Reach(T_q)] = 1$. Moreover, 
by \cite[Theorem 9.4]{ESY-icalp15-IC} this strategy 
$\sigma_{\{q\}}$ is non-static and deterministic. 
Analysing subset $K'$ of target non-terminals of size $l$ 
as part of step II., assume that, for every $K'' \subset K'$ 
of size $\le l - 1$, sets $S_{K''}$ and $F_{K''}$ 
have already been computed, and for each non-terminal 
$T_j$ belonging to set $F_{K''}$ (respectively, set $S_{K''}$) 
property $(A)_{K''}^j$ (respectively, $(B)_{K''}^j$) holds. 
That is, by induction assumption, for each $K'' \subset K'$, 
there is a randomized non-static strategy $\sigma_{K''}$ 
such that for any $T_j \in F_{K''}$: 
$Pr_{T_j}^{\sigma_{K''}}[\bigcap_{q \in K''} Reach(T_q)] = 1$, and for any $T_j \in S_{K''}$: 
$\forall \sigma \in \Psi,\; Pr_{T_j}^{\sigma}[\bigcap_{q \in K''} Reach(T_q)] < 1$. We now need to show that at end of 
the inner loop analysis of subset $K'$, property 
$(A)_{K'}^i$ (respectively, $(B)_{K'}^i$) holds for every 
non-terminal $T_i \in F_{K'}$ (respectively, $T_i \in S_{K'}$).

First we show that property $(A)_{K'}^i$ holds for each 
non-terminal $T_i$ belonging to set $D_{K'}$ 
($\subseteq F_{K'}$), pre-computed prior to the execution 
of the inner loop for subset $K'$.

\begin{lemma}
	Every non-terminal $T_i \in D_{K'}$ satisfies property 
	$(A)_{K'}^i$.
	\label{lemma:D_K'-analysis-AS}
\end{lemma}

\begin{proof}
    The lemma is proved via a nested induction based on 
the time of a non-terminal being added to set $D_{K'}$. 
Consider the base case where $T_i \in D_{K'}$ is a 
non-terminal, added at the initialization step II.1.
\begin{enumerate}[label=(\roman*)]
	\item Suppose $T_i$ is of \textsf{L}-form where 
$i \in K'$ and for all associated rules a child is 
generated that is of a non-terminal $T_j \in F_{K'_{-i}}$, 
where property $(A)_{K'_{-i}}^j$ holds. 
Then using the witness strategy from property $(A)_{K'_{-i}}^j$ 
for all such non-terminals $T_j$ in the next generation 
as if the play starts in it and, since target non-terminal 
$T_i$ is already reached, clearly property $(A)_{K'}^i$ holds.
		
	\item Suppose $T_i$ is of \textsf{M}-form where 
$i \in K'$ and $\exists a^* \in \Gamma^i$ such that 
$T_i \xrightarrow{a^*} T_j,\; T_j \in F_{K'_{-i}}$, where 
property $(A)_{K'_{-i}}^j$ holds. Let $h := T_i(u, T_j)$. 
Then, by combining the witness strategy $\sigma_{K'_{-i}}$ 
from the induction assumption for non-terminal $T_j$, as 
if the play starts in it, with the initial local choice of choosing 
deterministically action $a^*$ starting at 
a non-terminal $T_i$, we obtain a combined strategy 
$\sigma_{K'}$ such that starting at 
a (target) non-terminal $T_i$, we satisfy 
$Pr_{T_i}^{\sigma_{K'}}[\bigcap_{q \in K'} Reach(T_q)] = 
Pr_{T_i}^{\sigma_{K'}}[\bigcap_{q \in K'_{-i}} Reach(T_q) \mid 
Reach(T_i)] \cdot Pr_{T_i}^{\sigma_{K'}}[Reach(T_i)] = 
Pr_{T_i}^{\sigma_{K'}}[\bigcap_{q \in K'_{-i}} Reach(T_q)] = 
Pr_h^{\sigma_{K'}}[\bigcap_{q \in K'_{-i}} Reach(T_q)] = 
Pr_{T_j}^{\sigma_{K'_{-i}}}[\bigcap_{q \in K'_{-i}} Reach(T_q)] = 1$.
		
	\item Suppose $T_i$ is of \textsf{Q}-form 
($T_i \xrightarrow{1} T_j \; T_r$) where $i \in K'$ and 
there exists a split of the rest of the target non-terminals, 
implied by $K_L \subseteq K'_{-i}$ and $K'_{-i} - K_L$, 
such that $T_j \in F_{K_L} \wedge T_r \in F_{K'_{-i} - K_L}$. 
Let $h_l := T_i(l, T_j)$ and $h_r := T_i(r, T_r)$. 
By combining the two witness strategies $\sigma_{K_L}$ and 
$\sigma_{K'_{-i} - K_L}$ from the induction assumption 
for non-terminals $T_j$ and $T_r$, respectively, to be 
used from the next generation as if the play starts in it, 
and the fact that target $T_i$ is reached (since $T_i$ is the 
starting non-terminal), it follows that 
$\exists \sigma_{K'} \in \Psi$ such that 
$Pr_{T_i}^{\sigma_{K'}}[\bigcap_{q \in K'} Reach(T_q)] = 
Pr_{T_i}^{\sigma_{K'}}[\bigcap_{q \in K'_{-i}} Reach(T_q)] \ge Pr_{h_l}^{\sigma_{K'}}[\bigcap_{q \in K_L} Reach(T_q)] \cdot Pr_{h_r}^{\sigma_{K'}}[\bigcap_{q \in K'_{-i} - K_L} Reach(T_q)] = Pr_{T_j}^{\sigma_{K_L}}[\bigcap_{q \in K_L} Reach(T_q)] \cdot Pr_{T_r}^{\sigma_{K'_{-i} - K_L}}[\bigcap_{q \in K'_{-i} - K_L} Reach(T_q)] = 1$.
		
	\item Suppose $T_i$ is of \textsf{Q}-form 
($T_i \xrightarrow{1} T_j \; T_r$) where there exists a proper 
split of the target non-terminals from set $K'$, implied by 
$K_L \subset K'$ (where $K_L \not= \emptyset$) and $K' - K_L$, 
such that $T_j \in F_{K_L} \wedge T_r \in F_{K' - K_L}$. 
Combining the two witness strategies from the induction 
assumption for non-terminals $T_j, T_r$ in the same way 
as in (iii), it follows that there exists a strategy 
$\sigma_{K'} \in \Psi$ such that property $(A)_{K'}^i$ holds.
\end{enumerate}

Now consider non-terminals $T_i$ added to set $D_{K'}$ at 
step II.2., i.e., the inductive step. If non-terminal $T_i$ 
is of \textsf{L}-form, then all rules, associated with it, 
generate a child of a non-terminal $T_j$ already in $D_{K'}$, 
for which $(A)_{K'}^j$ holds by the (nested) induction. 
Hence, $(A)_{K'}^i$ clearly also holds for 
the same reason as in (i) above.

If non-terminal $T_i$ is of \textsf{M}-form, then 
$\exists a^* \in \Gamma^i: T_i \xrightarrow{a^*} T_j, T_j \in D_{K'}$. 
Again let $h := T_i(u, T_j)$. 
By combining the witness strategy $\sigma_{K'}$ for 
non-terminal $T_j$ (from the nested induction assumption), 
as if the play starts in it, 
with the initial local choice of choosing deterministically 
action $a^*$ starting at a non-terminal $T_i$, we obtain an 
augmented strategy $\sigma_{K'}$ for a starting non-terminal 
$T_i$ such that $Pr_{T_i}^{\sigma_{K'}}[\bigcap_{q \in K'} 
Reach(T_q)] = Pr_h^{\sigma_{K'}}[\bigcap_{q \in K'} Reach(T_q)] = 
Pr_{T_j}^{\sigma_{K'}}[\bigcap_{q \in K'} Reach(T_q)] = 1$.

If $T_i$ is of \textsf{Q}-form ($T_i \xrightarrow{1} T_j \; T_r$), 
then either $T_j \in D_{K'}$ or $T_r \in D_{K'}$, i.e., 
$\exists \sigma_{K'} \in \Psi$ such that 
$Pr_{T_y}^{\sigma_{K'}}[\bigcap_{q \in K'} Reach(T_q)] = 1 \Leftrightarrow Pr_{T_y}^{\sigma_{K'}}[\bigcup_{q \in K'} Reach^{\complement}(T_q)] = 0$, where $y \in \{j, r\}$. 
Let $h_y := T_i(x, T_y)$ and 
$h_{\bar{y}} := T_i(\bar{x}, T_{\bar{y}})$, where 
$\bar{y} \in \{j, r\} - \{y\}$, $x \in \{l, r\}$ and 
$\bar{x} \in \{l, r\} - \{x\}$. By augmenting 
this $\sigma_{K'}$ to be used from the next generation 
from the child of non-terminal $T_y$ as if the play starts in it 
and using an arbitrary strategy from the child of non-terminal 
$T_{\bar{y}}$, it follows that 
$Pr_{T_i}^{\sigma_{K'}}[\bigcup_{q \in K'} 
Reach^{\complement}(T_q)] \le Pr_{h_y}^{\sigma_{K'}} 
[\bigcup_{q \in K'} Reach^{\complement}(T_q)] \cdot 
Pr_{h_{\bar{y}}}^{\sigma_{K'}}[\bigcup_{q \in K'} 
Reach^{\complement}(T_q)] \le Pr_{T_y}^{\sigma_{K'}} 
[\bigcup_{q \in K'} Reach^{\complement}(T_q)] = 0$, i.e., 
property $(A)_{K'}^i$ is satisfied.
\end{proof}

\medskip

Next we show that if $T_i \in S_{K'}$, then property $(B)_{K'}^i$ holds.
\begin{lemma}
	Every non-terminal $T_i \in S_{K'}$ satisfies property $(B)_{K'}^i$.
	\label{lemma:property_B-AS}
\end{lemma}

\begin{proof}
    This can be done again via another (nested) induction, 
based on the time a non-terminal is added to set $S_{K'}$. 
That is, assuming all non-terminals $T_j$, added already to 
set $S_{K'}$ in previous iterations and steps of the inner loop, 
satisfy property $(B)_{K'}^j$, then we show that for 
a new addition $T_i$ to set $S_{K'}$, property $(B)_{K'}^i$ 
is also satisfied.

Consider the initialized set $S_{K'}$ of non-terminals $T_i$ 
constructed at step II.4. 

If $T_i$ is of \textsf{L}-form, where 
$T_i \rightarrow \varnothing \vee T_i \rightarrow T_j,\; T_j \in Z_{K'}$, 
then with a positive probability non-terminal $T_i$ 
immediately either does not generate a child at all or generates 
a child of non-terminal $T_j \in Z_{K'}$, for which we already 
know that $(B)_{K'}^j$ holds. Clearly, this results in 
$(B)_{K'}^i$ being also satisfied.

If, for some subset $K'' \subset K'$, $T_i \in S_{K''}$, 
then $\forall \sigma \in \Psi:\; 
Pr_{T_i}^{\sigma}[\bigcup_{q \in K''} 
Reach^{\complement}(T_q)] > 0$ 
(i.e., property $(B)_{K''}^i$). 
But, $\forall \sigma \in \Psi:\; 
Pr_{T_i}^{\sigma}[\bigcup_{q \in K'} 
Reach^{\complement}(T_q)] \ge 
Pr_{T_i}^{\sigma}[\bigcup_{q \in K''} 
Reach^{\complement}(T_q)] > 0$, so property $(B)_{K'}^i$ 
also holds. Note that if, for some subset $K'' \subset K'$, 
$T_i \in Z_{K''}$, then $T_i \in Z_{K'}$ and so already 
$T_i \not\in X$.

And if $T_i$ is a target non-terminal in set $K'$, then due to not 
being added to set $D_{K'}$ in step II.1. 
it follows that: (1) if of \textsf{L}-form, it generates with a 
positive probability a child of a non-terminal 
$T_j \in S_{K'_{-i}} \cup Z_{K'_{-i}}$, for which 
$(B)^j_{K'_{-i}}$ holds; (2) if of \textsf{M}-form, 
irrespective of the strategy it generates a child of a non-terminal 
$T_j \in S_{K'_{-i}} \cup Z_{K'_{-i}}$, for which again 
$(B)^j_{K'_{-i}}$ holds; (3) and if of \textsf{Q}-form, 
it generates two children of non-terminals $T_j, T_r$, 
for which no matter how we split the rest of 
the target non-terminals in set $K'_{-i}$ 
(into subsets $K_L \subseteq K'_{-i}$ and $K'_{-i} - K_L$), either 
$(B)^j_{K_L}$ holds or $(B)^r_{K'_{-i} - K_L}$ holds. 
In other words, a target $T_i$ in the initial set $S_{K'}$ 
has no strategy to ensure that the rest of the 
target non-terminals are reached with probability 1 
(the reasoning behind this last statement is the same as 
the arguments in (i) - (iii) below, since for 
a starting (target) non-terminal $T_i$: 
$\forall \sigma \in \Psi:\; Pr_{T_i}^{\sigma} 
[\bigcap_{q \in K'} Reach(T_q)] = 
Pr_{T_i}^{\sigma}[\bigcap_{q \in K'_{-i}} Reach(T_q)]$).

Observe that by the end of step II.4. all target non-terminals 
$T_q, q \in K'$ belong either to set $D_{K'}$ or set $S_{K'}$. 
Now consider a non-terminal $T_i$ added to set $S_{K'}$ in 
step II.5. during some iteration of the inner loop.

\begin{enumerate}[label=(\roman*)]
	\item Suppose $T_i$ is of \textsf{L}-form. Then 
$T_i \rightarrow T_j,\; T_j \in S_{K'} \cup Z_{K'}$, where 
property $(B)_{K'}^j$ holds. So regardless of the strategy 
$\sigma$ for the player, there is a positive probability 
to generate a child of the above non-terminal $T_j$, where 
$Pr_{T_j}^{\sigma}[\bigcup_{q \in K'} 
Reach^{\complement}(T_q)] > 0$. Let $h := T_i(u, T_j)$. 
But note that, $\forall \sigma \in \Psi$: 
$Pr_{T_i}^{\sigma}[\bigcup_{q \in K'} 
Reach^{\complement}(T_q)] \ge p_{ij} \cdot 
Pr_h^{\sigma}[\bigcup_{q \in K'} Reach^{\complement}(T_q)] > 0$ 
if and only if $\forall \sigma \in \Psi:\; p_{ij} \cdot 
Pr_{T_j}^{\sigma}[\bigcup_{q \in K'} 
Reach^{\complement}(T_q)] > 0$, 
where $p_{ij} > 0$ is the probability of the rule 
$T_i \xrightarrow{p_{ij}} T_j$. And since the latter part 
of the statement holds, then 
the former (i.e., property $(B)_{K'}^i$) is satisfied.

	\item Suppose $T_i$ is of \textsf{M}-form. Then 
$\forall a \in \Gamma^i:\; T_i \xrightarrow{a} T_j,\; 
T_j \in S_{K'} \cup Z_{K'}$. So irrelevant of strategy 
$\sigma$ for the player, 
starting in a non-terminal $T_i$, the next generation 
surely consists of some non-terminal $T_j$ such that 
$\forall \sigma \in \Psi:\; 
Pr_{T_j}^{\sigma}[\bigcap_{q \in K'} Reach(T_q)] < 1$. 
Clearly $\forall \sigma \in \Psi:\; 
Pr_{T_i}^{\sigma}[\bigcap_{q \in K'} 
Reach(T_q)] \le \max_{\{T_j \in S_{K'} \cup Z_{K'}\}} 
Pr_{T_i(u, T_j)}^{\sigma}[\bigcap_{q \in K'} 
Reach(T_q)] < 1$ (i.e., property $(B)_{K'}^i$) 
if and only if $\forall \sigma \in \Psi:\; 
\max_{\{T_j \in S_{K'} \cup Z_{K'}\}} 
Pr_{T_j}^{\sigma}[\bigcap_{q \in K'} Reach(T_q)] < 1$, 
where the latter is satisfied.
	
	\item Suppose $T_i$ is of \textsf{Q}-form 
(i.e., $T_i \xrightarrow{1} T_j \; T_r$). 
Then $T_j, T_r \in S_{K'} \cup Z_{K'}$, i.e., 
both $(B)_{K'}^j$ and $(B)_{K'}^r$ are satisfied. We know that: 
\begin{enumerate}
    \setlength{\topsep}{0em}
    \item[1)] Neither of the two children can 
single-handedly reach all target non-terminals 
from set $K'$ with probability $1$. 
That is, for every $\sigma \in \Psi$, 
$Pr_{T_j}^{\sigma}[\bigcap_{q \in K'} Reach(T_q)] < 1$ and 
$Pr_{T_r}^{\sigma}[\bigcap_{q \in K'} Reach(T_q)] < 1$.

    \item[2)] Moreover, since $T_i$ was not added to 
set $D_{K'}$ in step II.1., then $\forall K_L \subset K'$ 
(where $K_L \not= \emptyset$) either $(B)_{K_L}^j$ holds 
(i.e., $T_j \not\in F_{K_L}$) or $(B)_{K' - K_L}^r$ holds 
(i.e., $T_r \not\in F_{K' - K_L}$), i.e., either 
$\forall \sigma \in \Psi:\; Pr_{T_j}^{\sigma}[\bigcap_{q \in K_L} Reach(T_q)] < 1$ or 
$\forall \sigma \in \Psi:\; Pr_{T_r}^{\sigma}[\bigcap_{q \in K' - K_L} Reach(T_q)] < 1$.
\end{enumerate}

\noindent Let $h_l := T_i(l, T_j)$ and 
$h_r := T_i(r, T_r)$. 
Notice that for any $\sigma \in \Psi$ and for any 
$q' \in K'$, $Pr_{T_i}^{\sigma}[\bigcup_{q \in K'} Reach^{\complement}(T_q)] \ge Pr_{T_i}^{\sigma}[Reach^{\complement}(T_{q'})] = Pr_{h_l}^{\sigma}[Reach^{\complement}(T_{q'})] \cdot Pr_{h_r}^{\sigma}[Reach^{\complement}(T_{q'})]$.

We claim that $\forall \sigma \in \Psi:\; \bigvee_{q \in K'} 
Pr_{T_j}^{\sigma}[Reach^{\complement}(T_q)] \cdot 
Pr_{T_r}^{\sigma}[Reach^{\complement}(T_q)] > 0$. 
But for any $q \in K'$ and for any $\sigma \in \Psi$ 
one can easily construct $\sigma' \in \Psi$ such that 
$Pr_{T_j}^{\sigma}[Reach^{\complement}(T_q)] = 
Pr_{h_l}^{\sigma'}[Reach^{\complement}(T_q)]$ and 
similarly for non-terminal $T_r$. 
So it follows from the claim that 
$\forall \sigma \in \Psi:\; \bigvee_{q \in K'} 
Pr_{h_l}^{\sigma}[Reach^{\complement}(T_q)] \cdot 
Pr_{h_r}^{\sigma}[Reach^{\complement}(T_q)] > 0$ and, 
therefore, it follows that $\forall \sigma \in \Psi:\; 
Pr_{T_i}^{\sigma}[\bigcup_{q \in K'} 
Reach^{\complement}(T_q)] > 0 \Leftrightarrow 
Pr_{T_i}^{\sigma}[\bigcap_{q \in K'} 
Reach(T_q)] < 1$.

Suppose the opposite, i.e., assume $(\mathcal{P})$ such that 
$\exists \sigma' \in \Psi:\; \bigwedge_{q \in K'} 
Pr_{T_j}^{\sigma'}[Reach^{\complement}(T_q)] \cdot 
Pr_{T_r}^{\sigma'}[Reach^{\complement}(T_q)] = 0$. 
Now for any $q \in K'$, by statement 2) above, we know 
that $T_j \not\in F_{\{q\}} \vee T_r \not\in F_{K'_{-q}}$ and 
$T_j \not\in F_{K'_{-q}} \vee T_r \not\in F_{\{q\}}$. 
First, suppose that in fact for some $q' \in K'$ it 
is the case that $T_j \not\in F_{\{q'\}} \wedge T_r \not\in 
F_{\{q'\}}$ (i.e., $T_j \in S_{\{q'\}} \cup Z_{\{q'\}} 
\wedge T_r \in S_{\{q'\}} \cup Z_{\{q'\}}$). 
That is, $\forall \sigma \in \Psi:\; 
Pr_{T_j}^{\sigma}[Reach^{\complement}(T_{q'})] > 0$ 
and $Pr_{T_r}^{\sigma}[Reach^{\complement}(T_{q'})] 
> 0$, where our claims follows directly 
(hence, contradiction to $(\mathcal{P})$). 
Second, suppose that for some $q' \in K'$ it is the case that 
$T_j \not\in F_{K'_{-q'}} \wedge T_r \not\in F_{K'_{-q'}}$ 
(i.e., $T_j \in S_{K'_{-q'}} \cup Z_{K'_{-q'}} \; \wedge \; 
T_r \in S_{K'_{-q'}} \cup Z_{K'_{-q'}}$). 
But then $T_i$ would have been added to set $S_{K'_{-q'}}$ 
at step II.5.(c) when constructing the answer for subset of 
targets $K'_{-q'}$. However, we already know that 
$T_i \in \bigcap_{K'' \subset K'} F_{K''}$ (follows from 
steps II.3 and II.4. that $T_i \not\in 
\bigcup_{K'' \subset K'} (S_{K''} \cup Z_{K''})$). 
Hence, again a contradiction.

Therefore, it follows that for every $q \in K'$, either 
$T_j \not\in F_{\{q\}} \wedge T_j \not\in F_{K'_{-q}}$ or 
$T_r \not\in F_{\{q\}} \wedge T_r \not\in F_{K'_{-q}}$. 
And in particular, the essential part is that 
$\forall q \in K'$, either $T_j \not\in F_{\{q\}}$ or 
$T_r \not\in F_{\{q\}}$. 
That is, for every $q \in K'$, either 
$\forall \sigma \in \Psi: \; 
Pr_{T_j}^{\sigma}[Reach^{\complement}(T_q)] > 0$, or 
$\forall \sigma \in \Psi:\; 
Pr_{T_r}^{\sigma}[Reach^{\complement}(T_q)] > 0$. 
But then, combined with assumption $(\mathcal{P})$, 
it actually follows that there exists a subset 
$K'' \subseteq K'$ such that 
$\exists \sigma' \in \Psi:\; \bigwedge_{q \in K''} 
Pr_{T_r}^{\sigma'}[Reach^{\complement}(T_q)] = 0 \; 
\wedge \; \bigwedge_{q \in K' - K''} Pr_{T_j}^{\sigma'} 
[Reach^{\complement}(T_q)] = 0$. 
And by Proposition \ref{prop:equiv}(1.), it follows that 
$\exists \sigma' \in \Psi:\; 
Pr_{T_r}^{\sigma'}[\bigcap_{q \in K''} Reach(T_q)] = 1 
\wedge Pr_{T_j}^{\sigma'}[\bigcap_{q \in K' - K''} Reach(T_q)] = 1$, i.e., 
$T_j \in F_{K' - K''} \wedge T_r \in F_{K''}$, 
contradicting the known facts 1) and 2). 
Hence, assumption $(\mathcal{P})$ is wrong and our claim 
is satisfied.
\end{enumerate}

Now consider any non-terminal $T_i$ that is added to 
set $S_{K'}$ in step II.10. at some iteration of 
the inner loop (i.e., $T_i \in Y_{K'} := X - (S_{K'} 
\cup F_{K'}) \subseteq \bar{Z}_{K'}$). Since non-terminal $T_i$ has not been 
previously added to sets $D_{K'}$, $S_{K'}$ or $F_{K'}$, 
then all of the following hold:

\begin{enumerate}[label=(\arabic*.)]
	\item $i \not\in K'$;
	
	\item if $T_i$ is of \textsf{L}-form, then a non-terminal 
$T_i$ generates with probability 1 a non-terminal which 
belongs to $Y_{K'}$ (otherwise $T_i$ would have been 
added to sets $S_{K'}$ or $F_{K'}$ in step II.4, II.5. 
or II.9., respectively);
	
	\item if $T_i$ is of \textsf{M}-form, then 
$\forall a \in \Gamma^i:\; T_i \xrightarrow{a} T_d,\; 
T_d \not\in F_{K'} \cup D_{K'}$ 
(otherwise $T_i$ would have been added to 
sets $D_{K'}$ or $F_{K'}$ in step II.2. or step II.9., respectively), and 
$\exists a' \in \Gamma^i:\; T_i \xrightarrow{a'} T_j,\; 
T_j \not\in S_{K'} \cup Z_{K'}$, i.e., $T_j \in Y_{K'}$ 
(otherwise $T_i$ would have been added to set $S_{K'}$ 
in step II.5.); and
	
	\item if $T_i$ is of \textsf{Q}-form 
($T_i \xrightarrow{1} T_j \; T_r$), then w.l.o.g. 
$T_j \in Y_{K'}$ and $T_r \in (Y_{K'} \cup S_{K'} \cup Z_{K'})$ 
(as $T_i$ has not been added to the other sets in 
steps II.2., II.5. or II.9.).
\end{enumerate}

Due to the statements (2.) - (4.) above, notice that the 
dependency graph $G$ does not contain outgoing edges 
from set $Y_{K'}$ to sets $D_{K'}$ and $F_{K'}$. So any SCC in subgraph $G[X - S_{K'}]$, that contains a node from set $Y_{K'}$, is in fact entirely contained in subgraph $G[Y_{K'}]$. 

Furthermore, one of the following is the reason for a 
\textsf{Q}-form non-terminal $T_i \in Y_{K'}$ not 
having been added to set $F_{K'}$ at 
the initialization step II.8.:
\begin{itemize}
	\item[(4.1.)] either $T_i$ does not belong to any of 
the sets $H_q, q \in K'$. So, from step II.7., 
$T_r \in Z_{\{q\}}$ for every $q \in K'$ 
(recall that w.l.o.g. $T_j \in Y_{K'} \subseteq \bar{Z}_{K'} 
\subseteq \bar{Z}_{\{q\}},\; \forall q \in K'$),
	
	\item[(4.2.)] or $T_i$ does belong to some 
set $H_{q'}, {q'} \in K'$, but if $T_i$ belongs to a 
strongly connected component $C'$ in $G[Y_{K'}]$, then 
$\exists q'' \in K'$ such that $H_{q''} \cap C' = \emptyset$.
\end{itemize}

We can treat the \textsf{Q}-form non-terminals 
with property (4.1.) as if they have only one child 
(namely the child of non-terminal $T_j$), since 
the other child (of non-terminal $T_r$) does not 
contribute to reaching, even with a positive probability, 
any of the target non-terminals from set $K'$.

We need to show that for every non-terminal 
$T_i \in Y_{K'}$ property $(B)_{K'}^i$ holds, i.e., 
$\forall \sigma \in \Psi:\; Pr_{T_i}^{\sigma} 
[\bigcap_{q \in K'} Reach(T_q)] < 1$. 

From standard algorithms about SCC-decomposition, it is 
known that there is an ordering of the SCCs in 
$G[Y_{K'}]$, where the bottom level in this ordering 
(level $0$) consists of bottom strongly connected components 
(BSCCs) that have no edges leaving the BSCC at all, 
and for further levels in the ordering of SCCs 
the following is true: SCCs or nodes not in any SCC, 
at level $t \ge 1$, have directed paths out of them 
leading to SCCs or nodes not in any SCC, 
at levels $< t$. We rank the SCCs and 
the independent nodes (not belonging to any SCC) in 
$G[Y_{K'}]$ according to this ordering, denoting 
by $Y_{K'}^t,\; t \ge 0$ the nodes (non-terminals) 
at levels up to and including $t$, and use 
the following induction on the level:
\begin{itemize}
	\item[--] For the base case: for any BSCC, $C$, at 
level $0$ (i.e., $C \subseteq Y_{K'}^0$), clearly 
for any non-terminal $T_i \in C$, 
$\exists \sigma \in \Psi:\; 
Pr_{T_i}^{\sigma}[\bigcap_{q \in K'} Reach(T_q)] = 1$ 
if and only if $H_q \cap C \not= \emptyset,\; \forall q \in K'$.
But, by property (4.2), there is no such component $C$ 
in $G[Y_{K'}]$ that contains a \textsf{Q}-form non-terminal 
from each of the sets $H_q,\; q \in K'$.

	\item[--] As for the inductive step, assume that 
for some $t \ge 1$ for any $T_v \in Y_{K'}^{t-1},\; 
\forall \sigma \in \Psi:\; 
Pr_{T_v}^{\sigma} [\bigcap_{q \in K'} Reach(T_q)] < 1$, 
i.e., $(B)_{K'}^v$ is satisfied. 
Let $\sigma$ be an arbitrary strategy fixed for the player. 
For a SCC, $C'$, at level $t \ge 1$, let $w$ 
denote the path (in the play), where $w$ begins 
at a starting non-terminal $T_i \in C'$ and evolves 
in the following way. If the current copy $o$ on the path $w$ 
is of a \textsf{L}-form or a \textsf{M}-form 
non-terminal $T_j \in C'$, then $w$ follows along 
the unique successor of $o$ in the play. 
And if the current copy $o$ on path $w$ is of 
a \textsf{Q}-form non-terminal $T_j \in C'$ 
($T_j \xrightarrow{1} T_{j'} \; T_r$, where w.l.o.g. 
$T_{j'} \in C'$), then $w$ follows along the child 
of non-terminal $T_{j'}$. 
(Note that $T_r \not\in C'$, 
since we already know from (4.) that 
$T_{j'} \in Y_{K'} \subseteq \bar{Z}_{K'} \subseteq \bar{Z}_{\{q\}},\; 
\forall q \in K'$, and so if $T_r \in C' \subseteq Y_{K'}$ 
then property (4.2.) will be contradicted.) 
If the current copy $o$ on path $w$ is of a non-terminal 
not belonging in $C'$, then the path $w$ terminates. 
Denote by $\square C'$ the event that path $w$ is infinite, 
i.e., all non-terminals observed along path $w$ are in $C'$ 
and path $w$ never leaves $C'$ and never terminates. 
Then for any starting non-terminal $T_i \in C'$: 
	\begin{align*}
		& Pr_{T_i}^{\sigma} \Big[ \bigcap_{q \in K'} 
Reach(T_q) \Big] = Pr_{T_i}^{\sigma} 
\Big[ \Big( \bigcap_{q \in K'} Reach(T_q) \Big) \cap 
\square C' \Big] \\
		& + Pr_{T_i}^{\sigma} \Big[ \Big( \bigcap_{q \in K'} 
Reach(T_q) \Big) \cap \neg \square C' \Big] = 
Pr_{T_i}^{\sigma} \Big[ \Big( \bigcap_{q \in K'} 
Reach(T_q) \Big) \cap \neg \square C' \Big] 
	\end{align*}
Observe that $Pr_{T_i}^{\sigma}[(\bigcap_{q \in K'} 
Reach(T_q)) \cap \square C'] = 0$, due to 
statements (1.) and (4.2.). 

By property (3.) and also due to the ranking of SCCs and 
nodes in $G[Y_{K'}]$, if path $w$ terminates, then it does in 
a non-terminal $T_v \in S_{K'} \cup Z_{K'} \cup Y_{K'}^{t-1}$. 
Also due to properties (1.) - (4.) and (4.2.), 
in the case of event $\neg \square C'$ occurring, 
all the targets in set $K'$ are reached with probability $1$, 
starting in $T_i \in C'$, if and only if they are all reached 
with probability $1$, starting from such a non-terminal 
$T_v \in S_{K'} \cup Z_{K'} \cup Y_{K'}^{t-1}$.
To see this, note that for any of 
the \textsf{Q}-form non-terminals 
$T_j \in C'$ ($T_j \xrightarrow{1} T_{j'} \; T_r$, 
where w.l.o.g. $T_{j'} \in C'$), 
$T_r \not\in F_{\{q\}}$ for any $q \in K'$ (otherwise, 
if $T_r \in F_{\{q'\}}$ for some $q' \in K'$, then 
since $T_j$ was not added to set $D_{K'}$ at step II.1., 
it follows that $T_{j'} \not\in F_{K'_{-q'}}$, i.e., 
$T_{j'} \in S_{K'_{-q'}} \cup Z_{K'_{-q'}}$, and hence 
$T_{j'} \in S_{K'} \cup Z_{K'}$, which contradicts 
that $T_{j'} \in C' \subseteq Y_{K'}$).
So none of the targets in set $K'$ is reached 
with probability $1$ (but it is possible with a positive 
probability) from a non-terminal spawned off 
of the path $w$. 

It follows that for a starting non-terminal $T_i \in C'$:
    \begin{align*}
        & \exists \sigma' \in \Psi:\; Pr_{T_i}^{\sigma'} 
\Big[ \Big( \bigcap_{q \in K'} Reach(T_q) \Big) \cap 
\neg \square C' \Big]  = 1 
	       \hspace*{0.1in} \text{if and only if} \\ 
	    & \exists \sigma'' \in \Psi:\; 
\max_{\langle T_v \in S_{K'} \cup Z_{K'} \cup Y_{K'}^{t-1} \; \mid \;
\exists T_j \in C', b \in \Gamma^j:\; 
T_j \xrightarrow{b} T_v \rangle} 
Pr_{T_v}^{\sigma''} \Big[ \bigcap_{q \in K'} 
Reach(T_q) \Big] = 1
    \end{align*}
The right-hand side of this statement is clearly 
not satisfied since we already know that  
$T_v \in S_{K'} \cup Z_{K'} \cup Y_{K'}^{t-1}$ 
satisfy property $(B)_{K'}^v$. 

So it follows that 
$\forall \sigma' \in \Psi:\; 
Pr_{T_i}^{\sigma'} [\bigcap_{q \in K'} Reach(T_q)] = 
Pr_{T_i}^{\sigma'} 
[(\bigcap_{q \in K'} Reach(T_q)) \cap \neg \square C'] < 1$.

As for nodes (non-terminals) $T_i \in Y_{K'}^t$ at level $t$, 
that do not belong to any SCC, using a similar argument, 
$\forall \sigma \in \Psi:\; 
Pr_{T_i}^{\sigma}[\bigcap_{q \in K'} Reach(T_q)] < 1$. 
\end{itemize}
By this inductive argument, 
it follows that for any non-terminal $T_i \in Y_{K'}$ and 
for any strategy $\sigma \in \Psi$: 
$Pr_{T_i}^{\sigma}[\bigcap_{q \in K'} Reach(T_q)] < 1$.
\end{proof}

\medskip

Now we show that for non-terminals $T_i \in F_{K'}$, 
when the inner loop for subset $K' \subseteq K$ terminates, 
the property $(A)_{K'}^i$ is satisfied. We will also 
construct a witness strategy, under which this property 
holds for each non-terminal $T_i \in F_{K'}$. 
Since we have already proved it for non-terminals 
in set $D_{K'}$, in the following Lemma we refer to 
the part of set $F_{K'}$ not containing set $D_{K'}$, 
i.e., to set $F_{K'} = X - S_{K'}$.

\begin{lemma}
	Every non-terminal $T_i \in F_{K'}$ satisfies 
the property $(A)_{K'}^i$.
	\label{lemma:property_A-AS}
\end{lemma}

\begin{proof}
    For the rest of this proof denote by $F_{K'}^0$ 
the initialized set at step II.8. Let us first observe 
the properties for the non-terminals $T_i \in F_{K'} = X - S_{K'}$. 
None of the non-terminals is a target non-terminal from 
set $K'$, i.e., $i \not\in K'$. 
If $T_i$ is of \textsf{L}-form, then:
	\begin{enumerate}
		\setlength{\itemsep}{0em}
		\item[(L)] a non-terminal $T_i$ generates with 
probability 1 as offspring some non-terminal belonging either 
to set $F_{K'}$ or to set $D_{K'}$.
	\end{enumerate}
If $T_i$ is of \textsf{M}-form, then 
$\forall a \in \Gamma^i: T_i \xrightarrow{a} T_d,\; 
T_d \not\in D_{K'}$ and:
	\begin{enumerate}
		\setlength{\itemsep}{0em}
		\item[(M)] $\exists a^* \in \Gamma^i: 
		    T_i \xrightarrow{a^*} T_j,\; T_j \in F_{K'}$.
	\end{enumerate}
If $T_i$ is of \textsf{Q}-form (i.e., $T_i \xrightarrow{1} T_j \; T_r$), 
then $T_j, T_r \not\in D_{K'}$ and:
	\begin{enumerate}
		\setlength\itemsep{0em}
		\item[(Q.0)] if $T_i \in F_{K'}^0$, 
		$\exists q \in K'$ such that w.l.o.g. $T_j \in F_{K'} 
		\wedge T_r \in \bar{Z}_{\{q\}}$,
				
		\item[(Q.1)] otherwise, w.l.o.g. $T_{j} \in F_{K'}$.
	\end{enumerate}

\medskip

\begin{enumerate}[label=($\mathfrak{P}$)]
	\item Let us recall that for every $q \in K'$, 
there is a deterministic static strategy $\sigma'_{\{q\}}$ 
for the player and a value $b_{\{q\}} > 0$ such that, 
starting at a non-terminal $T_r \in \bar{Z}_{\{q\}}$, 
$Pr_{T_r}^{\sigma'_{\{q\}}}[Reach(T_q)] \ge b_{\{q\}}$. 
Let $b := \min_{q \in K'} \{b_{\{q\}}\} > 0$.
\end{enumerate}

We construct now the non-static strategy $\sigma_{K'}$ 
for the player in the following way. In each generation, 
there is going to be one non-terminal in the generation 
that is declared to be a ``queen" and the rest of 
the non-terminals in the generation are called ``workers" 
(we will see the difference between the two labels, 
especially in the choices of actions). Suppose 
the initial population is a non-terminal $T_v \in F_{K'}$, 
declared to be the initial queen. 

Consider any finite ancestor history $h$, that starts at 
the initial non-terminal $T_v \in F_{K'}$, and let 
$o$ denote the non-terminal copy at the end of the ancestor history 
$h$. If $o$ is a queen of \textsf{L}-form non-terminal $T_i$, then 
from statement (L) we know that in the next generation 
the single generated successor child $o'$ is of some non-terminal 
$T_j \in F_{K'} \cup D_{K'}$. If $T_j \in D_{K'}$, 
then we use at $o'$ and its subtree of descendants 
the randomized non-static witness strategy 
from property $(A)_{K'}^j$ as if 
the play is starting in $o'$. If $T_j \in F_{K'}$, 
then we label $o'$ as the queen in the next generation 
and use the same strategy $\sigma_{K'}$ 
(that is currently being described) at it. 
If $o$ is a queen of \textsf{M}-form non-terminal $T_i$, 
then $\sigma_{K'}$ chooses at $o$ 
uniformly at random among actions $a^*$ from statement (M) and, 
hence, in the next generation a single child $o'$ of 
some non-terminal $T_j \in F_{K'}$ will be generated. 
Again $o'$ is declared to be the queen in the next generation 
and the same strategy $\sigma_{K'}$ (currently being described) 
is used at it. 
If $o$ is a queen of \textsf{Q}-form non-terminal $T_i$ 
(i.e., $T_i \xrightarrow{1} T_j \; T_r$), then there 
are two cases for the two successor children $o'$ and $o''$ 
of non-terminals $T_j$ and $T_r$, respectively:
\begin{itemize}
	\setlength{\itemsep}{0em}
	\item \textbf{either} property (Q.0) is satisfied, 
i.e., $T_i \in F_{K'}^0$, and 
$T_j \in F_{K'} \wedge T_r \in \bar{Z}_{\{q\}}$, for some 
target $q \in K'$. Then, in the next generation, 
we declare $o'$ to be the queen and use the currently 
described strategy $\sigma_{K'}$ for it. 
As for the child $o''$, it is declared to be a worker 
and the strategy used at the entire subtree, 
rooted at $o''$, is some strategy $\sigma'_{\{q'\}}$ 
(from statement ($\mathfrak{P}$)), 
where $q' \in K'$ is chosen uniformly at random among all 
targets $q \in K$ such that $T_r \in \bar{Z}_{\{q\}}$. 
The randomization in the strategy 
of the worker is needed, since non-terminal $T_i$ can belong to 
more than one set $H_q$, i.e., $T_r$ can belong to more than 
one set $\bar{Z}_{\{q\}}$.
	
	\item \textbf{or} property (Q.1) is satisfied, i.e., 
$T_i \in F_{K'} - F_{K'}^0$ and 
w.l.o.g. $T_j \in F_{K'}$. Then, in the next generation, 
the child $o'$ is again declared to be the queen and 
the same strategy $\sigma_{K'}$ is used for it, 
whereas the child $o''$ is again labelled 
as a worker, but the strategy for it is irrelevant and so 
an arbitrary one is chosen for its entire subtree of descendants.
\end{itemize}
That concludes the description of strategy $\sigma_{K'}$. 
Now we need to show that, indeed, the randomized non-static 
strategy $\sigma_{K'}$ is an almost-sure strategy for the player, 
i.e., that for any $T_i \in F_{K'}:\; 
Pr_{T_i}^{\sigma_{K'}}[\bigcap_{q \in K'} Reach(T_q)] = 1$.

As previously stated, $F_{K'}^0$ is the initial set $F_{K'}$ 
at step II.8. Also let $T_{x_1}, T_{x_2}, \ldots, T_{x_t}$ 
be the non-terminals in set $F_{K'} - F_{K'}^0$ 
indexed with respect to the time 
at which they were added to set $F_{K'}$ at step II.9. 
Let $\gamma := \max_{i \in [n]} |\Gamma^i|$ and 
let $\lambda$ be the minimum of $\frac{1}{\gamma}$ and 
the minimum rule probability in the OBMDP. 
Consider the sequence of queens. We claim that with a 
positive probability $\ge \lambda^n$ 
in the next $n = |V|$ generations we reach a 
\textsf{Q}-form queen of a (specific) non-terminal in 
set $F_{K'}^0$. To show this, we define, for each non-terminal 
$T_i \in F_{K'}$, a finite ``auxiliary" tree $\mathcal{T}_i$, 
rooted at $T_i$, which represents why $T_i$ was added 
to set $F_{K'}$ (i.e., based on steps II.8. and II.9. 
in the last iteration before step II.11. terminates 
the inner loop). If $T_i \in F_{K'}^0$, then 
the tree $\mathcal{T}_i$ is constructed of just a 
single node (leaf) labelled by $T_i$. If $T_i$ is of 
\textsf{L}-form, added at step II.9., then 
$T_i \rightarrow T_j,\; T_j \in F_{K'}$ 
(otherwise $T_i$ would have been added to set $D_{K'}$) 
and the tree $\mathcal{T}_i$ has an edge from its root 
(labelled by $T_i$) to a child labelled by $T_j$ 
(the root of the subtree $\mathcal{T}_j$), for each such 
$T_j \in F_{K'}$. If $T_i$ is of \textsf{M}-form, 
added at step II.9., then the tree $\mathcal{T}_i$ 
has an edge from its root (labelled by $T_i$) to a child 
labelled by $T_j$ (the root of the subtree $\mathcal{T}_j$), 
for every $T_j$ such that 
$\exists a^* \in \Gamma^i:\; T_i \xrightarrow{a^*} T_j,\; 
T_j \in F_{K'}$. And if $T_i$ is of \textsf{Q}-form, 
added at step II.9., then the tree $\mathcal{T}_i$ has an 
edge from its root (labelled by $T_i$) to a child 
labelled by $T_j$ (from property (Q.1)), which is 
the root of the subtree $\mathcal{T}_j$.

The ``auxiliary'' tree, just defined, has depth of at most $n$, 
since there is a strict order in which the non-terminals 
entered set $F_{K'}$. Now observe that, if we consider 
any generation of the play, assuming that the current queen 
(in this generation) is of some non-terminal $T_i \in F_{K'}$, 
it can be inductively shown that with a positive probability 
(at least $\lambda^n$) in at most $n$ generations 
the sequence of queens follows a \textit{specific} 
root-to-leaf path in $\mathcal{T}_i$. That is because if 
we are at a queen of a \textsf{L}-form non-terminal $T_j$ 
(respectively, in node labelled by $T_j$, which is 
the root of tree $\mathcal{T}_j$), then in 
the next generation with probability $\ge \lambda$ 
the queen is of non-terminal $T_{j'} \in F_{K'}$, which 
is a child of the root of $\mathcal{T}_j$ and is also 
itself the root of $\mathcal{T}_{j'}$. And if we are 
at a queen of a \textsf{M}-form non-terminal $T_j$, then 
in the next generation (due to the fixed strategy $\sigma_{K'}$) 
with probability $\ge 1 / |\Gamma^j| \ge 1/\gamma \ge \lambda$ 
the successor queen is of a non-terminal $T_{j_a} \in F_{K'}$, 
which is a child of the root of $\mathcal{T}_j$ and is also 
the root of $\mathcal{T}_{j_a}$. And if we are at a queen 
of a \textsf{Q}-form non-terminal $T_j$, which is not a leaf 
in this ``auxiliary" tree, then in the next generation 
with probability $1$ the queen is of a non-terminal $T_{j'}$, 
which is the root of $\mathcal{T}_{j'}$ and the unique child 
of the root of $\mathcal{T}_j$. Since the depth of 
the ``auxiliary" defined tree is at most $n$, then 
with probability $\ge \lambda^n$, from a current queen 
of some non-terminal $T_i \in F_{K'}$, in the next $\le n$ 
steps we arrive at a \textit{specific} leaf $T_v$ of 
the tree $\mathcal{T}_i$, i.e., 
a queen of non-terminal $T_v \in F_{K'}^0$ is generated.

If somewhere along the sequence of queens, a queen of a 
\textsf{L}-form non-terminal happens to generate a 
non-terminal in $D_{K'}$, then the sequence of queens 
is actually finite. Therefore, if the sequence of queens 
is infinite, since it has to follow root-to-leaf paths 
in the defined ``auxiliary" tree, then it follows 
that with probability $1$ infinitely often a queen of 
a \textsf{Q}-form non-terminal in set $F_{K'}^0$ is observed.

Now consider any $q \in K'$ and any \textsf{Q}-form 
non-terminal $T_u \in F_{K'}^0 \cap H_q$. Since in the 
subgraph of the dependency graph, induced by 
$X - S_{K'} = F_{K'}$ (i.e., $G[F_{K'}]$), node $T_u$ 
is part of a SCC that contains at least one node 
(non-terminal) from each set $H_{q'}, q' \in K'$, then, 
along the sequence of queens, from a queen of non-terminal 
$T_u$, for any $q' \in K'$ there is a non-terminal 
$T_{u'} \in F_{K'}^0 \cap H_{q'}$ that can be reached 
as a queen, under the described strategy $\sigma_{K'}$, 
in at most $n$ generations with a positive probability 
bounded away from zero (in fact, at least $\lambda^n$). 
Note: There is a positive probability, under 
strategy $\sigma_{K'}$, to exit the particular SCC of 
$T_u$. However, %
under $\sigma_{K'}$ and starting at any non-terminal 
$T_i \in F_{K'}$, almost-surely the sequence of queens 
eventually reaches a queen whose non-terminal is in a 
SCC, $C''$, in $G[F_{K'}]$ which can only have an outgoing edge 
to set $D_{K'}$ and where, moreover, for each target in $K'$ 
there is a branching (\textsf{Q}-form) node in $C''$ whose 
``extra'' child can hit that target with a positive probability 
(bounded away from zero).

Hence, starting at a non-terminal $T_i \in F_{K'}$ and 
under strategy $\sigma_{K'}$, the sequence of queens 
follows root-to-leaf paths in the defined ``auxiliary" tree 
and, for each $q \in K'$, infinitely often a queen 
of a \textsf{Q}-form non-terminal from set $H_{q}$ 
is observed. And each 
such queen generates an independent worker, that reaches 
the respective target non-terminal $T_q$ in a finite number 
of generations with a positive probability bounded away 
from zero (due to the uniformly at random choice over 
strategies from statement ($\mathfrak{P}$), for each worker, 
and due to the fact that the value $b > 0$ from statement 
($\mathfrak{P}$) does not depend on history or time when 
the worker is generated). 
And, more importantly, since the 
queens of \textsf{Q}-form non-terminals from the sets 
$H_q, q \in K'$ form SCCs in $G[F_{K'}]$, then 
collectively the independent workers 
(under their respective strategies) have infinitely often 
a positive probability bounded away from zero to reach all 
target non-terminals from set $K'$ in a finite number of 
generations. Hence, all target non-terminals from set $K'$ 
are reached with probability $1$.
\end{proof}

\medskip
This completes the proof of Theorem \ref{theorem:QualMoR-AS-Reach} 
and the analysis of the almost-sure algorithm. 
The proof of Lemma \ref{lemma:property_A-AS} describes how to construct, 
for any subset $K' \subseteq K$, 
the witness strategy $\sigma_{K'}$ for the non-terminals 
in set $F_{K'}$. 
These non-static strategies $\sigma_{K'}$ are described 
as functions that map finite ancestor histories belonging 
to the controller to distributions over actions for the current 
non-terminal of the ancestor history, and 
can be described in such a form 
in time $4^k \cdot |\mathcal{A}|^{O(1)}$.
\end{proof}

\section{Further cases of qualitative multi-objective (non-)reachability}

In this section we present algorithms for deciding some other cases 
of qualitative multi-objective problems for 
OBMDPs,
involving certain
kinds of boolean combinations of qualitative 
reachability and non-reachability queries with respect to 
given target non-terminals.

\subsection{$\stackrel{?}{\exists} \sigma \in \Psi:\; 
\bigwedge_{q \in K} Pr_{T_i}^{\sigma}[Reach(T_q)] < 1$}

\vspace*{0.1in}

\begin{proposition}
There is an algorithm that, given an OBMDP, $\mathcal{A}$, and a set $K \subseteq [n]$ of 
$k = |K|$ target non-terminals, 
computes the set 
$F := \{T_i \in V \mid \exists \sigma \in \Psi:\; 
\bigwedge_{q \in K} Pr_{T_i}^{\sigma}[Reach(T_q)] < 1\}$. 
The algorithm runs in time $k \cdot |\mathcal{A}|^{O(1)}$ and 
can also compute a randomized static witness strategy $\sigma$ 
for the non-terminals in set $F$.
	\label{prop:conjunction_of_<1_reach}
\end{proposition}

\begin{proof}
	First, as a preprocessing step, for each $q \in K$ we 
compute the set $W_q := \{T_i \in V \mid \exists \sigma_q \in \Psi:\; 
Pr_{T_i}^{\sigma_q}[Reach(T_q)] < 1\}$, together with a 
single deterministic static strategy $\sigma_q$ that witnesses 
the property for every non-terminal in set $W_q$. 
This can be done in time 
$k \cdot |\mathcal{A}|^{O(1)}$, using the algorithm from 
\cite[Theorem 9.3]{ESY-icalp15-IC}.

Then the Proposition is a direct consequence from the 
following Claim.

\begin{claim}
	$F = \bigcap_{q \in K} W_q$.
\end{claim}

\begin{proof}
	We need to show that $T_i \in \bigcap_{q \in K} W_q$ if 
and only if $\exists \sigma' \in \Psi:\; \bigwedge_{q \in K} 
Pr_{T_i}^{\sigma'}[Reach(T_q)] < 1$.

\noindent ($\Leftarrow$.) Suppose $T_i \not\in \bigcap_{q \in K} W_q$, 
i.e., $T_i \in \bigcup_{q \in K} \overline{W}_q$, where 
$\overline{W}_q := V - W_q$ for each $q \in K$. 
Then there exists some $q' \in K$ such that $T_i \in \overline{W}_{q'}$, 
i.e., $\forall \sigma \in \Psi:\; Pr_{T_i}^{\sigma}[Reach(T_{q'})] 
= 1$. Clearly, this implies that $\forall \sigma \in \Psi:\; 
\bigvee_{q \in K} Pr_{T_i}^{\sigma}[Reach(T_q)] = 1$.

\medskip

\noindent ($\Rightarrow$.) Suppose that $T_i \in \bigcap_{q \in K} W_q$. 
Recall that for each $q \in K$ 
there is a deterministic static witness strategy $\sigma_q$ for 
the non-terminals in set $W_q$. Let $\sigma'$ be a randomized static 
strategy for the player defined as follows: 
$\sigma'$ chooses uniformly at random a target $q \in K$ and 
copies exactly the deterministic static strategy $\sigma_q$. 
Then, for each target $T_q,\; q \in K$, under $\sigma'$ and 
starting at a non-terminal $T_i \in \bigcap_{q \in K} W_q$:
\begin{align*}
	& Pr_{T_i}^{\sigma'}[Reach(T_q)] = \sum_{q' \in K} \frac{1}{k} \cdot Pr_{T_i}^{\sigma_{q'}}[Reach(T_q)] = \\
	& \frac{1}{k} \cdot Pr_{T_i}^{\sigma_q}[Reach(T_q)] + \frac{1}{k} \sum_{q' \in K, q' \not= q} Pr_{T_i}^{\sigma_{q'}}[Reach(T_q)] < \frac{1}{k} + \frac{k - 1}{k} = 1 
\end{align*}
\end{proof}

The randomized static witness strategy $\sigma$ for 
the non-terminals in set $F$ is precisely the $\sigma'$ 
constructed in the proof of the Claim above.
\end{proof}

\subsection{$\stackrel{?}{\exists} \sigma \in \Psi:\; 
Pr_{T_i}^{\sigma}[\bigcap_{q \in K} Reach(T_q)] < 1$}

\vspace*{0.1in}

\begin{proposition}
There is an algorithm that,
given an OBMDP, $\mathcal{A}$, and a set $K \subseteq [n]$ of 
$k = |K|$ target non-terminals, computes the set 
$F := \{T_i \in V \mid \exists \sigma \in \Psi:\; 
Pr_{T_i}^{\sigma}[\bigcap_{q \in K} Reach(T_q)] < 1\}$. 
The algorithm runs in time $k \cdot |\mathcal{A}|^{O(1)}$ and can 
also compute a deterministic static witness strategy $\sigma$ 
for a given starting non-terminal $T_i \in F$.
	\label{prop:interesection_of_reach_<1}
\end{proposition}

\begin{proof}
	First, as a preprocessing step, for each $q \in K$ we 
compute the set $W_q := \{T_i \in V \mid \exists \sigma_q \in 
\Psi:\; Pr_{T_i}^{\sigma_q}[Reach(T_q)] < 1\}$, together with a 
single deterministic static strategy $\sigma_q$ that witnesses 
the property for every non-terminal in set $W_q$. This can be done 
in time $k \cdot |\mathcal{A}|^{O(1)}$, using the algorithm 
from \cite[Theorem 9.3]{ESY-icalp15-IC}.

Then the Proposition is a direct consequence from the claim 
that $F = \bigcup_{q \in K} W_q$. To see this claim, note that 
$T_i \in \bigcup_{q \in K} W_q$ if and only if there exists 
$\sigma' \in \Psi$ and some $q \in K$ such that 
$Pr_{T_i}^{\sigma'}[Reach(T_q)] < 1$ (by definition 
of the $W_q,\; q \in K$ sets). Then the claim 
follows directly from Proposition \ref{prop:equiv}(2.).

For each $T_i \in F$, the witness strategy $\sigma$ selects 
deterministically some $q \in K$, such that $T_i \in W_q$, 
and copies exactly the deterministic static strategy $\sigma_q$.
\end{proof}

\medskip
Consider the following two examples of OBMDPs with non-terminals 
$\{M, T, T', L, R_1, R_2\}$ and target non-terminals $R_1$ and $R_2$. 
$M$ is the only controlled non-terminal. The examples provide 
a good idea of the difference between the objectives in Propositions 
\ref{prop:conjunction_of_<1_reach} and \ref{prop:interesection_of_reach_<1}.

\vspace*{0.1in}

\noindent \textbf{Example 2}
\begin{align*}
	& M \xrightarrow{a} T  && T \xrightarrow{1} L \; R_1 && L \xrightarrow{1/2} \varnothing\\
	& M \xrightarrow{b} T' && T' \xrightarrow{1} R_1 \; R_2 && L \xrightarrow{1/2} R_2
\end{align*}

There exists a deterministic static witness strategy $\sigma'$ 
for the player such that 
$Pr_{M}^{\sigma'}[Reach(R_1) \linebreak \cap Reach(R_2)] < 1$, 
namely, starting at a non-terminal $M$, let the player choose 
deterministically action $a$. Thus, the probability of observing 
both target non-terminals in the generated tree is $1/2$. 
However, notice that for any strategy $\sigma$, 
starting at non-terminal $M$, 
target non-terminal $R_1$ is reached with 
probability $1$. That is, $\forall \sigma \in \Psi:\; 
\bigvee_{q \in \{1, 2\}} Pr_{M}^{\sigma}[Reach(R_q)] = 1$.

\medskip \medskip
\noindent \textbf{Example 3}
\begin{align*}
	& M \xrightarrow{a} T  && T \xrightarrow{1} L \; R_1 && L \xrightarrow{1/2} R_1\\
	& M \xrightarrow{b} T' && T' \xrightarrow{1} L \; R_2 && L \xrightarrow{1/2} R_2
\end{align*}

There exists a static strategy $\sigma'$ such that 
$\bigwedge_{q \in \{1, 2\}} Pr_{M}^{\sigma'}[Reach(R_q)] < 1$, 
but the strategy needs to randomize, otherwise a deterministic choice 
in non-terminal $M$ will generate a target non-terminal immediately 
in the next generation. Note that the same strategy $\sigma'$ 
(although a deterministic one suffices) 
also guarantees $Pr_M^{\sigma'}[Reach(R_1) \cap Reach(R_2)] < 1$.

\subsection{$\stackrel{?}{\exists} \sigma \in \Psi:\; 
\bigwedge_{q \in K} Pr_{T_i}^{\sigma}[Reach(T_q)] > 0$}

\vspace*{0.1in}

\begin{proposition}
There is an algorithm that,
given an OBMDP, $\mathcal{A}$, and a set $K \subseteq [n]$ of 
$k = |K|$ target non-terminals, computes the set 
$F := \{T_i \in V \mid \exists \sigma \in \Psi:\; 
\bigwedge_{q \in K} Pr_{T_i}^{\sigma}[Reach(T_q)] > 0\}$. 
The algorithm runs in time $O(k \cdot |V|^2)$ and can 
also compute a randomized static witness strategy $\sigma$ 
for the non-terminals in set $F$.
	\label{prop:conjunction_of_>0_reach}
\end{proposition}

\begin{proof}
	First, for each $q \in K$, we compute the attractor set 
of target non-terminal $T_q$ with respect to the dependency graph 
$G = (U, E)$, $U = V$, of $\mathcal{A}$. 
That is, for each $q \in K$, we compute the 
set $Attr(T_q)$ as the limit of the following sequence 
$(Attr_t(T_q))_{t \ge 0}$:
\begin{align*}
	& Attr_0(T_q) =  \{T_q\} \\
	& Attr_t(T_q) = Attr_{t-1}(T_q) \cup \{T_i \in V \mid \exists \; 
T_j \in Attr_{t-1}(T_q) \text{ s.t. } (T_i, T_j) \in E\}
\end{align*}
In other words, $Attr(T_q)$ is the set of nodes in $G$ 
(or equivalently, non-terminals in $\mathcal{A}$) that 
have a directed path to the target node (non-terminal) $T_q$ 
in the dependency graph $G$. For each $q \in K$, 
such a set can be computed in time $O(|V|^2)$. 
So all $k$ attractor sets (one for each target non-terminal 
$T_q, q \in K$) can be computed in time $O(k \cdot |V|^2)$. 
The Proposition is a direct consequence from the following Claim.

\begin{claim}
	$F = \bigcap_{q \in K} Attr(T_q)$.
\end{claim}

\begin{proof}
    To prove the Claim, we need to show that
$T_i \in \bigcap_{q \in K} Attr(T_q)$ if and only if 
$\exists \sigma' \in \Psi:\; \bigwedge_{q \in K} 
Pr_{T_i}^{\sigma'}[Reach(T_q)] > 0$.

\noindent ($\Leftarrow.$) Suppose that $T_i \not\in 
\bigcap_{q \in K} Attr(T_q)$, i.e., there exists some $q' \in K$ 
such that $T_i \not\in Attr(T_{q'})$. This implies that 
in the dependency graph $G$ there is even no path from 
$T_i$ to $T_{q'}$. Therefore, regardless of strategy $\sigma$ 
for the player, $Pr_{T_i}^{\sigma}[Reach(T_{q'})] = 0$ 
and hence, $\forall \sigma \in \Psi:\; \bigvee_{q \in K} 
Pr_{T_i}^{\sigma}[Reach(T_q)] = 0$.

\medskip
\noindent ($\Rightarrow.$) Suppose that $T_i \in 
\bigcap_{q \in K} Attr(T_q)$. Let $\sigma'$ be 
the randomized static strategy such that in every non-terminal 
$T_j \in V$ of \textsf{M}-form it chooses uniformly at random 
an action among its set of actions $\Gamma^j$. For each $q \in K$, 
in the dependency graph $G$ there is a directed path from 
$T_i$ to $T_q$. Then under the described strategy $\sigma'$, 
starting at a non-terminal $T_i$, there is a positive probability 
to generate any of the target non-terminals $\{T_q \mid q \in K\}$, 
because there is a positive probability for a path in the play (tree) 
to follow the directed path in $G$ from $T_i$ to $T_q$, 
for any $q \in K$.

Denote by $\lambda$ the minimum of 
$\frac{1}{\max_{j \in [n]} |\Gamma^j|}$ 
and the minimum probability among the probabilistic rules of 
$\mathcal{A}$. Then, in fact, for each $q \in K$, 
under $\sigma'$ there is a probability $\ge \lambda^n$ 
to generate a copy of target non-terminal $T_q$ 
in the next $\le n$ generations, i.e., 
$\bigwedge_{q \in K} Pr_{T_i}^{\sigma'}[Reach(T_q)] 
\ge \lambda^n > 0$.
\end{proof}

The randomized static witness strategy $\sigma$ for the non-terminals 
in set $F$ is the $\sigma'$ constructed in the proof of the Claim above.
\end{proof}

\subsection{$\stackrel{?}{\exists} \sigma \in \Psi:\; 
Pr_{T_i}^{\sigma}[\bigcap_{q \in K} Reach^{\complement}(T_q)] 
\triangle \{0, 1\}$} 

\vspace*{0.1in}

Now let us consider the qualitative cases of multi-objective 
reachability where for a given OBMDP and a given set 
$K \subseteq [n]$ of target non-terminals, the aim is 
to compute those non-terminals $T_i \in V$ that satisfy the 
property that $\exists \sigma \in \Psi:\; Pr_{T_i}^{\sigma} 
[\bigcap_{q \in K} Reach^{\complement}(T_q)] \triangle \{0, 1\}$, 
where $\triangle := \{<, = , >\}$.

First, due to the fact that the complement of 
the set (of plays) $\bigcap_{q \in K} Reach^{\complement}(T_q)$ 
is the set (of plays) $\bigcup_{q \in K} Reach(T_q)$, 
we give the following Lemma to show that 
this complement objective reduces to the objective 
of reachability of a single target non-terminal 
in a slightly modified OBMDP.

\begin{lemma}
There is an algorithm
that, given an OBMDP, $\mathcal{A}$, and a set $K \subseteq [n]$ of 
$k = |K|$ target non-terminals $\{T_q \in V^{\mathcal{A}} 
\mid q \in K\}$, 
runs in linear time $O(|\mathcal{A}|)$ and outputs another OBMDP, $\mathcal{A'}$, with 
a single target non-terminal $T_f$, 
such that for 
any $T_i \in V^{\mathcal{A}} - \{T_q \in V^{\mathcal{A}} \mid 
q \in K\} = V^{\mathcal{A'}} - \{T_f\}$ and any strategy 
$\sigma \in \Psi^{\mathcal{A}}$, 
there exists a strategy $\sigma' \in \Psi^{\mathcal{A'}}$ 
such that 
$Pr_{T_i}^{\sigma, \mathcal{A}} 
[\bigcup_{q \in K} Reach(T_q)] = 
Pr_{T_i}^{\sigma', \mathcal{A'}}[Reach(T_f)]$. 
	\label{lemma:union_of_reachability}
\end{lemma}

\begin{proof}
	Consider the OBMDP, $\mathcal{A'}$, obtained from 
OBMDP, $\mathcal{A}$, by adding a new 
purely probabilistic target non-terminal $T_f$ with 
a single rule $T_f \xrightarrow{1} \varnothing$, removing all 
target non-terminals $\{T_q \in V^{\mathcal{A}} \mid q \in K\}$ 
and their associated rules, 
and replacing any occurrence of a non-terminal 
$T_q \in V^{\mathcal{A}}$, $q\in K$, on the right-hand side 
of some rule with non-terminal $T_f$. 
Hence, $V^{\mathcal{A'}} = (V^{\mathcal{A}} \cup 
\{T_f\}) - \{T_q \in V^{\mathcal{A}} \mid q \in K\}$. Clearly, for any 
$T_q \in V^{\mathcal{A}}$, with $q \in K$ and for any 
$\sigma \in \Psi^{\mathcal{A}}$, 
$Pr_{T_q}^{\sigma, \mathcal{A}}[\bigcup_{q' \in K} 
Reach(T_{q'})] = 1$. Also, for $T_f \in V^{\mathcal{A'}}$ and 
for any $\sigma' \in \Psi^{\mathcal{A'}}$, 
$Pr_{T_f}^{\sigma', \mathcal{A'}}[Reach(T_f)] = 1$.

Observe that for any play (tree) $\mathcal{T}$ in $\mathcal{A}$, 
there is a play $\mathcal{T'}$ in $\mathcal{A'}$ such that any 
copy $o$ of a non-terminal $T_q \in V^{\mathcal{A}}, q \in K$ 
in $\mathcal{T}$ is replaced in $\mathcal{T'}$ by a copy 
of non-terminal $T_f$ and the subtree of descendants of $o$ 
is non-existent in $\mathcal{T'}$. 

Now consider any starting non-terminal 
$T_u \in V^{\mathcal{A}} - \{T_q \in V^{\mathcal{A}} \mid q \in K\} 
= V^{\mathcal{A'}} - \{T_f\}$.

Let $\sigma \in \Psi^{\mathcal{A}}$ be any strategy for 
the player in $\mathcal{A}$. Define strategy 
$\sigma' \in \Psi^{\mathcal{A'}}$ in $\mathcal{A'}$ in 
the following way: for each non-terminal 
$T_i \in V^{\mathcal{A'}} - \{T_f\}$, strategy 
$\sigma'$ behaves exactly like $\sigma$ for all 
ancestor histories ending in $T_i$, and for 
non-terminal $T_f$ strategy $\sigma'$ acts arbitrarily 
in all ancestor histories ending in $T_f$ since it is irrelevant. 
Note that, due to the construction of $\mathcal{A'}$ 
and $\sigma'$, if a play (tree) $\mathcal{T}$, generated under 
strategy $\sigma$, belongs to set (objective) 
$\bigcup_{q \in K} Reach(T_q)$ in $\mathcal{A}$, then 
in $\mathcal{A'}$ under $\sigma'$ the corresponding unique 
play $\mathcal{T'}$ (as described above) belongs to set (objective) 
$Reach(T_f)$. Furthermore, all plays $\mathcal{T}$ in 
$\mathcal{A}$ with the same corresponding play $\mathcal{T'}$ 
in $\mathcal{A'}$ have a combined probability, of being 
generated under $\sigma$, equal to the probability of 
$\mathcal{T'}$ being generated under $\sigma'$ in $\mathcal{A'}$. 
Hence, $Pr_{T_u}^{\sigma, \mathcal{A}}[\bigcup_{q \in K} 
Reach(T_q)] = Pr_{T_u}^{\sigma', \mathcal{A'}}[Reach(T_f)]$. 
But $\sigma$ was an arbitrary strategy.

For the opposite direction, let 
$\sigma' \in \Psi^{\mathcal{A'}}$ be any strategy for 
the player in $\mathcal{A'}$. Define 
$\sigma \in \Psi^{\mathcal{A}}$ to be the strategy that, 
for all non-terminals  
$T_i \in V^{\mathcal{A}} - \{T_q \in V^{\mathcal{A}} \mid q \in K\}$, 
acts the same as $\sigma'$ in all ancestor histories 
ending in $T_i$; and for all non-terminals 
$T_q \in V^{\mathcal{A}},\; q \in K$ 
the strategy $\sigma$ acts arbitrarily in 
all ancestor histories ending in $T_q$ 
as it is irrelevant. Then, for any play 
$\mathcal{T'} \in Reach(T_f)$ in $\mathcal{A'}$ under 
strategy $\sigma'$, there is at least one play 
$\mathcal{T} \in \bigcup_{q \in K} Reach(T_q)$ in 
$\mathcal{A}$ under strategy $\sigma$, such that for any 
copy of non-terminal $T_f$ in tree $\mathcal{T'}$ there is a copy 
of some non-terminal $T_q \in V^{\mathcal{A}}, q \in K$ at 
the corresponding position in tree $\mathcal{T}$. But note that 
the probability of generating $\mathcal{T'}$ in $\mathcal{A'}$ 
under $\sigma'$ is equal to the sum of probabilities of generating 
all such corresponding plays $\mathcal{T}$ in $\mathcal{A}$ 
under $\sigma$. Hence, 
$Pr_{T_u}^{\sigma', \mathcal{A'}}[Reach(T_f)] = 
Pr_{T_u}^{\sigma, \mathcal{A}}[\bigcup_{q \in K} Reach(T_q)]$. 
But $\sigma'$ was an arbitrary strategy.
\end{proof}

We now present a Proposition that deals with 
all four qualitative questions for the (set of plays) 
objective $\bigcap_{q \in K} Reach^{\complement}(T_q)$ 
for a given set $K \subseteq [n]$ of target non-terminals.

\begin{proposition}
There is a P-time algorithm that,
given an OBMDP, $\mathcal{A}$, and a set $K \subseteq [n]$ of 
$k = |K|$ target non-terminals, computes the set $F := \{T_i \in V \mid 
\exists \sigma \in \Psi:\; 
Pr_{T_i}^{\sigma}[\bigcap_{q \in K} 
Reach^{\complement}(T_q)] \triangle \linebreak \{0, 1\}\}$, 
where $\triangle := \{<, = , >\}$. 
The algorithm can also compute a deterministic witness strategy 
$\sigma$ for the non-terminals in set $F$.
	\label{prop:intersetion_of_nonreachability}
\end{proposition}

\begin{proof}
    We can rephrase the question of whether 
$\exists \sigma \in \Psi^{\mathcal{A}}:\; 
Pr_{T_i}^{\sigma, \mathcal{A}}[\bigcap_{q \in K} 
Reach^{\complement}(T_q)] \triangle x$ accordingly into 
the form of asking whether $\exists \sigma \in \Psi^{\mathcal{A}}:\; 
Pr_{T_i}^{\sigma, \mathcal{A}}[\bigcup_{q \in K} Reach(T_q)] 
\triangle_{\complement} 1 - x$, where $x \in \{0, 1\}$ and 
$\triangle_{\complement}$ is $<, = , >$ if 
$\triangle$ is $ >, =, <$, respectively. And as a 
consequence of Lemma \ref{lemma:union_of_reachability}, 
there exists a modified OBMDP, $\mathcal{A'}$, 
with a single target non-terminal $T_f$ such that 
$\exists \sigma \in \Psi^{\mathcal{A}}:\; 
Pr_{T_i}^{\sigma, \mathcal{A}}[\bigcup_{q \in K} Reach(T_q)] 
\triangle_{\complement} 1 - x$ if and only if 
$\exists \sigma' \in \Psi^{\mathcal{A'}}:\; 
Pr_{T_i}^{\sigma', \mathcal{A'}}[Reach(T_f)] 
\triangle_{\complement} 1 - x$.

For the case of $1 - x = 0$, by 
\cite[Proposition 4.1]{ESY-icalp15-IC}, there is a 
P-time algorithm to compute the set $F^{\mathcal{A'}}$ 
of non-terminals $T_i$ in $\mathcal{A'}$ and a deterministic static 
witness strategy $\sigma' \in \Psi^{\mathcal{A'}}$ such that 
$T_i \in F^{\mathcal{A'}}$ are precisely the non-terminals 
that satisfy the property 
$Pr_{T_i}^{\sigma', \mathcal{A'}}[Reach(T_f)] 
\triangle_{\complement} 0$.

For the case of $1 - x = 1$ and $\triangle_{\complement}$ 
equal to $<$ (respectively, $=$), by 
\cite[Theorem 9.3, 9.4]{ESY-icalp15-IC}, there is again a 
P-time algorithm to compute the set $F^{\mathcal{A'}}$ 
of non-terminals $T_i$ in $\mathcal{A'}$ and a 
deterministic static (respectively, non-static) witness strategy 
$\sigma' \in \Psi^{\mathcal{A'}}$ such that 
$T_i \in F^{\mathcal{A'}}$ are the non-terminals that satisfy 
the property $Pr_{T_i}^{\sigma', \mathcal{A'}}[Reach(T_f)] 
< 1$ (respectively, $Pr_{T_i}^{\sigma', \mathcal{A'}} 
[Reach(T_f)] = 1$).

Now for the qualitative decision questions where 
tuple $(\triangle_{\complement}, 1 - x)$ is equal to 
$(=, 0)$ or $(<, 1)$, let $F = F^{\mathcal{A}} := F^{\mathcal{A'}}$; 
and where tuple $(\triangle_{\complement}, 1 - x)$ is equal 
to $(>, 0)$ or $(=, 1)$, let 
$F = F^{\mathcal{A}} := (F^{\mathcal{A'}} - \{T_f\}) \cup 
\{T_q \in V^{\mathcal{A}} \mid q \in K\}$. By the proof of 
Lemma \ref{lemma:union_of_reachability}, from a deterministic 
witness strategy $\sigma' \in \Psi^{\mathcal{A'}}$ for the 
starting non-terminals from set $F^{\mathcal{A'}}$ we can 
obtain a corresponding deterministic (non-)static witness strategy 
$\sigma \in \Psi^{\mathcal{A}}$ for the starting non-terminals from 
set $F - \{T_q \in V^{\mathcal{A}} \mid q \in K\}$. As for each 
non-terminal $T_q \in V^{\mathcal{A}}, q \in K$, let strategy 
$\sigma$ make deterministically and statically an arbitrary choice 
of action from the action set $\Gamma^q$ (in the case if 
$T_q$ is of \textsf{M}-form), since if 
$T_q \not\in F$ then strategy is irrelevant at $T_q$ and 
if $T_q \in F$ then the property holds for any choice of 
the strategy in $T_q$.
\end{proof}

\subsection{$\stackrel{?}{\exists} \sigma \in \Psi:\; 
\bigwedge_{q \in K} Pr_{T_i}^{\sigma}[Reach(T_q)] = 0$}

\vspace*{0.1in}

\begin{proposition}
There is a P-time algorithm that,
given an OBMDP, $\mathcal{A}$, and a set $K \subseteq [n]$ of 
$k = |K|$ target non-terminals, 
computes the set 
$F := \{T_i \in V \mid \exists \sigma \in \Psi:\; 
\bigwedge_{q \in K} Pr_{T_i}^{\sigma}[Reach(T_q)] = 0\}$. 
The algorithm can also compute a deterministic static witness 
strategy $\sigma$ for the non-terminals in set $F$.
	\label{prop:conjunction_of_=0_reach}
\end{proposition}

\begin{proof}
	Note that the question of deciding whether there 
exists a strategy $\sigma \in \Psi$ for the player such that 
$\bigwedge_{q \in K} Pr_{T_i}^{\sigma}[Reach(T_q)] = 0$ can be 
rephrased as asking whether there exists a strategy 
$\sigma \in \Psi$ such that $\bigwedge_{q \in K} 
Pr_{T_i}^{\sigma}[Reach^{\complement}(T_q)] = 1$. 
By Proposition \ref{prop:equiv}(1.), we already know that 
it is equivalent to ask instead whether there exists 
a strategy $\sigma \in \Psi$ such that 
$Pr_{T_i}^{\sigma}[\bigcap_{q \in K} 
Reach^{\complement}(T_q)] = 1$. 
Hence, $F = \{T_i \in V \mid \exists \sigma \in \Psi:\; 
Pr_{T_i}^{\sigma}[\bigcap_{q \in K} Reach^{\complement}(T_q)] 
= 1\}$. And by Proposition \ref{prop:intersetion_of_nonreachability}, 
there is a P-time procedure to compute the set $F$ and to compute 
a deterministic static witness strategy $\sigma$ for 
the non-terminals in set $F$.
\end{proof}

We leave open the decidability of  general boolean combinations of arbitrary qualitative reachability and non-reachability queries.

\end{document}